\documentclass{article}
\usepackage{natbib}
\setcitestyle{authoryear,open={(},close={)}}

\usepackage{comment}
\usepackage{xcolor}
\usepackage{tikz}
\usepackage{graphics}
\usepackage{graphicx}
\usepackage{color}
\usepackage{amsmath}
\usepackage{amsthm}
\usepackage{amsfonts}
\usepackage{geometry}
\usepackage[normalem]{ulem}
\usepackage{float}
\usepackage{enumitem}

\usepackage{mathabx} 

\usepackage{hyperref}
\hypersetup{colorlinks=true,linkcolor=blue,citecolor=blue}
\usepackage{dsfont} 
\usepackage{nicefrac}

\title{The geometry of consumer preference aggregation\footnote{The paper has benefited from discussions with our colleagues (in alphabetic order) Michael Amior, Alexander W. Bloedel, Anna Bogomolnaia, Peter Caradonna, Tommaso Denti, Bram De Rock, Federico Echenique, Cecile Gaubert, David Genesove, Ben Golub, Aram Grigoryan, Yoram Halevy, Eran Hoffmann,  Matthew O. Jackson, Deniz Kattwinkel,
 Peter Klibanoff, Ilan Kremer, Alexey Kushnir, Arthur Lewbel, Kiminori Matsuyama, John Moore, Herv\'e Moulin, Ichiro Obara, Luciano Pomatto, Antonio Rangel, Joseph Root, Jean-Laurent Rosenthal, Ariel Rubinstein, Kota Saito, Itay Saporta-Eksten, Ilya Segal, Roberto Serrano, Denis Shishkin, Ali Shourideh, Joel Sobel, 
Ran Spiegler, Charles Sprenger, Ross Starr, Omer Tamuz, Ina A. Taneva, Laura Taylor, William Thomson, Alexis Akira Toda, Aleh Tsyvinski, Quitz\'e Valenzuela-Stookey, Richard Van Weelden, Shoshana Vasserman, Xavier Vives, Adam Wierman, Kai Hao Yang, Leeat Yariv, and William Zame, and comments of seminar audiences at Berkeley, Caltech, CMU, Edinburgh, HUJI, LA Theory Workshop, Northwestern, Pitt, Princeton, San~Diego, TAU, Technion, Toronto, and UCL.
}}

\author{
Fedor Sandomirskiy\thanks{Princeton University.
\ifdefined\JPE
Email:  \href{mailto:fsandomi@princeton.edu}{fsandomi@princeton.edu}. 
\fi
Fedor thanks the Linde Institute at Caltech, PIMCO Fellows Program, and National Science Foundation (grant
CNS 1518941) for their support.}
\\ 
Philip Ushchev\thanks{ECARES, Universit\'e libre de Bruxelles, and CEPR. 
\ifdefined\JPE
Email: \href{mailto:filipp.ushchev@ulb.be}{filipp.ushchev@ulb.be}. 
\fi
}
}

\date{}

\ifdefined\DRAFT

\definecolor{ForestGreen}{rgb}{.13,.54,.13}
\definecolor{violet}{cmyk}{0.79,0.88,0,0}
\newcommand{\fed}[1]{{\color{ForestGreen}{(\textbf{Fedor:} #1)}}}
\newcommand{\ph}[1]{{\color{red}{(\textbf{Philip:} #1)}}}
\newcommand{\ed}[1]{{{\color{red} {#1}}}}

\else
\newcommand{\fed}[1]{}
\newcommand{\ph}[1]{}
\newcommand{\ed}[1]{#1}

\fi

\newtheorem{theorem}{Theorem}
\newtheorem*{theorem*}{Theorem}

\newtheorem{lemma}{Lemma}

\newtheorem{corollary}{Corollary}

\newtheorem{proposition}{Proposition}

\newtheorem{definition}{Definition}

\theoremstyle{definition}

\theoremstyle{remark}

\newtheorem{example}{Example}

\renewcommand{\vec}[1]{\mathbf{#1}}

\newcommand{\E}{\mathbb{E}}

\newcommand{\R}{\mathbb{R}}

\newcommand{\cav}{\mathrm{cav}}
\newcommand{\vex}{\mathrm{vex}}

\newcommand{\const}{\mathrm{const}}
\newcommand{\ind}{\mathrm{indec}}
\newcommand{\ext}{\mathrm{extrem}}
\newcommand{\D}{\mathcal{D}}

\newcommand{\agr}{\mathrm{aggregate}}
\renewcommand{\L}{\mathcal{L}}
\newcommand{\inv}{\mathrm{complete}}
\newcommand{\C}{\mathbb{C}}

\newcommand{\mrs}{\mathrm{MRS}}
\newcommand{\conv}{\mathrm{conv}}
\renewcommand{\E}{E}
\newcommand{\T}{T}
\newcommand{\one}{\mathds{1}}
\newcommand{\cv}{\mathrm{CV}}
\newcommand{\ev}{\mathrm{EV}}
\newcommand{\av}{\mathrm{AV}}

\def\dd{\mathrm{d}}

\def\cC{\mathcal{C}}
\newcommand{\argmax}{\operatornamewithlimits{argmax}}
\newcommand{\argmin}{\operatornamewithlimits{argmin}}

\begin{document}

\maketitle

\begin{abstract}
We revisit a classical question of how individual consumer preferences and incomes shape aggregate behavior. We develop a method that applies to populations with homothetic preferences and reduces the hard problem of aggregation to simply computing a weighted average in the space of logarithmic expenditure functions. We apply the method to identify aggregation-invariant preference domains, characterize aggregate preferences from common domains like linear or Leontief, and describe indecomposable preferences that do not correspond to the aggregate behavior of any non-trivial population. Applications include robust welfare analysis, information design, discrete choice models, pseudo-market mechanisms, and preference identification.

\end{abstract}


\ph{TBD for me: $(ii)$ the explicit representation of CES with $\sigma>1$ as a convex combination of linear preferences}

\fed{Philip, add acks to those whom you discussed the paper with.}

\newpage 

\section{Introduction}

\medskip

We revisit a classical question in economics: how do individual consumer preferences and
incomes shape aggregate behavior?
Aggregate demand of consumers determines market outcomes and is often
used as a primitive in economic models.
However, aggregate behavior is always the result of individual choices.
We aim to understand how assumptions about individuals---about their preferences and incomes---impose restrictions on aggregate demand, and also to determine what can be inferred about individual characteristics based on observed aggregate behavior. 
Since \cite*{sonnenschein1973walras}, the profession has been quite pessimistic regarding progress one can possibly make in that direction. \cite*{kreps2020course} provides an accurate summary of the ``anything goes’’ consensus:


\begin{quote}\emph{``So what can we say about aggregate demand based on the hypothesis that individuals are preference/utility maximizers? Unless we are able to make strong assumptions about the distribution of preferences or income throughout the economy (e.g., everyone has the same preferences) there is little we can say.’’}
\end{quote}
\cite{jackson2019non} have recently rekindled interest in the aggregate consumer agenda by reaffirming this negative statement in a new context.

In this paper, we show that there is a middle ground, a rich setting where the aggregation problem is neither trivial nor intractable. Namely, we develop a method for aggregate demand analysis 
which works for any population of consumers with homothetic preferences and delivers precise and exhaustive answers to a number of otherwise hard questions. For example, what can one say about aggregate demand if all goods are perfect substitutes at the individual level, so that preferences are linear but differ across consumers? Or, can a CES aggregate demand with complements be obtained by aggregating Leontief preferences? Furthermore, our method has immediate applications to robust welfare analysis and the algorithmic complexity of markets.

\medskip 

The key idea of our method is to do aggregation in the space of logarithmic expenditure functions.
In this space, the allegedly hard problem of demand aggregation boils down to simply computing a weighted average
of logarithmic expenditure functions with weights equal to relative incomes. This weighted average corresponds to a homothetic preference of the aggregate consumer whose demand coincides with the demand of the original population for all prices.\footnote{In contrast to the textbook definition of a representative consumer by \cite*{gorman1961class}, the preference of the aggregate consumer is allowed to depend on income distribution. As a result, the aggregate consumer is well-defined for all populations with homothetic preferences \citep*{eisenberg1961aggregation} while Gorman's representative fails to exist for such populations unless all preferences are identical.}

The demand aggregation problem reduces to the problem of preference aggregation: understanding how the preference of the aggregate consumer depends on preferences and incomes of the underlying population. Our method provides insights into the economics and geometry of this problem:

   \begin{enumerate}[label=(\roman*)]    
\item \emph{Aggregation-invariant classes of preferences.}
Consider a domain (which is the term we use in what follows for ``a class'' or ``a family'') of preferences.\footnote{ E.g., linear, Leontief, or any other subset of homothetic preferences} We call a domain  invariant with respect to aggregation if the aggregate preference always belongs to this domain. For example, the  domain of all homothetic preferences, a domain containing just one preference, and the domain of Cobb-Douglas preferences are invariant. Invariant domains are important as in such domains, the aggregate behavior of a population has the same structure as that of an individual. 
We characterize all invariant domains by the convexity property of the set of associated logarithmic expenditure functions.  This characterization allows us to construct simple parametric invariant domains and describe the minimal invariant domains containing popular ones.
    
\item \emph{Characterization of feasible aggregate behaviors for a given class of preferences.}
Suppose we know a domain to which individual consumers' preferences belong. What aggregate behaviors can be micro-founded? This question boils down to 
understanding what preferences can be obtained by aggregation of individual preferences from this domain. We call the set of all such aggregate preferences the domain completion.
The notion of completion is closely related to invariance:  the completion of a domain is the minimal invariant collection of preferences containing the domain. The characterization of the invariance implies that the completion of a domain can be found by computing the convex hull of the set of logarithmic expenditure functions.
We describe the completion explicitly for the domains of linear and Leontief preferences.\footnote{Surprisingly, this problem happens to be connected to several branches of economics and mathematics such as additive random utility models (ARUM), completely monotone functions, the Stieltjes transform, and even complex analysis.}
A viable conjecture would be that the completion of linear preferences gives all preferences exhibiting substitutability among goods. We show that this holds true in the case of two goods only. The completion of Leontief domain does not encompass all preferences with complementarity even for two goods.

\item \emph{Decomposition of preferences.} Consider the inverse to the problem of aggregation: given a preference from a particular domain, represent it as an aggregation of preferences from the same domain.
There is always a trivial representation since we can take a population where each agent has the same given preference. We call those preferences that can only be represented by themselves indecomposable. Geometrically,  indecomposable preferences correspond to extreme points in the space of logarithmic expenditure functions. As 
any point of a convex set can be represented as a convex combination of extreme points, indecomposable preferences play the role of elementary building blocks: any preference can be represented as an aggregation of indecomposable ones.
For example, linear and Leontief preferences are indecomposable in the domain of all homothetic preferences.
 We show that the set of indecomposable preferences is much bigger  and contains all Leontief preferences on linear composite goods. In particular, aggregation of linear and Leontief preferences together does not give the whole domain of homothetic preferences. We also explore indecomposable preferences in the domains with substitutability or  complementarity.

\end{enumerate}

\medskip

We illustrate how our approach can be applied to several economic environments.

\paragraph{Robust welfare analysis.} An analyst observing  aggregate behavior aims to estimate a quantity depending  on the structure of the population on the micro level, e.g., the change in welfare induced by a change in prices. As the same aggregate behavior can be compatible with different populations, it can also be compatible with a range of values for welfare. This range is non-trivial even for standard welfare measures such as the equivalent variation. 
We show how one can compute this range using information-design tools. As a corollary, we
obtain that the utility of a representative consumer---the standard proxy for the population's welfare---gives the most pessimistic estimate on the actual welfare gains. This conclusion suggests a possible explanation for unexpectedly low values of gains from trade obtained in recent quantitative literature.

\paragraph{Domain complexity, Fisher markets, and bidding languages for pseudo-market mechanisms.}
 For invariant domains, the aggregate behavior is as simple as that of a single agent. Since the completion of a domain is the minimal invariant domain containing it, the completion reflects the complexity of possible aggregate behavior. 
 We formalize this intuition in application to Fisher markets: simple exchange economies where consumers with fixed incomes face a fixed supply of goods. Such markets are essential for the pseudo-market (or competitive) approach to fair allocation of resources \citep*{moulin2019fair,pycia2022pseudo} and serve as a benchmark model for equilibrium computation in algorithmic economics \citep*{nisanalgorithmic}. Computing an equilibrium of a Fisher market turns out to be a challenging problem even in a seemingly innocent case of linear preferences, thus limiting the applicability of pseudo-market mechanisms. We explore the origin of the complexity and demonstrate that computing equilibria can be hard, even in small parametric domains, if their completion is large. We show how to construct domains with small completion and describe an algorithm making use of this smallness.  The choice of a domain is interpreted as bidding language design.

\paragraph{Preference identification.} Given a domain of individual preferences, we ask whether observing the  market demand is enough to identify the distribution of preferences and income over the population. We relate the possibility of identification to the geometric simplex property of the domain, meaning that there is a unique way to represent each preference  as an aggregation of indecomposable ones.
 Examples of domains where identification is possible include Leontief and linear preferences over two goods. 

\bigskip

The methodological importance of the link
 between aggregation and weighted averages is 
that it expands the machinery of consumer demand analysis by delivering new tools, the most important of which  
are convexification and extreme-point techniques of Choquet theory. Our paper demonstrates the power of these tools, which are increasingly popular in other branches of economic theory, such as information economics and mechanism design \citep*{kleiner2021extreme}, for  aggregate demand analysis. We also uncover a connection between aggregation and modern literature in convex geometry on the geometric mean of convex sets \citep*{milman2017non,boroczky2012log}. This connection not only enables the use of geometric tools for studying our economic problem but also leads to new insights about the geometric mean of convex sets suggested by the economic interpretation.


\subsection{Related literature}

As suggested by the quote from~\citep*{kreps2020course},
the existing results on demand aggregation have fallen into one of the two extremes. One extreme stems from the classical general equilibrium literature dealing with economies where
agents have general convex preferences and earn money by trading  their endowments. This literature  concludes that the aggregate demand inherits no properties of individual behavior \citep*{sonnenschein1973walras, mantel1974characterization, debreu1974excess,  chiappori1999aggregation}. 
The opposite extreme is given by the representative agent literature aiming to replace the population with a single rational agent whose preferences are independent of the income distribution \citep*{gorman1953community,gorman1961class};  
 see also earlier results by Antonelli and Nataf discussed by \citep*{shafer1982market}.
The independence requirement is so restrictive that it is fair to say that Gorman's representative almost never exists. Exceptions are very special cases, e.g., when the whole population has identical homothetic preferences.
 The profession has been divided on how seriously the non-existence should be taken.
Applied researchers often postulate the existence of a representative consumer \cite[e.g.,][]{chamley1986optimal,roooff1990equilibrium} but this approach is criticized as lacking micro-foundations \citep[e.g.,][]{kirman1992repr,carroll2000requiem}. 

As demonstrated by \cite*{jackson2019non}, tweaking Gorman's notion of the representative consumer while maintaining an analog of income independence does not  substantially alter  the non-existence conclusion. 
We escape this conclusion by allowing the aggregate consumer to depend on the income distribution. The existence of such an aggregate consumer was pointed out by \cite*{eisenberg1959consensus} for populations with linear preferences and, by \cite*{eisenberg1961aggregation} and \cite*{chipman1979social}, in the whole domain of homothetic preferences (see a survey by \cite*{shafer1982market}).\footnote{\cite*{eisenberg1959consensus} were motivated by the question of probabilistic forecast aggregation and introduced an auxiliary exchange economy of bets, a ``prediction market'' in modern terms. A closely related idea for belief aggregation is known in the financial-market literature as the Negishi approach \citep*{jouini2007consensus}.}
In more recent literature, this insight has gone largely unnoticed, except for algorithmic economics and fair allocation mechanisms; see Section~\ref{sec_Fisher}. The converse statement to Eisenberg's result was obtained by \cite*{jerison1984aggregation}, who showed that homotheticity is necessary for the existence provided that incomes are fixed.

Demand aggregation has been studied in the context of household behavior, where heterogeneous agents redistribute
their incomes so that the resulting individual consumption maximizes the household's welfare. Under mild assumptions, such households behave like a single representative agent  \citep*{samuelson1956social,varian1984social,jerison1994optimal} confirming the common wisdom that the unrestricted endogeneity of  incomes is more important for the negative  Sonnenschein-Mantel-Debreu results than the generality of preferences \citep*{mantel1976homothetic,hildenbrand2014market}.
An active empirical literature aims to link household consumption with individual characteristics of its members; see
\citep*{browning1998efficient,browning2013estimating} and references therein.

The relation between individual preferences and the representative preference of a welfare-maximizing household has been analyzed by \cite*{chambers2018structure} for egalitarian welfare functionals. Their analysis suggests that the role played by the geometric mean of convex sets in our setting with independent consumers (Section~\ref{sec_geom_mean}) is played by the Minkowski sum in egalitarian households.

Our application to robust welfare analysis is close in spirit to that of two concurrent papers  \cite*{kang2022robust} and \cite*{kocourek2022demand}. The three approaches address different aspects of robustness and are complementary. We are interested in the setting where the aggregate demand is not a sufficient statistic for welfare.
\cite*{kang2022robust} assume that aggregate demand is a sufficient statistic, but the number of distinct observations of the demand is not enough to pin it down.
In the paper by \cite*{kocourek2022demand}, the aggregate demand is perfectly observed and is a sufficient statistic, but consumers may not be fully aware of prices.
Curiously, all three papers rely on auxiliary Bayesian persuasion problems of different forms and origins. If the aggregate behavior is not a sufficient statistic for welfare, an alternative to our robust approach is using extra information about the population's structure \citep[e.g,][]{blundell2003nonparametric, schlee2007measuring, bhattacharya2015nonparametric,maes2022price}. For example, \cite*{maes2022price} discuss how to correct welfare estimates based on the representative agent approach if the analyst additionally knows moments of the demand distribution over the population.

The notion of aggregation-invariant classes of preferences that we use in our paper is stronger than the requirement of aggregation-invariance of parametric demand systems used in the empirical consumer demand literature. The almost ideal demand system (AIDS) developed by \cite*{deaton1980almost} and its extensions (e.g., \cite*{banks1997quadratic,lewbel2009tricks}) have a feature that the aggregate demand of a heterogeneous population can be captured by the same functional form. The nuance is that such functional forms may correspond to a preference of a rational individual only locally, i.e., for a choice set that is a proper subset of the positive orthant; see the discussion of this issue for the translog preference (a member of AIDS family) in Example~\ref{ex_translog} and Footnote~\ref{footnote_translog} therein. Similarly, 
 the Stone–Geary linear expenditure system (LES)---also known to be convenient for empirical demand aggregation---defines a valid preference only for high-income individuals.   To summarize, AIDS and related demand systems do not give rise to globally aggregation-invariant preferences. This phenomenon has a fundamental origin: by the result of \cite*{jerison1984aggregation},  the aggregate demand of heterogeneous consumers with non-homothetic preferences may even fail to correspond to a preference of a rational consumer as we vary incomes.

Our results on the identification of preference distributions contribute to broad econometric literature on non-parametric identification of stochastic choice models; see surveys  \citep*{matzkin2007nonparametric, matzkin2013nonparametric}. This  literature  has mostly focused on identifying an unknown deterministic part of the decision maker's utility, whereas our problem can be interpreted as identifying the noise distribution when the deterministic component is known; see Section~\ref{sec_ARUM} for a formal relation between additive random utility models and market demand for populations with linear preferences. Preference identification has also been studied in the literature on  household behavior, and identification has been obtained either for small populations, e.g., two-agent households \citep*{chiappori1988rational}
or under the assumption of preferences ``orthogonality''  \citep*{chiappori2009microeconomics}. Our results do not restrict the size of populations and allow agents to have closely aligned preferences. 
\fed{Mention \cite*{chiappori1999aggregation}}




\section{Preliminaries}\label{sec_preliminaries}
In this section, which mainly serves the purpose of consistency, we introduce the notation and recap the concepts of consumer demand theory, which we need in what follows.
\paragraph{Notation.} We use $\R$ for the set of all real numbers, $\R_+$ and $\R_-$ for non-negative/non-positive ones, and $\R_{++}$ and $\R_{--}$ for strictly positive/negative ones. Ratios of the form  ${t}/{0}$ with $t\geq 0$ are assumed to be equal to $+\infty$.

Bold font is used for vectors, e.g., $\vec{x}=(x_1,\ldots, x_n)\in\R^n$. For a pair of vectors of the same dimension, we write $\vec{x}\geq \vec{y}$ if the inequality holds component-wise, i.e., $x_i\geq y_i$ for all $i$. The scalar product of $\vec{x},\vec{y}\in \R^n$ is denoted by   $\langle\vec{x}, \vec{y}\rangle=\sum_{i=1}^n x_i\cdot y_i$.
For subsets of $\R^n$, multiplication by a scalar and summation are defined element-wise: $\alpha\cdot X=\{\alpha\cdot  \vec{x}\,\colon \, \vec{x}\in X\}$ and $X+Y=\{\vec{x}+\vec{y}\,\colon\, \vec{x}\in X,\,\vec{y}\in Y\}$; the latter is known as the Minkowski sum of sets. 
 The standard $(n-1)$-dimensional simplex  is denoted by $\Delta_{n-1}=\{\vec{x}\in \R_+^n\,\colon\, x_1+\ldots+x_n=1\}$.

  The gradient of a function $f=f(\vec{x})$ is the vector of its partial derivatives $\nabla f=\left(\frac{\partial f}{\partial x_1},\ldots, \frac{\partial f}{\partial x_n}\right).$ 

\paragraph{Preferences and demand.} 
Consider a consumer endowed with a budget $b\in \R_{++}$ and a convex, monotone, continuous and homothetic preference $\succsim$ over bundles $\vec{x}\in {\R}_{+}^{n}$ of $n\geq 1$ divisible goods. We assume that the preference is non-degenerate, i.e., the consumer is not indifferent across all bundles. Such a preference can be represented by a non-decreasing, continuous, concave, homogeneous utility function $u\colon \R_+^n\to \R_+$,  not identically equal to zero. The homogeneity property $u(\alpha\cdot \vec{x})=\alpha\cdot  u(\vec{x})$ with $\alpha\geq 0$ pins down $u$ uniquely up to multiplication by a positive constant. In what follows, we refer to such preferences as homothetic and utilities as homogeneous. 

The consumer demand consists of her most preferred affordable bundles:
$$D(\vec{p},b)=\argmax_{\vec{x}\in \R_+^n\,\colon\, \langle\vec p, \vec{x}\rangle\leq b} u_{\succsim}(\mathbf{x}).$$
The demand is a non-empty closed convex subset of the budget set. The demand 
is
 a singleton (one-element set) for almost all~$\vec{p}$, which allows us to think  
of the demand 
as a single-valued function of $\vec{p}$ defined almost everywhere; see Appendix~\ref{sec_convex_basics}.

\paragraph{Expenditure functions.} 
One can represent a homothetic preference by its expenditure function
\begin{equation}\label{eq_price_index}
\E(\vec{p}) = \min _{\vec{x}\in \R_+^n\,\colon \, u(\vec{x})\geq 1} \langle\vec{p}, \vec{x}\rangle.
\end{equation}
This dual representation will be key for our analysis. 
The function $\E\colon \R_{+}^n\to \R$ is continuous, non-decreasing, concave, homogeneous, non-negative, and not identically equal to zero. Conversely, any continuous concave, homogeneous non-negative function on $\R_+^n$ that is not identically zero is an expenditure function for some homothetic preference.\footnote{We do not need to assume that this function is non-decreasing as any concave non-negative function $f$ on $\R_+^n$ is automatically non-decreasing. Indeed, if $f(\vec{x})>f(\vec{x}')$ for some pair $\vec{x}\leq \vec{x'}$, then, by concavity, $f\big(\vec{x'}+t\cdot \delta\vec{x}\big)\leq f(\vec{x'})+t\cdot  \delta f$ where $t\geq 0$, $\delta\vec{x}=\vec{x'}-\vec{x}$, and $\delta f=f(\vec{x'})-f(\vec{x})<0$. Choosing $t$ large enough, we get a contradiction with non-negativity.}
 

 Expenditure functions are related to indirect utilities. The indirect utility is the maximal utility that a consumer can achieve as a function of prices and her budget: 
 \begin{equation}\label{eq_indirect}
v(\vec{p},b)=\max_{\vec{x}\in \R_+^n\colon \langle\vec{p},\vec{x}\rangle\leq b} u(\vec{x}).
\end{equation}
It can be expressed through the expenditure function by
\begin{equation}\label{eq_indirect_through_expenditure}
	v(\vec{p},b)=\frac{b}{\E(\vec{p})}.
\end{equation}
 
For homothetic preferences, Shephard's lemma implies the following identity:\footnote{In this form, the result can be found in \citep*{samuelson1972unification}; see Appendix~\ref{sec_convex_basics} for a derivation.}
\begin{equation}\label{eq_demand_price_index}
   D(\vec{p},b)=b\cdot 
   \nabla \ln \left(\E(\vec{p})\right),
\end{equation}
i.e., the demand is proportional to the gradient of the logarithm of the expenditure function (the \emph{logarithmic expenditure function} in what follows). 
The identity holds for all prices $\vec{p}\in \R_{++}^n$ where~$\E$ is differentiable. This set of prices has full measure; see Appendix~\ref{sec_convex_basics}.

Consider the expenditure share function $\vec{s}(\vec{p})$ whose $i$th component $s_{i}(\vec{p})$ is the fraction of the budget that the consumer spends on good $i=1,\ldots, n$ given the prices, i.e.,
\begin{equation}\label{eq_budget_share}
s_{i}(\vec{p})=
p_i\cdot \frac{D_{i}(\vec{p},b)}{b}= p_i\cdot D_{i}(\vec{p},1).
\end{equation}
We treat $\vec{s}$ as a single-valued vector function taking values in the standard simplex $\Delta_{n-1}$ and defined on the set of $\vec{p}\in \R_{++}^n$ of full measure where the demand is a singleton. 
By~\eqref{eq_demand_price_index}, 
expenditure shares can be computed as the   elasticities of the expenditure function with respect to prices
\begin{equation}\label{eq_budget_shares_as_elasticities}
s_{i}(\vec{p})=p_i\cdot \frac{\partial \ln \left(\E(\vec{p})\right)}{\partial p_i}=\frac{\partial \ln \left(\E(\vec{p})\right)}{\partial \ln(p_i)}.
\end{equation}

For two goods, preferences can be represented via expenditure share functions using the following characterization.\footnote{To the best of our knowledge, this characterization has not appeared in the literature.} For any homothetic preference  $\succsim$ over $\mathbb{R}^2_+$, the expenditure share of the first good takes the form
\begin{equation}\label{eq_budget_share_characterization_for_two}
s_{1}(p_1,p_2)= \frac{1}{1+ Q\left(\frac{p_1}{p_2}\right)/\frac{p_1}{p_2}}
\end{equation}
for some non-decreasing non-negative function $Q\colon \R_{++}\to \R_+\cup\{+\infty\}$. Moreover, for any such function $Q$, there is a unique homothetic preference; see Lemma~\ref{lm_budget_share_two_goods} in Appendix~\ref{app_feasible shares}. 
\medskip

By plugging a function $Q$ with an infinite number of jumps in~\eqref{eq_budget_share_characterization_for_two}, we see that, rather counter-intuitively, $s_{1}(p_1,p_2)$ may change monotonicity infinitely many times as $p_2$ increases, i.e., the consumer starts spending more on the first good as the price of the second one goes up, then less, then more again, and so on.\footnote{For example, expenditure shares change monotonicity infinitely many times for Leontief preferences over an infinite number of linear  composite goods discussed in
Section~\ref{sec_indecomposable_full_domain}.}

\paragraph{Substitutes and complements.} The two important subdomains of homothetic preferences are free from the non-monotone behavior of expenditure shares described above. 

A preference $\succsim$ exhibits substitutability among the goods if the expenditure share $s_{i}(\vec{p})$ is a non-decreasing function of $p_j$ for each pair of goods $i\ne j$. For differentiable expenditure shares,
\begin{equation}
\label{eq:subst}
\frac{\partial s_i(\mathbf{p})}{\partial p_j}>0\quad\text{for all }i\neq j,
\end{equation}
i.e., when the price of a good increases, the consumer starts allocating a higher budget share on other goods. In the opposite case, when all the inequalities in \eqref{eq:subst} are reversed, $\succsim$ exhibits complementarity among the goods.

\medskip

\paragraph{Examples.} The canonical example of substitutes is given by linear preferences represented by
$$u(\vec{x})=\langle \vec{v},\vec{x}\rangle$$
for some vector of values $\vec{v}\in \R_+^n\setminus \{0\}$. An elementary computation gives the expenditure function, and formula~\eqref{eq_demand_price_index} provides expenditure shares
\begin{equation}\label{eq_price_index_and_shares_for_linear}
    \E(\vec{p}) = \min_{i=1,\ldots, n}\frac{p_i}{v_i} \qquad\mbox{and}\qquad  s_{i}(\vec{p}) =
    \begin{cases}
    1,& \text{if } \frac{v_i}{p_i}>\frac{v_j}{p_j}\ \mbox{for all $j\ne i$},\\
    0,&\text{otherwise}
    \end{cases}.
\end{equation}
From \eqref{eq_price_index_and_shares_for_linear}, the consumer spends her whole budget on the good with the highest value-to-price ratio.

\medskip
 
The standard example of complements is given by the Leontief preferences, which correspond to the following utility function
$$u(\vec{x})= \min_{i=1,\ldots, n}\frac{x_i}{v_i}$$
for some vector of values $\vec{v}\in \R_+^n\setminus \{0\}$. Note that the Leontief utility has the same functional form as the expenditure function \eqref{eq_price_index_and_shares_for_linear} for linear preferences. By duality, the expenditure function for Leontief preferences is linear
\begin{equation}\label{eq_price_index_and_shares_for_Leontief}
    \E(\vec{p}) = \langle\vec{v},\vec{p}\rangle  \qquad\mbox{and}\qquad  s_{i}(\vec{p}) =\frac{v_i\cdot p_i}{\langle\vec{v},\vec{p}\rangle}.
\end{equation}

The intersection of the domains of preferences exhibiting substitutability and complementarity consists of those preferences $\succsim$ for which expenditure shares are constant, i.e., there is a fixed vector $\vec{a}\in \Delta_{n-1}$ such that $\vec{s} (\vec{p})=\vec{a}$ for any $\vec{p}$. From \eqref{eq_demand_price_index}, the expenditure function is
\begin{equation}\label{eq_Cobb_Douglas_pindex}
\E(\vec{p})=\prod_{i=1}^{n} p_{i}^{a_i}
\end{equation}
and the corresponding preference is given by the Cobb-Douglas utility function 
\begin{equation}\label{eq_Cobb_Douglas_utility}
u(\vec{x})=\prod_{i=1}^{n} x_{i}^{a_i}.
\end{equation}

Leontief, Cobb-Douglas, and linear preferences are contained as limit cases in a widely used parametric family of preferences with constant elasticity of substitution (CES). A preference $\succsim$ is a CES preference with elasticity of substitution $\sigma>0$ if the corresponding utility function has the form
\begin{equation}\label{eq_CES}
u(\vec{x})=\left(\sum_{i=1}^n \left(a_i\cdot x_i\right)^{\frac{\sigma-1}{\sigma}}\right)^{\frac{\sigma}{\sigma-1}}
\end{equation}
for some vector $\vec{a}\in \Delta_{n-1}$. The expenditure functions and expenditure shares are given by 
\begin{equation}\label{eq_price_index_and_shares_for_CES}
    \E(\vec{p}) =\left(\sum_{i=1}^n \left(\frac{p_i}{a_i}\right)^{1-\sigma}\right)^{\frac{1}{1-\sigma}}  \qquad\mbox{and}\qquad  s_{i}(\vec{p}) =\frac{ \left(\frac{p_i}{a_i}\right)^{1-\sigma}}{\sum_{j=1}^n \left(\frac{p_j}{a_j}\right)^{1-\sigma}}.
\end{equation}
CES preferences exhibit substitutability for $\sigma>1$ and complementarity for $\sigma\in (0,1)$. Leontief, Cobb-Douglas, and linear preferences are the limiting cases as $\sigma$ goes, respectively, to $0$, $1$, and $+\infty$. The limits are taken with respect to the topology that we discuss next.

\paragraph{Topology on preferences.}
Convergence of preferences, closed and open sets, and the Borel structure are understood with respect to the following metric. We define the distance between preferences $\succsim$ and $\succsim'$
with expenditure functions $\E$ and $\E'$ by
\begin{equation}\label{eq_distance_body}
d(\succsim,\succsim')=\sup_{\vec{p}\in \Delta_{n-1}\cap \R_{++}^n} \left|\frac{\left(\ln \E(\vec{p})-\ln \E(\vec{e})\right)-\left(\ln \E'(\vec{p})-\ln \E'(\vec{e})\right)}{\left(1+\max_i |\ln p_i|\right)^2}\right|,
\end{equation}
where $\vec{e}=(1,\ldots,1)$.
The main advantage of this way of introducing the distance is that it makes the set of all homothetic preferences a compact metric space.\footnote{Economic literature has considered compact topologies on the set of preferences, e.g., the closed convergence topology of upper contour sets \citep*{hildenbrand2015core,bridges2013representations}. To the best of our knowledge, an explicit metric structure giving compactness has not appeared in the literature.}  In particular, the distance between any pair of preferences is finite and bounded by~$2$.
See Appendix~\ref{app_topology} for the intuition behind the definition.

%


\section{Preference aggregation}\label{sec_representative}

Consider $m\geq 1$ consumers $k=1,\ldots, m$.\footnote{It is straightforward to extend our analysis to the case of a continuum of consumers or, more generally, to an abstract measure space of consumers.} Consumer $k$ has a positive budget $b_k\in \R_{++}$ and a homothetic preference $\succsim_k$ over bundles of $n\geq 1$ divisible goods, which generates the individual demand $D_k(\mathbf{p},b_k)$. For any vector of prices $\vec{p}$, this population generates the market demand equal to the sum of individual demands.
%
Denote by $B$ the total budget of the population,
$
    B = \sum_{k=1}^{m} b_{k}
$.
\begin{definition}\label{def_aggregate}
A preference 
$\succsim_\agr$ is referred to as the aggregate preference for 
a population of consumers with preferences $\succsim_1,\ldots,\succsim_m $ and budgets $b_1,\ldots, b_m$ if  
\begin{equation}\label{eq_aggregation}
D_{\agr}\left(\vec{p}, \ B\right)=\sum_{k=1}^m D_{k}(\vec{p},b_k)
\end{equation}
for any price vector $\vec{p}\in \R_{++}^n$.
A consumer with preference $\succsim_\agr$ and budget $B$ is referred to as the aggregate consumer.
\end{definition}
We stress that the aggregate consumer is selected for a given collection of budgets $b_1,\ldots, b_m$ of individual consumers, and so, for a different distribution of incomes over the population, we may end up with a different aggregate consumer. This is an important distinction between Definition~\ref{def_aggregate} and the approach of  \cite*{gorman1961class} 
who insists on the independence of the aggregate preference on the income distribution which can only be achieved under very restrictive conditions.

\begin{example}\label{ex_CobbDouglas_as_aggregation_of_linear}
Consider $m=n$ single-minded consumers: $u_{i}(\vec{x})=x_i$, i.e., consumer $i$ only cares about good~$i$. Hence, no matter what the prices are, consumer~$i$ spends her total budget $b_i$ on good $i$. This observation helps to guess the aggregate consumer without any computations. Indeed, the aggregate consumer spends the amount $b_i$ out of her total budget $B$ on good $i$ independently of prices. In other words, the expenditure share of each good $i$ for the aggregate consumer is price-independent and equal to ${s}_{\agr,i}(\vec{p})={b_i}/{B}$. Hence, the aggregate consumer must have the Cobb-Douglas preferences~\eqref{eq_Cobb_Douglas_utility} with  $a_i={b_i}/{B}$. One can verify this guess directly by checking that the demand identity~\eqref{eq_aggregation} holds.
Alternatively, the result can be deduced immediately from Theorem~\ref{th_representative_index} below and explicit formulas for expenditure functions of Cobb-Douglas and linear preferences. 
\end{example}

The existence of an aggregate preference was established by \cite*{eisenberg1961aggregation} for any population of consumers with homothetic preferences. Denote  by~\(\beta_{k}\) the relative fraction of  consumer~\(k\)'s budget
\begin{equation}\label{eq_relative_budgets}
\beta_{k} = \frac{b_{k}}{B}.
\end{equation}
\cite*{eisenberg1961aggregation} showed that the aggregate preference corresponds to the utility function obtained as the solution to the following optimization problem
\begin{equation}\label{eq_EisenbergGale_utility}
u_{\agr}(\vec{x})=\max\left\{\prod_{k=1}^m \left(\frac{u_{k}(\vec{x}_k)}{\beta_k}\right)^{\beta_k}\quad:\quad \vec{x}_k\in \R_+^n,\ \ k=1,\ldots,m,\quad \sum_{k=1}^m \vec{x}_k=\vec{x}\right\}.
\end{equation}
In other words, the utility for the aggregate preference at a bundle $\vec{x}$ is equal to the maximal weighted Nash social welfare where the maximum is taken over all possible allocations of $\vec{x}$ over the consumers and consumer's weight is equal to her relative budget.\footnote{The welfare function equal to the product of consumer's utilities is dubbed the Nash social welfare or the Nash product as this welfare function naturally arises in the context of axiomatic bargaining studied by \cite*{nash1950bargaining}.}
The optimization problem~\eqref{eq_EisenbergGale_utility} is called the Eisenberg-Gale problem since it is similar to a problem studied by
\cite*{eisenberg1959consensus} in the context of probabilistic forecast aggregation.

To determine the utility of an aggregate consumer, one needs to solve the Eisenberg-Gale problem~\eqref{eq_EisenbergGale_utility} for each $\vec{x}\in\R_+^n$. Except for special cases such as Cobb-Douglas preferences, it does not admit an explicit solution and is not easy to work with both analytically and computationally; see Section~\ref{sec_Fisher}. 

We observe that the question of describing the aggregate consumer substantially simplifies if we use the dual representation of preferences via expenditure functions.
\begin{theorem}\label{th_representative_index}
Consider a population of consumers with homothetic preferences $\succsim_1,\ldots,\succsim_m $ and budgets $b_1,\ldots, b_m$.
The  preference of the aggregate consumer is described by the  expenditure function \({\E}_{\agr}\) satisfying
\begin{equation}\label{eq:pr}
\ln {\E}_{{\agr}}(\vec{p})=\sum_{k=1}^{m} \beta_{k}\cdot  \ln \E_{k}(\vec{p}),
\end{equation}
where the weights $\beta_k$ are given by~\eqref{eq_relative_budgets}.
\end{theorem}

\fed{[Move somewhere or just delete?] It is worth noting that the source of convexification in our setting is very different from, say,
\cite*{jackson2019non}, who have in mind macroeconomic settings with intertemporal utility maximization, where additive separability across outcomes matters.}

\medskip


\begin{proof} The identity~\eqref{eq:pr} follows almost immediately from ~\eqref{eq_demand_price_index}. The definition of the aggregate consumer implies the following equality 
\begin{equation}\label{eq_aggregation_gradients}
B\cdot \nabla \ln {\E}_{{\agr}}(\vec{p})=\sum_{k=1}^{m} b_{k}\cdot \nabla  \ln \E_{k}(\vec{p}),
\end{equation}
which must hold at all points of differentiability of the expenditure functions. As any concave function is differentiable almost everywhere with respect to the Lebesgue measure (Section~\ref{sec_convex_basics}), \eqref{eq_aggregation_gradients} holds on the set of full measure and can be integrated resulting in the identity~\eqref{eq:pr}. Integration constants get absorbed by the expenditure functions defined up to multiplicative constants. \end{proof}

One can see from \eqref{eq:pr} that preference aggregation is equivalent to taking convex combinations of individual logarithmic expenditure functions. The simplicity of this operation will allow us to describe domains invariant with respect to aggregation (Section~\ref{sec_invariant})
and to study the decomposition of a given preference as an aggregation of elementary ones (Section~\ref{sec_indecomposable}). 

In Appendix~\ref{sec_th1_proof}, we provide an alternative proof of Theorem~\ref{th_representative_index}, similar to the one used by~\cite*{eisenberg1961aggregation} and not relying on formula~\eqref{eq_demand_price_index}. This alternative proof clarifies that Theorem~\ref{th_representative_index} is dual to Eisenberg's result.\footnote{The connection between Eisenberg's theorem and Theorem~\ref{th_representative_index}  can be seen as a version of duality for the $\inf$-convolution \citep[Chapter 1.H]{rockafellar2009variational}.}

\fed{Comment why we do not use demand instead of LEF}

\subsection{Connection to the geometric mean of convex sets}\label{sec_geom_mean} Theorem~\ref{th_representative_index} links preferences aggregation and recent attempts to define the geometric mean of convex sets; see a survey by \cite*{milman2017non}. 
Recall that the support function of a convex set $X\subset \R^n$ is defined by 
$$h_X(\vec{p})=\sup_{x\in X}\  \langle\vec{p},\vec{x}\rangle.$$
 \cite*{boroczky2012log} define
 the weighted geometric mean of convex sets by taking the usual weighted geometric mean in the space of support functions.\footnote{Defining algebraic operations on convex sets through the usual algebraic operations on their support functions is a standard approach. For example, the Minkowski addition of convex sets corresponds to the pointwise summation of their support functions. \fed{Add reference}}
 Formally, the weighted geometric mean  of convex sets $X$ and $Y$ with weights $(\lambda,\,1-\lambda)$, $\lambda\in[0,1]$,  is the convex set $Z$ denoted by $X^{\lambda}\otimes Y^{1-\lambda}$ such that 
\begin{equation}\label{eq_geom_mean_def}
\left|{h}_{X^\lambda\otimes Y^{1-\lambda}}\right|= \left|{h}_{X}\right|^\lambda\cdot \left| {h}_{Y}\right|^{1-\lambda}.
\end{equation}
 The weighted geometric mean extends to any number of convex sets straightforwardly.\footnote{\cite*{boroczky2012log} refer to $X^\lambda\otimes Y^{1-\lambda}$ as the logarithmic sum of convex sets to distinguish it from other definitions of the geometric mean. Since we do not consider other definitions, we call $X^\lambda\otimes Y^{1-\lambda}$  the weighted geometric mean.}

 \medskip
To see the connection between preference aggregation and the geometric mean, note that the expenditure function~$\E$ is equal to the support function of the upper contour set up to a sign: 
$$\E(\vec{p})=-h_{X}(\vec{-p})\qquad\mbox{where}\quad X=\{{u}(\vec{x})\geq 1\}.$$
 We obtain the following equivalent version of  Theorem~\ref{th_representative_index}.
\begin{corollary}\label{cor_geometric_mean}
 An  upper contour set of the aggregate consumer's preferences $\big\{u_\agr(\vec{x})\geq 1\big\}$ is the weighted geometric mean of individual upper contour sets with budget-proportional weights:
 \begin{equation}\label{eq_aggregation_as_geom_mean}
\big\{u_\agr(\vec{x})\geq 1\big\}=\big\{u_1(\vec{x})\geq 1\big\}^{\beta_1}\otimes\big\{u_2(\vec{x})\geq 1\big\}^{\beta_2}\otimes\ldots\ldots\otimes \big\{u_m(\vec{x})\geq 1\big\}^{\beta_k}.
\end{equation}
\end{corollary}
In Example~\ref{ex_CobbDouglas_as_aggregation_of_linear}, we saw that Cobb-Douglas preferences over $n$ goods originate as an aggregation of~$n$ extreme linear preferences. Figure~\ref{fig_cobb_doug_as_geom_mean_of_linear} illustrates the corresponding identity of convex sets for~$n=2$ and equal budgets. 

\begin{figure}
	\begin{center}
			\begin{tikzpicture}[scale=0.275, line width=1pt]				
				\filldraw[scale=1, gray] (6,0.35) -- (11.25,0.35) -- (11.25,11.25) --  (6,11.25);
				\draw[scale=1, black] (6,0.35) --  (6,11.25);				
				\draw [->,>=stealth,black] (0,0) -- (11,0);
				\draw [->,>=stealth,black] (0,0) -- (0,11);
				\node[below] at (6,0) {$1$};
				\node[below] at (0,-0.2) {$0$};
				\node[below right] at (11,0) {\large $x_1$};
				\node[above left] at (0,11) {\large $x_2$}; 
				\node[scale=1] at (13,12) {\Huge $\frac{1}{2}$}; 						
				\node[scale=1] at (15,5) {\Huge $\otimes$}; 						
			\end{tikzpicture}
			\begin{tikzpicture}[scale=0.275, line width=1pt]				
				\filldraw[scale=1, gray] (0.35,6) -- (0.35,11.25) -- (11.25,11.25) --  (11.25,6);
				\draw[scale=1, black] (0.35,6) --  (11.25,6);				
				\draw [->,>=stealth,black] (0,0) -- (11,0);
				\draw [->,>=stealth,black] (0,0) -- (0,11);
				\node[left] at (0,6) {$1$};
				\node[below] at (0,-0.2) {$0$};
				\node[below right] at (11,0) {\large $x_1$};
				\node[above left] at (0,11) {\large $x_2$};
				\node[scale=1] at (13,12) {\Huge $\frac{1}{2}$};
				\node[scale=1] at (15,5) {\Huge $=$}; 						 						 			
			\end{tikzpicture}
		\begin{tikzpicture}[scale=0.275, line width=1pt]
			
			\draw [->,>=stealth,black] (0,0) -- (11,0);
			\draw [->,>=stealth,black] (0,0) -- (0,11);
			%
			
			\filldraw[scale=1, gray, domain=0.8:11.25,samples=1000,variable=\t] plot ({\t},{9/\t}) -- (10.5,10.5);
			\draw[scale=1,domain=0.8:11.25,samples=1000,variable=\t] plot ({\t},{9/\t});
			\draw [dashed,gray] (6,0) -- (0,6);
			\node[below] at (6,0) {$1$};
			
			\node[below] at (0,-0.2) {$0$};
			\node[below right] at (11,0) {\large $x_1$};
			\node[above left] at (0,11) {\large $x_2$};

		\end{tikzpicture}
	\end{center}
	\caption{ \textbf{Geometry:} the set bounded by the hyperbola is the geometric mean of the two orthogonal halfspaces. \textbf{Economics:} an aggregation of two extreme linear preferences where each consumer cares only about her own good  gives a Cobb-Douglas preference.\label{fig_cobb_doug_as_geom_mean_of_linear} \fed{add parentheses}
	}
\end{figure}
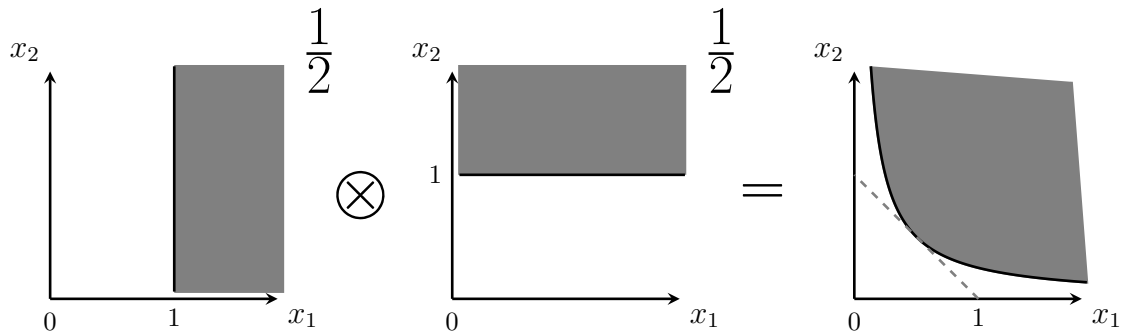

\fed{Figure~???????? depicts aggregation of two generic linear preferences: the boundary of the upper contour set has a hyperbolic segment patched between the two linear ones; computations can be found in Appendix~??????.
TODO:complete. Figure and analysis in appendix}

\medskip
Corollary~\ref{cor_geometric_mean} highlights a peculiar property of the class of convex sets that can be obtained as upper contour sets of homothetic preferences.
From formula~\eqref{eq_geom_mean_def}, it is not  evident that the geometric mean is well-defined, i.e., that we can always find a convex set whose support function is equal to ${h}_{X^\lambda\otimes Y^{1-\lambda}}$. A byproduct of Corollary~\ref{cor_geometric_mean} is that the weighted geometric mean is well-defined within the class of all  convex subsets of $\R_+^n$ that do not contain zero and are upward-closed. Indeed, any such set is an upper contour set of some homothetic preference, and the weighted geometric mean is an upper contour set of the aggregate consumer's preference.
Contrast this observation with the case of bounded convex sets which mathematical literature has mostly focused on. The weighted geometric mean is well-defined for bounded convex sets containing zero; however, sets that do not contain zero are problematic as the support function can be negative and the definition of the weighted geometric mean requires ad~hoc modifications.

\subsection{Application: robust welfare analysis as information design}\label{sec_robust}

Understanding how the population's welfare changes as a function of prices is crucial for economic policy evaluation. 
Consider an analyst who observes market demand as a function of prices and, based on this  information, aims to estimate a certain aggregate measure of individual well-being, e.g., the change in the population's welfare before and after introducing a policy 
affecting prices.

 The standard approach is postulating a representative agent and then using that agent's utility as a proxy for the population's welfare. Implicitly, this approach assumes that the market demand is a sufficient statistic for welfare.\footnote{For example, the recent quantitative literature on gains from trade uses representative consumer's utility as an aggregate welfare measure; see \citep*{costinot2014trade} for a survey. Hence, there is a one-to-one mapping between market demand and welfare. As pointed out by \cite*{arkolakis2012new} and, more recently, by \cite*{arkolakis2019elusive}, this approach leads to surprisingly low gains from trade.} However, the same market demand---hence, the same aggregate preference---can be generated by different populations of consumers. As a result, the same aggregate behavior may be compatible with a range of welfare levels. 
 
 The aggregate behavior may not be a sufficient statistic for welfare, even for standard welfare measures and preference domains. Suppose prices change from $\vec{p^0}$ to $\vec{p^1}$. The standard way of quantifying the corresponding change in individual well-being is by the change in income that makes the consumer indifferent between the two prices. Depending on whether the indifference is considered with respect to the old prices or the new ones, we get welfare measures known as the equivalent variation $\ev$ and the compensating variation $\cv$, respectively \citep[Chapter 3.I]{mas1995microeconomic}. For a consumer with income $b$ and homothetic  preference $\succsim$ represented by an expenditure function $\E$, the two measures take the form
 \begin{equation}\label{eq_variations}
 \ev_{\vec{p^0}\to\vec{p^1}}(\succsim,b)=b\cdot\left(\frac{\E(\vec{p^0})}{\E(\vec{p^1})}-1\right)\qquad \mbox{and} \qquad \cv_{\vec{p^0}\to\vec{p^1}}(\succsim,b)=b\cdot \left(1-\frac{\E(\vec{p^1})}{\E(\vec{p^0})}\right).
 \end{equation}
 For a population of consumers, the change in welfare is defined as the sum of individual changes, e.g., for the equivalent variation $\ev$, we get
 \begin{equation}\label{eq_welfare_equivalent_variation}
 W_\ev\big[(\succsim_k,b_k)_{k=1,2,\ldots}\big]=\sum_k \ev_{\vec{p^0}\to\vec{p^1}}(\succsim_k,b_k).
 \end{equation}
 The following toy example illustrates that the aggregate behavior may not only be compatible with a range of welfare levels but may even fail to determine the direction of the welfare change.
 \begin{example}\label{ex_direction_of_welfare} In a two-good economy, a population of Cobb-Douglas consumers generates market demand equal to $1/p_1$ for the first good and $2/p_2$ for the second. In other words, this population behaves like a single Cobb-Douglas consumer with utility $u_\agr(\vec{x})=x_1^{{1}/{3}}\cdot x_2^{{2}/{3}}$ and budget $B=3$.
The government considers introducing new tariffs changing prices from $\vec{p^0}=(1,64)$ to $\vec{p^1}=(32,32)$. Will this policy be beneficial for the population?

Assume that the well-being of an individual is measured by the equivalent variation, and so the population's welfare is given by~\eqref{eq_welfare_equivalent_variation}. Recall that the expenditure function for a Cobb-Douglas consumer with $u(\vec{x})=x_1^{\alpha}\cdot x_2^{1-\alpha}$ equals $\E(\vec{p})=p_1^{\alpha}\cdot p_2^{1-\alpha}$.

Consider a population of consumers with identical preferences $\succsim_k=\succsim_\agr$ for any $k$ and total budget $B$. For such a population,~$W_\ev$ is equal to the equivalent variation for the aggregate consumer:
$$
W_\ev= \ev_{\vec{p^0}\to\vec{p^1}}(\succsim_\agr, B)=-\frac{3}{2}<0.$$
On the other hand, the same aggregate demand can be generated by a population of two single-minded consumers: one with utility $u_1(\vec{x})=x_2$ and budget $b_1=2$ and the other with $u_2(\vec{x})=x_1$ and $b_2=1$; see Example~\ref{ex_CobbDouglas_as_aggregation_of_linear}. For this heterogeneous population, we get, using \eqref{eq_welfare_equivalent_variation}:
$$
 W_\ev=\frac{63}{32}>0.$$
We conclude that the same aggregate behavior can be compatible with a range of values for $W_\ev$ and, moreover, this range can contain zero, i.e., the aggregate consumer may not even predict whether the change in prices is beneficial for the underlying population or not. 
\end{example}
 \medskip

 We propose a robust approach to welfare analysis. Since the aggregate behavior may not pin down the exact value of the population's welfare (or any other functional of interest depending on the population's structure on the micro level), we aim to compute the range of possible values being fully agnostic about the specific decomposition of the market demand into individual demands.  For this purpose, we combine  Theorem~\ref{th_representative_index} with insights from information design. 

Let us describe the problem formally. Since market demand determines the aggregate consumer, we assume that the aggregate preference $\succsim_\agr$ and the total income $B$ in the economy are given.
In addition, we assume that the analyst knows that individual consumers have preferences from some subset $\D$ of homothetic preferences; we will refer to $\D$ as a domain.\footnote{Examples of domains include all homothetic preferences, linear preferences, preferences exhibiting complementarity, any finite collection of preferences, or any other subset of homothetic preferences.}
The goal is to find the range $[\underline{W},\,\overline{W}]$ of values of a  functional
\begin{equation}\label{eq_welfare_general}
	W=W\big[(\succsim_k,b_k)_{k=1,2,\ldots}\big]  
\end{equation}
over all finite populations of consumers $k=1,2,\ldots$ with  preferences $\succsim_1,\succsim_2,\ldots$ and  incomes $b_1,b_2,\ldots$ such that the individual preferences  aggregate up to $\succsim_\agr$, incomes sum up to $B=\sum_k b_k$, and individual preferences belong to $\D$.

With a domain $\D$, we associate the set $\L_\D$ of all logarithmic expenditure functions of preferences from $\D$:
$$\L_\D=\Big\{f\ :\  \R_{++}^n\to\R\quad:\quad f=\ln  \E_{\succsim},\quad \succsim\ \in\ \D\Big\}.$$
The set $\L_\D$ inherits the freedom in the choice of expenditure functions: if $f\in \L_\D$, then  $f+\const$ is also in $\L_\D$ and corresponds to the same preference; we will call such elements of $\L_\D$ equivalent.

We can think of preferences as classes of equivalent points in $\L_\D$. 
The key observation that makes the problem tractable is that, in this space, populations with given aggregate behavior $\succsim_\agr$ correspond to all the different ways of representing  $\ln E_\agr$ as a weighted average of logarithmic expenditure functions $\ln \E_{\agr}=\sum_k \beta_k \ln \E_k$; see Figure~\ref{fig_persuasion}. 
 Hence, to find the range of values of $W$,
 we need to minimize or maximize it over all such representations.  A 
 similar optimization problem arises in the context of Bayesian persuasion \citep*{aumann1995repeated,kamenica2011bayesian} and, indeed, there is a formal connection between the two problems (see Appendix~\ref{app_connection_to_persuasion}).
 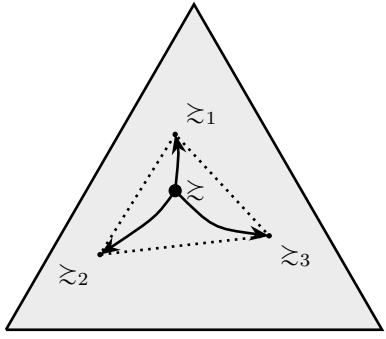
\begin{figure}
	\begin{center}
		\begin{tikzpicture}[scale=0.5, line width=1pt]	
			\usetikzlibrary {arrows.meta}
			\filldraw[scale=1, lightgray!30] (0,0) -- (10,0) --(5,8.66)--(0,0);
			\draw  (0,0) -- (10,0) -- (5,8.66)--(0,0);
			\coordinate (A) at (4.5,3.7);
			\coordinate (B) at (4.5,5.2);
			\coordinate (C) at (2.5,2);
			\coordinate (D) at (7,2.5);
			
			\filldraw[black] (A) circle (4pt) node [anchor=west]  {$\succsim$}; 
			\draw [-Stealth] (A) .. controls (4.6,4.7) .. (B);
			\filldraw[black] (B) circle (1pt) node[anchor=west]  {}; 
			\draw [-{Stealth[]}] (A) .. controls (4.0,3)  .. (C);
			\filldraw[black] (C) circle (1pt) node[anchor=west]  {};
			\draw [-Stealth] (A) .. controls (5.5,2.7) .. (D);
			\filldraw[black] (D) circle (1pt) node[anchor=west]  {};
			\draw [dotted] (B) -- (C) --  (D) -- (B);
			\node [above right] at (B) {$\succsim_1$};
			\node [below left] at (C) {$\succsim_2$};
			\node [below right] at (D) {$\succsim_3$};
		\end{tikzpicture}
	\end{center}
	\caption{In the space of logarithmic expenditure functions, all populations with 
		aggregate behavior captured by $\succsim$ 
		correspond to all possible ways to represent  $\succsim$ as a weighted average of other points.
		\label{fig_persuasion} 
	}
\end{figure}

To obtain an explicit formula for the range of values, we assume a particular functional form of~$W$: 
\begin{equation}\label{eq_welfare_additive}
	W=\sum_{k} b_k\cdot w(\succsim_k).
\end{equation}
This functional form captures both the equivalent and the compensating variations~\eqref{eq_welfare_equivalent_variation}. In particular, the individual  measure of well-being $w$
can also depend on other parameters---e.g., prices before and after a market intervention---but such dependence does not affect our analysis and is therefore omitted. 
Functionals of the form~\eqref{eq_welfare_additive}  can also capture distributional objectives with no welfare interpretations, e.g., the total income of agents having preferences from a certain subset $\D'\subset \D$, which corresponds to $w=\one_{\D'}$.


For a function $f$ defined on a subset $X$ of a linear space, its concavification $\cav_X[f]$ is the smallest concave function larger than~$f$ on~$X$. It can be computed as follows:
\begin{equation}\label{eq_concavification}
\cav_X[f](x)=\sup\left\{\sum_k \beta_k\cdot  f(x_k)\quad : \quad x=\sum_k \beta_k\cdot  x_k,\quad x_k\in X, \quad \beta_k\geq 0, \quad \sum_k \beta_k=1\right\},
\end{equation}
where the supremum is taken over all finite convex combinations of points in $X$ that average to~$x$. 
 Note that the set $X$ may not be convex, but the concavification is well-defined on its convex hull. 
\fed{Mention 
\cite*{doval2018constrained}} 
One can similarly define convexification $\vex_X[f]$ as the biggest convex function smaller than~$f$ on~$X$, i.e., 
$\vex_X[f]=-\cav_X[-f]$. 
\begin{proposition}\label{prop_welfare_as_persuasion}
	For $W$ of the form~\eqref{eq_welfare_additive}, the range of values $[\underline{W},\,\overline{W}]$ compatible with an aggregate preference $\succsim_\agr$, total income $B$, and individual domain of preferences $\D$ is given by  
\begin{equation}\label{eq_welfare_range}
	[\underline{W},\,\overline{W}] = \Big[B\cdot \vex_{\L_\D}\big[w\big]\big(\succsim_\agr\big),\ \ B\cdot \cav_{\L_\D}\big[w\big]\big(\succsim_\agr\big)\Big],
\end{equation}
where $\L_\D$ denotes the set of logarithmic expenditure functions corresponding to the domain $\D$.\footnote{Note that the function $w(\succsim)$ in~\eqref{eq_welfare_range} is treated as a functional on the space of logarithmic expenditure functions and is concavified and convexified over this space.}
\end{proposition}
 
We prove Proposition \ref{prop_welfare_as_persuasion} in Appendix~\ref{app_persuasion}. The following example illustrates how the proposition applies to the setting from Example~\ref{ex_direction_of_welfare}. 

\medskip

\begin{example}\label{ex_direction_welfare_revisited}
Let  $\D$ be the domain of Cobb-Douglas preferences over two goods. The corresponding logarithmic expenditure functions are given by
$\ln E(\vec{p})=\alpha \ln p_1+(1-\alpha)\ln p_2$.
As in Example~\ref{ex_direction_of_welfare},
suppose that the aggregate behavior corresponds to $\alpha=1/3$ and total budget $B=3$, and prices change from $\vec{p^0}=(1,64)$ to $\vec{p^1}=(32,32)$.
The goal is to estimate the range of the welfare change measured by the equivalent variation~\eqref{eq_welfare_equivalent_variation}. The corresponding functional $w(\succsim)$ takes the form
$$w_\ev(\succsim)=\frac{\E(\vec{p^0})}{\E(\vec{p^1})}-1=2^{1-6\alpha}-1.$$
Since the set $\L_D$ of logarithmic expenditure functions can be identified with the interval $\alpha\in [0,1]$ via a linear transformation, convexification and concavification over $\L_D$ boil down to the same operations with respect to $\alpha$. Abusing the notation, we denote $2^{1-6\alpha}-1$ by $w_\ev(\alpha)$. The function $w_\ev(\alpha)$ is convex and so $\vex_{[0,1]}[w_\ev](\alpha)=w_\ev(\alpha)$ and 
$\cav_{[0,1]}[w_\ev](\alpha)=\alpha \cdot w_\ev(0)+(1-\alpha)w_\ev(1)$; see Figure~\ref{fig_example_range}.
\begin{figure}
	\begin{center}
		\includegraphics[width=5.5cm]{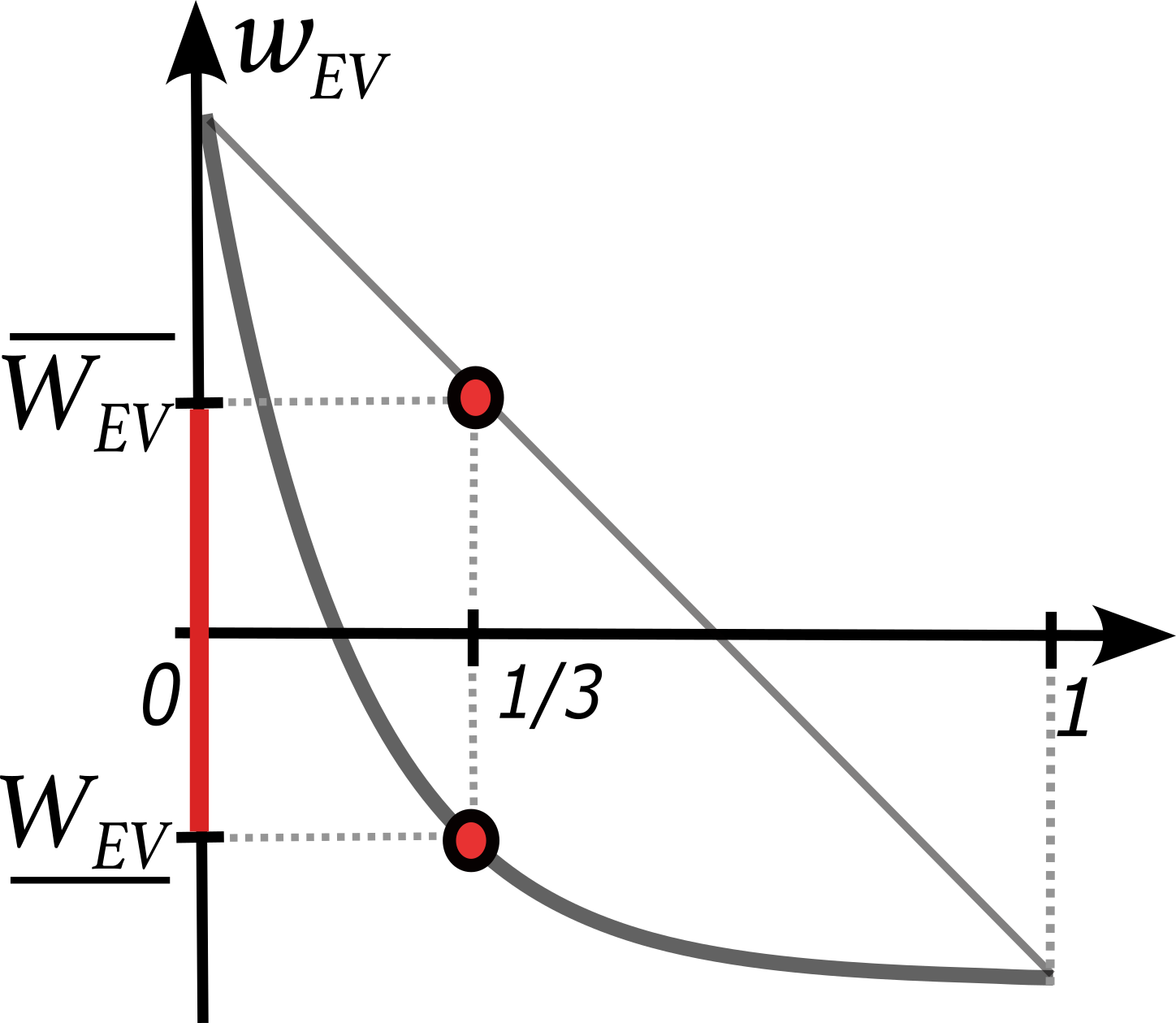}
	\end{center}
	\caption{The range of values $[\underline{W_\ev},\, \overline{W_\ev}]$ for the equivalent variation in Example~\ref{ex_direction_welfare_revisited}. \label{fig_example_range}
	}
\end{figure}
We obtain the following answer for the range of values of $W_\ev$: 
$$\big[\underline{W_\ev},\,\overline{W_\ev}\big]= \Big[3\cdot \vex_{[0,1]} [w_\ev](1/3),\ \ 3\cdot \cav_{[0,1]} [w_\ev](1/3) \Big]= \left[-\frac{3}{2},\ \ \ \frac{63}{32} \right].$$
In particular, the two populations described in Example~\ref{ex_direction_of_welfare} give the lowest and the highest possible values for the welfare change.
\end{example}	

Let us now discuss the general implications of Proposition~\ref{prop_welfare_as_persuasion}.
In Example~\ref{ex_direction_welfare_revisited}, the functional $w$ corresponding to the equivalent variation turned out to be a convex function of the logarithmic expenditure function. This is true in general. Indeed, the ratio
$$\frac{\E(\vec{p^0})}{\E(\vec{p^1})}=\exp\left( \ln\E(\vec{p^0})-\ln\E(\vec{p^1})\right) $$
 is a convex function over the space of logarithmic expenditure functions as the composition of a convex function  $\exp(t)$ and a linear functional evaluating the difference between the values of $\ln E$ at two different points $\vec{p}$. 
 Therefore, $\vex_{\L_\D}[w_\ev]=w_\ev$ for the equivalent variation. Similarly, one concludes that the compensated variation is concave, and we obtain the following corollary of Proposition~\ref{prop_welfare_as_persuasion}.
\begin{corollary}\label{cor_pessimistic}
If the change in welfare is measured by the equivalent variation, then $\underline{W_\ev}$ equals $\ev(\succsim_\agr,B)$, i.e.,  
the representative agent approach gives the most pessimistic prediction for the welfare change of the population. On the other hand, if the compensating variation is used, $\overline{W_\cv}$ is equal to $\cv(\succsim_\agr,B)$, and so the representative consumer provides the most optimistic estimate.
\end{corollary}
The trade literature relies on the equivalent variation, and hence, the representative agent approach used by this literature gives the lower bound on the actual welfare change. Therefore, the result of Corollary \ref{cor_pessimistic} can serve as a possible explanation of surprisingly low gains from trade found in quantitative trade literature
\citep*{arkolakis2012new}.
\smallskip

A common assumption in industrial organization and trade literature is that a population behaves like a single CES consumer. The following example illustrates how to compute the range of values for the equivalent variation for CES aggregate behavior, provided that individual preferences exhibit substitutability.

\begin{example}\label{ex_CES_welfare}
Let the domain of individual preferences $\D$ be all the preferences over two goods exhibiting substitutability. Assume that the aggregate consumer has a CES preference corresponding to
\begin{equation}\label{eq_CES_two_goods}
u_\agr(\vec{x})=\left(\left(a_1\cdot x_1\right)^{\frac{\sigma-1}{\sigma}}+\left(a_2\cdot x_2\right)^{\frac{\sigma-1}{\sigma}}\right)^{\frac{\sigma}{\sigma-1}}\qquad \mbox{with}\qquad \sigma>1.
\end{equation}
Our goal is to find the range of values $[\underline{W_\ev},\overline{W_\ev}]$ for welfare change measured by the equivalent variation~\eqref{eq_welfare_equivalent_variation}. By Corollary~\ref{cor_pessimistic}, the pessimistic bound $\underline{W_\ev}$ is given by $\ev_{\vec{p^0}\to\vec{p^1}}(\succsim_\agr,B)$ and so we focus on computing $\overline{W_\ev}$. 

To concavify a convex function $f$ over a convex set $X$, it is enough to maximize in~\eqref{eq_concavification} over extreme points of $X$. We will need the following facts established in subsequent sections: the set $\L_{\D}$ for two substitutes is convex (Section~\ref{sec_invariant}), and linear preferences correspond to its extreme points (Section~\ref{sec_S_indecomposable}).  Formula~\eqref{eq_CES_as_aggregation_of_linear} gives a unique way to represent a CES preference as an aggregation of linear ones: the marginal rate of substitution $\mrs=v_1/v_2$ of a linear preference $u(\vec{x})=v_1\cdot x_1+v_2\cdot x_2$ has to be distributed over the population according to the distribution $\mu$ with the following cumulative distribution function:
\begin{equation}\label{eq_CES_CDF}
\mu\big([0,\,t)\big)=\frac{(a_2\cdot t)^{\sigma-1}}{(a_1)^{\sigma-1}+ (a_2\cdot t)^{\sigma-1}}.
\end{equation}

Since $\ev$ is a convex functional, we conclude that its concavification corresponds to representing $\ln E_\agr$ as a convex combination of extreme points of $\L_{\D}$ which correspond to linear preferences. Using formula~\eqref{eq_price_index_and_shares_for_linear}
for the expenditure function of a linear preference~$\succsim$ with $\mrs=t$, we get
\begin{equation}\label{eq_EV_linear}
\ev_{\vec{p^0}\to\vec{p^1}}(\succsim,b)=b\left( \frac{p_2^0}{p_2^1}\cdot \frac{\min\{p_1^0/p_2^0,\  t\}}{\min\{p_1^1/p_2^1,\ t\}}-1\right).
\end{equation}
Hence, the maximal value of the welfare change compatible with the CES aggregate preference $\succsim_\agr$ is given by averaging~\eqref{eq_EV_linear} with respect to $\mu$:
$$\overline{W_\ev}=B\cdot \int_0^\infty \left( \frac{p_2^0}{p_2^1}\cdot \frac{\min\{p_1^0/p_2^0,\  t\}}{\min\{p_1^1/p_2^1,\ t\}}-1\right) \dd\mu(t).$$
Except for particular values of $t$ such as $t=1$, this integral cannot be computed in elementary functions because of the contribution of the interval $t\in \big[\min_i p_1^i/p_2^i, \ \max_i p_1^i/p_2^i\big]$ leading to the hypergeometric function. 
This integral can be computed approximately, provided that the change in prices is small. Instead of performing this computation, we explain below the general technique to approximate $\overline{W_\ev}$. For this purpose, we first discuss functionals for which the aggregate behavior is a sufficient statistic.
\end{example} 

For a general objective $W$ of the form~\eqref{eq_welfare_additive}, the use of the representative agent approach  is justified if the interval $[\underline{W},\overline{W}]$ is, in fact, a singleton, i.e.,
the convexification of $w$ coincides with the concavification. The two coincide only for affine functionals.
\begin{corollary}
Market demand is a sufficient statistic for the value of $W$ if the measure of individual well-being $w(\succsim)$ is an affine functional of the logarithmic expenditure function $\ln(\E_\succsim)$. If it is not affine, there is an aggregate preference $\succsim_\agr$ and two populations with the same total budget whose preferences aggregate up to $\succsim_\agr$ but levels of $W$ are different. 
\end{corollary}
The equivalent and compensating variations, as well as the total income of agents with preferences from $\D'\subset \D$, are not affine.
An example of an affine functional is given by 
the area variation
$$  \av_{\vec{p^0}\to\vec{p^1}}(\succsim,b)=  \int_{\vec{p^1}}^{\vec{p^0}}\label{eq_consume_surplus} D\big(\vec{p},\,b\big)\,\dd \vec{p},
$$
where the integration is over a curve connecting $\vec{p^1}$ and $\vec{p^0}$ in the space of prices.\footnote{The area variation is also known as the consumer surplus. It does not have a straightforward welfare interpretation except for the case of the quasilinear domain where $\av=\ev=\cv$. Beyond the quasilinear domain, $\av$ is often used as an approximation to the equivalent or compensating variation since $\av$ is easier to compute in practice thanks to the direct observability of the Marshallian  demand~$D$ \citep*{willig1976consumer}. Note that $\cv\leq \av\leq \ev$ for homothetic preferences which follows from the chain of inequalities $1-\frac{1}{t}\leq \ln t\leq t+1$ with $t=E(\vec{p^0})/E(\vec{p^1})$.}
Since the demand is proportional to the gradient of the logarithmic expenditure function~\eqref{eq_demand_price_index}, we get
\begin{equation}\label{eq_welfare_consume_surplus}
	\av_{\vec{p^0}\to\vec{p^1}}(\succsim,b)= b\cdot\big(\ln \E(\vec{p^1})-\ln \E(\vec{p^0}) \big),
\end{equation}
and thus $\av$ is indeed an affine functional of the logarithmic expenditure function.
\smallskip 

Let us discuss the behavior of the welfare change provided that the change in prices $\delta\vec{p}=\vec{p^1}-\vec{p^0}$ is small. We focus on the case where the welfare change is measured by the equivalent variation~\eqref{eq_welfare_equivalent_variation}. Let us first derive an approximate formula for the equivalent variation $\ev(\succsim,b)$ of a single agent. Denote $\ln E(\vec{p}^1)-\ln E(\vec{p}^0)$ by $\delta\ln E$.
Using the Taylor formula, we get
\begin{equation}\label{eq_EV_approximation}
\ev(\succsim,b)=b\cdot (\exp(-\delta\ln E)-1)\simeq b\cdot\left(-\delta \ln E+\frac{1}{2}(\delta \ln E)^2\right)=\av(\succsim,b)+\frac{b}{2}(\delta \ln E)^2,
\end{equation}
where $\simeq$ means equality up to the third order of magnitude in $\delta\vec{p}.$\footnote{Formally, $\ev-\av- \frac{b}{2}(\delta \ln E)^2=O(||\delta\vec{p}||^3)$ as $||\delta\vec{p}||\to 0$.} Approximating the small change in the argument of $\ln E$ via the gradient and expressing the gradient through the demand~\eqref{eq_demand_price_index}, we rewrite
$$b\cdot \big(\delta\ln E\big)^2\simeq b\cdot \left\langle\delta\vec{p},\, \nabla\ln E(\vec{p^0}) \right\rangle^2= b\cdot \left\langle\delta\vec{p},\, D\left(\vec{p^0},\, 1\right) \right\rangle^2= \left\langle\delta\vec{p},\, D\left(\vec{p^0},\, \sqrt{b}\right) \right\rangle^2.$$
 Taking into account that the area variation is an affine functional, we get the following corollary of Proposition~\ref{prop_welfare_as_persuasion}.

\begin{corollary}\label{cor_approximation}
If individuals have preferences from a domain~$\D$,
the welfare change is measured by the equivalent variation, and the change in prices $\delta\vec{p}$ is small, then 
$$\overline{W_\ev}\simeq \av_{\vec{p^0}\to\vec{p^1}}(\succsim_\agr,B)+\frac{1}{2} \cdot\cav_{\L_D}\left[\left\langle\delta\vec{p},\, D\left(\vec{p^0},\, \sqrt{B}\right) \right\rangle^2\right](\succsim_\agr).$$
\end{corollary}
Formula~\eqref{eq_EV_approximation} allows one to express $\av$ through $\ev$. Hence, Corollary~\ref{cor_approximation} implies
\begin{equation}\label{eq_welfare_range_small_change}
    \overline{W_\ev}-\underline{W_\ev}\simeq\frac{1}{2}\cdot\Big(\cav_{\L_D}\left[\left\langle\delta\vec{p},\, D\left(\vec{p^0},\, \sqrt{B}\right) \right\rangle^2\right](\succsim_\agr)-\left\langle\delta\vec{p},\, D_\agr\left(\vec{p^0},\, \sqrt{B}\right) \right\rangle^2\Big). 
\end{equation}
One can check that formula~\eqref{eq_welfare_range_small_change} also remains valid if the welfare change is measured by the compensating variation.

Note that the concavification $\cav_X[f](x)$  can be seen as the maximal expectation of $\mathbb{E}_{y}[f(y)]$ with $y\sim \mu$ over all distributions $\mu$ on $X$ such that $\mathbb{E}_{y}[y]=x$ with a finite number of atoms. Based on that, formula~\eqref{eq_welfare_range_small_change} admits a simple probabilistic interpretation summarized by the next corollary.
\begin{corollary}\label{cor_range}
Assume that individuals have preferences from~$\D$,
the welfare change is measured by either $\ev$ or $\cv$, and the change in prices $\delta\vec{p}$ is small. The range of possible values for the change in welfare can be expressed through the variance of demand as follows:
\begin{equation}\label{eq_range_as_variance}
    \overline{W}-\underline{W}\simeq\frac{1}{2}\cdot \sup_{(\succsim_k,b_k)_{k=1,\ldots}}\mathrm{Var}_\succsim\, \Big[\left\langle\delta\vec{p},\, D_\succsim\left(\vec{p^0},\, \sqrt{B}\right) \right\rangle\Big],
  \end{equation}  
    where the supremum is taken over all populations with preferences from $\D$ and given aggregate behavior $\succsim_\agr$, and $\succsim$ is a random preference equal to $\succsim_k$ with probability $\beta_k$.
\end{corollary}
An agent with budget $\sqrt{B}$ cannot buy more than $\sqrt{B}/p_i^0$ units of a good $i$ at prices $\vec{p^0}$. Hence, the range for small $\delta\vec{p}$ is bounded from above by
$$\overline{W}-\underline{W}\lesssim\frac{B}{2}\cdot \max_i \left(\frac{\delta p_i}{p_i^0}\right)^2.$$
We conclude that both for the equivalent and compensating variations, the range of values $\underline{W}-\overline{W}$ has the second order of magnitude in $\vec{p^1}-\vec{p^0}$ when the change in prices is small.
\begin{example}\label{ex_CES_welfare_approx}
Consider the setting from Example~\ref{ex_CES_welfare}: a CES aggregate consumer~\eqref{eq_CES_two_goods} and individuals having preferences from the domain~$\D$ of substitutes. Our goal is to compute the range $\overline{W}-\underline{W}$ explicitly provided that $\delta\vec{p}=\vec{p^1}-\vec{p^0}$ is small. 

Note that $\left\langle\delta\vec{p},\, D_\succsim\left(\vec{p^0},\, \sqrt{B}\right) \right\rangle^2$ is a convex functional on logarithmic expenditure functions. Hence, as in Example~\ref{ex_CES_welfare}, maximal value in~\eqref{eq_welfare_range_small_change} and~\eqref{eq_range_as_variance} is attained when $\succsim_\agr$ is generated by a population of linear consumers. The distribution of $\mrs=v_1/v_2$ over the population is given by~\eqref{eq_CES_CDF}. A linear agent with $\mrs=t$ spends all her money on the first good if $\frac{p_1^0}{p_2^0}<t$ and on the second if the inequality is reversed. 
Hence,
$$\left\langle\delta\vec{p},\, D_\succsim\left(\vec{p^0},\, \sqrt{B}\right) \right\rangle=\sqrt{B}\left(\frac{\delta p_1}{p_1^0}\cdot \one_{\{{p_1^0}/{p_2^0}<t\}}+\frac{\delta p_1}{p_1^0}\one_{\{{p_1^0}/{p_2^0}\geq t\}}\right).$$
Thus the variance is equal to the variance of a random variable that equals ${\sqrt{B}\delta p_1}/{p_1^0}$ or ${\sqrt{B}\delta p_2}/{p_2^0}$ with probabilities
$$\gamma_1=\mu\big((p_1^0/p_2^0,\, \infty)\big)=\frac{(a_1\cdot p_2^0)^{\sigma-1}}{(a_1\cdot p_2^0)^{\sigma-1}+ (a_2\cdot p_1^0)^{\sigma-1}}\qquad \mbox{or}\qquad \gamma_2=1-\gamma_1.$$
We obtain the following formula
$$\overline{W}-\underline{W}\simeq \frac{B}{2}\left(\left(\frac{\delta p_1}{p_1^0}\right)^2\gamma_1+\left(\frac{\delta p_2}{p_2^0}\right)^2\gamma_2-\left(\frac{\delta p_1}{p_1^0}\gamma_1+\frac{\delta p_2}{p_2^0}\gamma_2\right)^2   \right) $$
for the range of welfare change. It is applicable both to the equivalent and compensating variations. 
\end{example}



\section{Invariant domains}\label{sec_invariant}

In this section, we study domains of homothetic preferences invariant with respect to aggregation: if each consumer's preference belongs to the domain, so does the aggregate preference. Tools developed in the previous section
reduce invariance to the convexity of the set of logarithmic expenditure functions and yield a flexible procedure for constructing invariant domains.

\begin{definition}
A domain $\D$ of homothetic preferences over $\R_+^n$ is invariant with respect to aggregation if for any $m\geq 2$ and any population of $m$ consumers with preferences $\succsim_k\ \in\,\D$ and budgets $b_k\in\R_{++}$, $k=1,\ldots,m$, the aggregate preference $\succsim_\agr$ also belongs to $\D$. 
\end{definition}
The set of all homothetic preferences and a domain containing just one preference $\D=\{\succsim\}$ are elementary examples of invariant domains. 

Note that it is enough to check the condition of invariance for populations of $m=2$ consumers. Indeed, aggregation for a population of $m>2$ consumers reduces to aggregation for pairs by adding consumers one by one sequentially. Hence, if the outcome of aggregation belongs to the domain for any pair, the outcome will belong to this domain for any population.


\smallskip

Recall that $\L_\D$ is  the set of all logarithmic expenditure functions of preferences from $\D$: 
$\L_\D=\Big\{f\ :\  \R_{++}^n\to\R\quad:\ f=\ln  \E_{\succsim},\ \succsim\ \in\ \D\Big\}.$\footnote{To be precise, each preference $\succsim\in \D$ corresponds to a family of functions from $\L_D$ that differ by an additive constant. }
The following result is a direct corollary 
of Theorem~\ref{th_representative_index}.
\begin{corollary}\label{cor_invariant}
A domain $\D$ is invariant with respect to aggregation if and only if the set of logarithmic expenditure functions $\L_\D$ is a convex set of functions on $\R_{++}^n$. 
\end{corollary}
In other words, $\D$ is invariant whenever, for any pair of preferences $\succsim',\ \succsim''\ \in \  \D$ with expenditure functions $\E'$ and $\E''$ and $\lambda\in(0,1)$,  the preference~$\succsim$ whose expenditure function $\E$ is given by 
\begin{equation}\label{eq_requirement_invartiance}
\ln \E=\lambda\cdot \ln \E'+ (1-\lambda)\cdot \ln \E''
\end{equation}
also belongs to $\D$. 
For example, the domain of Cobb-Douglas preferences~\eqref{eq_Cobb_Douglas_utility} satisfies the requirement~\eqref{eq_requirement_invartiance} and, hence, is invariant. The domains of preferences exhibiting substitutability or complementarity are also invariant. Indeed, expenditure shares can be obtained by differentiating logarithmic expenditure functions~\eqref{eq_budget_shares_as_elasticities} and so the monotonicity conditions defining these domains are preserved under convex combinations.

\medskip 

Corollary~\ref{cor_invariant} more than just characterizes invariant domains in geometric terms; it also provides a handy tool to construct invariant domains containing a given domain.
Suppose $\D$ is not invariant. How can we complete it to an invariant domain? Of course, $\D$ is contained in the domain of all homothetic preferences, which is invariant. To exclude such a trivial answer, we require the completion to be minimal with respect to set inclusion. 
\begin{definition}\label{def_invariant_completion}
For a domain $\D$, its completion $\D^\inv$ is the minimal closed domain that is invariant with respect to aggregation and contains $\D$.
\end{definition}
The closure is defined with respect to the metric structure~\eqref{eq_distance_body} on preferences. The closedness assumption helps to get a tractable answer for infinite domains. As we will see, taking closure is equivalent to enriching $\D^\inv$ by aggregate preferences of non-atomic populations with preferences from $\D$. 

The completion $\D^\inv$ exists since it can be obtained as the intersection of all closed invariant domains containing $\D$, and there is at least one such domain, namely, the domain of all homothetic preferences. Corollary~\ref{cor_invariant} implies the following geometric characterization of $\D^\inv$.
\begin{corollary}\label{cor_convex_hull}
For any domain $\D$ of homothetic preferences, its  completion $\D^{\inv}$ is equal to the set of all preferences corresponding to logarithmic expenditure functions from the closed convex hull of $\mathcal{L}_{\D}$:
$$\D^{\inv}=\left\{ \ \succsim \ \  : \ \ \ln (\E_\succsim)\in \conv\Big[\mathcal{L}_{\D}\Big]\right\},$$
where $\conv[X]$ denotes the minimal closed convex set containing $X$.
\end{corollary}
This corollary assumes that the choice of the topology on preferences is aligned with that on logarithmic expenditure functions. This requirement is satisfied by the topology from  Appendix~\ref{app_topology}.

Note that $\conv[X]$ can be obtained as the closure of the set of all convex combinations of finite collections of elements from~$X$.
For a finite domains $\D=\{\succsim_1,\ldots,\succsim_q\}$, looking at combinations of at most $q=|\D|$ elements is enough and, hence, Corollary~\ref{cor_convex_hull} is especially easy to apply. For such $\D$, the  completion $\D^\inv$  consists of all preferences $\succsim$ with expenditure functions of the form
$\ln \E(\vec{p})=\sum_{k=1}^q t_k\cdot \ln \E_{k}(\vec{p})$
with $\vec{t}\in \Delta_{q-1}$. Reinterpreting Example~\ref{ex_CobbDouglas_as_aggregation_of_linear}, we conclude that  Cobb-Douglas preferences over $n$ goods is the  completion of  $\D=\big\{\succsim_1,\ldots,\succsim_n\big\}$ where $\succsim_i$ corresponds to the utility function $u_{i}(\vec{x})=x_i$.
\smallskip

To compute the completion for infinite domains $\D$, 
we need to take the closure of the set of preferences $\succsim$ corresponding to all finite convex combinations of logarithmic expenditure functions 
$$\ln \E(\vec{p})=\sum_{k=1}^q t_k\cdot \ln \E_k(\vec{p}), $$
where $q\geq 1$, a vector $\vec{t}\in \Delta_q$, and $\E_k$ represents some preference $\succsim_k$ from $\D$.
It is convenient to think about this convex combination as a result of integration with respect to the atomic probability measure $\mu$, such that $\mu(\succsim_k)= t_k$:
\begin{equation}\label{eq_continuous_aggregation}
    \ln \E(\vec{p})=\int_\D
\ln \E_{\succsim'}(\vec{p})\,\dd\mu(\succsim').
\end{equation}
It turns out that taking closure is equivalent to allowing arbitrary probability measures $\mu$ in~\eqref{eq_continuous_aggregation}, not necessarily atomic. For parametric domains such as linear or Leontief preferences discussed below, the integral in~\eqref{eq_continuous_aggregation} can be seen as the integral over the space of parameters and, hence, passing to an arbitrary $\mu$ is straightforward. In Appendix~\ref{app_topology}, we explain how to define~\eqref{eq_continuous_aggregation} for any  domain $\D$ and measure $\mu$.
\begin{theorem}\label{th_continuous_aggregation}
The  completion of a domain $\D$ consists of all preferences $\succsim$ such that their expenditure function $\E$ can be represented as
\begin{equation}\label{eq_continuous_aggregation_proposition}
    \ln \E(\vec{p})=\int_{\overline{\D}}
\ln \E_{\succsim'}(\vec{p})\,\dd\mu(\succsim')
\end{equation}
with some Borel probability measure $\mu$ supported on the closure $\overline{\D}$ of $\D$.
\end{theorem}
With general $\mu$, representation~\eqref{eq_continuous_aggregation_proposition} can be interpreted as the result of preference aggregation where non-atomic populations are allowed and $\mu$ plays the role of preference distribution over the population. 
 In what follows, we refer to~\eqref{eq_continuous_aggregation_proposition} as continuous aggregation.

We prove Theorem~\ref{th_continuous_aggregation} (in a more general form) in Appendix~\ref{sec_th_continuous_aggregation_proof}, together with Theorem~\ref{th_choquet} formulated in the next section. Both results rely on Choquet theory, which
studies extreme points of compact convex sets in topological vector spaces \citep*{phelps2001lectures}. Using this theory requires careful choice of a topology and a measurable structure. For the proof to work, it is crucial that the sets of preferences and logarithmic expenditure functions endowed with the distance~\eqref{eq_distance_body} are compact and admit an isometric embedding into a Banach space.

\subsection{Completion of linear preferences and random utility models}\label{sec_ARUM}

Consider the domain $\D$ of all linear preferences. Our goal is to characterize its  completion~$\D^\inv$.

\medskip

By Theorem~\ref{th_continuous_aggregation}, the completion must contain all preferences $\succsim$ whose expenditure functions correspond to all finite convex combinations of logarithmic expenditure functions for linear preferences.
This convex combination can be seen as a result of integration with respect to an atomic probability measure $\mu$, such that  $\mu\big(\vec{v}_k\big)=\beta_k$:
\begin{equation}\label{eq_price_index_linear_aggregate_appendix}
\ln \E(\vec{p})=\sum_{k=1}^m \beta_k \ln\left(\min_{i=1,\ldots, n}\frac{p_i}{v_{k,i}} \right)=\int_{\R_+^n} \ln\left(\min_{i=1,\ldots, n}\frac{p_i}{v_{i}} \right)\,\dd\mu(\vec{v}).
\end{equation}

To better understand the meaning of \eqref{eq_price_index_linear_aggregate_appendix}, we now elaborate on the link between heterogeneous linear preferences and random utility models of discrete choice.\footnote{
The intuition of this link is as follows. By Theorem~\ref{th_continuous_aggregation}, finding the completion of a class of preferences is essentially taking the average of logarithmic expenditure functions with respect to some measure~$\mu$. One can interpret this average as the expectation over random preferences of a single decision maker, and expenditure shares can be interpreted as probabilities of choosing one of $n$ possible alternatives.} 

In the additive random utility model (ARUM), there is a single decision maker choosing between one of $n$ alternatives. Her utility for alternative $i$ is equal to $w_i+\varepsilon_i$, where $w_i$ is a deterministic component and $\varepsilon_i$ is a stochastic shock. The vector $\vec{w}=(w_1,\ldots, w_n)\in \R^n$ and the joint distribution of shocks $\boldsymbol{\varepsilon}=(\varepsilon_1,\ldots, \varepsilon_n)\in \R^n$ are given. For each realization of the shocks, the agent selects the alternative with the highest utility. Hence, the expected utility $U(\vec{w})$ of the decision maker and the probability $S_i(\vec{w})$ that she chooses alternative~$i$ are equal to\footnote{The formula for the choice probabilities holds for $\vec{w}$ such that the probability of a tie  $w_i+\varepsilon_i=w_j+\varepsilon_j$ is zero. This requirement is satisfied for Lebesgue almost all $\vec{w}$ no matter what the distribution of the shocks is.} 
\begin{equation}
\label{eq: ARUM}U(\vec{w})=\mathbb{E}\left[\max_{i=1,\ldots n} (w_i+\varepsilon_i)\right]
\quad\mbox{and}\quad 
S_i(\vec{w})= \mathbb{P}\{w_i+\varepsilon_i> w_j+\varepsilon_j\ \forall j\ne i\}, 
\end{equation}
where $\mathbb{E}$ and $\mathbb{P}$ denote the expectation and the probability w.r.t. the shock distribution. 
\begin{proposition}
\label{prop_linear_invariant_hull}
A preference $\succsim$ with an expenditure function $\E$ belongs to the  completion of the domain of all linear preferences over $n$ goods if and only if there is a joint distribution of the shocks $(\varepsilon_1,\ldots, \varepsilon_n)$ in ARUM \eqref{eq: ARUM} such that
\begin{equation}\label{eq_ARUM_utility_and_price_index}
U(\vec{w})=-\ln \big(\E(e^{-w_1},\ldots,e^{-w_n}\big)
\end{equation}
is the expected utility in ARUM with deterministic utilities $\vec{w}\in  \R^n$.
\end{proposition}

We prove Proposition~\ref{prop_linear_invariant_hull} in Appendix \ref{app_proof_linear_hull}.

Taking the gradient on both sides of~\eqref{eq_ARUM_utility_and_price_index} leads to a reformulation of Proposition~\ref{prop_linear_invariant_hull} in terms of expenditure shares:\footnote{The fact that the choice probabilities $S_i(\vec{w})$ can be obtained as partial derivatives of decision maker's utility is known as the Williams–Daly–Zachary theorem and its classical version requires regularity assumptions on the distribution of shocks \citep*{mcfadden1981econometric}. The possibility to drop all the assumptions and get the conclusion for Lebesgue almost all $\vec{w}$ is a recent result \citep*{sorensen2022mcfadden}. The connection between ARUM and aggregation of linear preferences makes this result a corollary of general formula~\eqref{eq_budget_shares_as_elasticities} expressing expenditure shares as the gradient of logarithmic expenditure functions for almost all prices.} $\succsim$ is in the   completion of linear preferences whenever $\vec{s}(e^{-w_1},\ldots,e^{-w_n}\big) $ is the vector of choice probabilities for some additive random utility model, i.e., there exists a distribution of shocks such that
\begin{equation}\label{eq_budget_shares_and_choice_probabilities_linear}
s_{i}(e^{-w_1},\ldots,e^{-w_n}\big)=\mathbb{P}\{w_i+\varepsilon_i> w_j+\varepsilon_j\ \forall j\ne i\}    
\end{equation} for all $i=1,\ldots, n$ and Lebesgue almost all $\vec{w}\in\R^n$.
\medskip

Proposition \ref{prop_linear_invariant_hull} establishes the equivalence between the completion of the linear preference domain and the class of demand systems generated by discrete-choice models. Combining this with Theorem~\ref{th_representative_index} implies the result obtained independently (and using different methods) by \cite*{dube2022discrete}: every demand system generated by an ARUM $v_i = -\ln p_i +\varepsilon_i$ allows a representative consumer.\footnote{See Theorem 1 in \citep*{dube2022discrete}.}
\medskip

The class of vector functions that can arise as choice probabilities $\vec{S}(\vec{w})$ for some ARUM is well-studied in the discrete choice theory. We need the following necessary condition applicable to smooth vector functions. \fed{Reference?} For any ARUM with $n$ alternatives and any subset of distinct alternatives $i,j_1,j_2,\ldots,j_q$ with $q\leq n-1$, the following inequality holds 
$$
\frac{\partial^q {S}_i(\vec{w})}{\partial w_{j_1}\partial w_{j_2}\ldots\partial w_{j_q}}\cdot(-1)^q\leq0$$
at any $\vec{w}$ where $\vec{S}$ is $q$ times differentiable \citep*{anderson1992discrete}.\fed{Philip, please, add a reference \& CHECK correctness: in the old version we had $q\leq n$.} Taking into account the connection between expenditure shares and choice probabilities~\eqref{eq_budget_shares_and_choice_probabilities_linear} and the identity $\frac{\partial}{\partial w_i}=-\frac{\partial}{\partial \ln p_i}$ for $p_i=e^{-w_i}$, we obtain the following corollary of Proposition~\ref{prop_linear_invariant_hull}.
\begin{corollary}\label{corollary_ARUM}
If a preference $\succsim$ belongs to the completion of the domain of all linear preferences, then its expenditure shares  satisfy the following inequalities
\begin{equation}\label{eq_budget_shares_necessary_linear}
\frac{ \partial s_i(\vec{p})}{\partial \ln p_{j_1}\partial \ln p_{j_2}\ldots \partial \ln p_{j_q}} \geq 0
\end{equation}
for any distinct goods $i,j_1,j_2,\ldots,j_q$ with $q\leq n-1$ at any price vector $\vec{p}\in \R_{++}^n$ where $\vec{s}$ is differentiable~$q$ times.
\end{corollary}
For $q=1$, the condition~\eqref{eq_budget_shares_necessary_linear} becomes the substitutability condition~\eqref{eq:subst}. In other words,  any preference $\succsim$ from the completion of linear preferences exhibits substitutability among goods. This conclusion is not surprising as linear preferences exhibit substitutability, and aggregation respects this property. 

An important special case of the ARUM \eqref{eq: ARUM} is the multinomial logit (MNL), where the stochastic shocks $\boldsymbol{\varepsilon}$ follow an i.i.d.  type-1 extreme-value distribution (aka Gumbel distribution):

\begin{equation}
\label{eq: Gumbel}
\mathbb{P}\left\{\varepsilon_i<x\right\}=e^{-e^{-x/\gamma}},\quad\gamma>0.
\end{equation}

In this case, the ARUM \eqref{eq: ARUM} takes the following form:

\begin{equation}
\label{eq: MNL}U(\vec{w})=\frac{1}{\gamma}\ln\left(\sum_{i=1}^n e^{w_i/\gamma}\right)
\quad\mbox{and}\quad 
S_i(\vec{w})= \frac{e^{w_i/\gamma}}{\sum_{j=1}^ne^{w_j/\gamma}}\quad \forall i=1,2,\ldots,n.
\end{equation}

Combining \eqref{eq: MNL} with \eqref{eq_ARUM_utility_and_price_index}, and setting $\sigma=1+\frac{1}{\gamma}$, we obtain a subclass of CES preferences~\eqref{eq_CES}: symmetric CES preferences with substitutes (i.e., satisfying $a_1 = a_2 = \ldots =a_n$ and $\sigma>1$). Hence, any symmetric CES with substitutes belongs to the completion of linear preferences.\footnote{The connection between CES and MNL has been pointed out by~\cite*{anderson1987ces}. The result on CES as an aggregation of linear preferences is ours.} For asymmetric CES with substitutes, the same conclusion can be obtained by introducing shifts into~\eqref{eq: Gumbel}. In contrast, from Corollary \ref{corollary_ARUM}, the CES preferences with complements ($\sigma<1$) clearly do not belong to the completion of linear preferences, because they violate the necessary conditions~\eqref{eq_budget_shares_necessary_linear}.

One might be tempted to think that the completion of the linear preference domain encompasses all homothetic preferences with substitutability. This fails to be true for $n\geq 3$ goods since the condition~\eqref{eq_budget_shares_necessary_linear}  imposes more restrictions than just substitutability. To show this, we provide an example of a preference over $n=3$ substitutable goods which cannot be obtained by aggregating linear preferences.\footnote{This example is a special case of Example~4 in \citep*{matsuyama2017beyond}.}
\begin{example}Consider the following expenditure function over $n=3$ goods:
\begin{equation}
\label{eq:E3}
\E(\vec{p}) = \left(p_1+p_2+p_3\right)^{1-\alpha}\left(p_1^{1/3}\cdot p_2^{1/3}\cdot p_3^{1/3}\right)^\alpha,
\end{equation}
with $1<\alpha<3/2$. We prove in Appendix \ref{app_exp} that \eqref{eq:E3} indeed defines an expenditure function. The corresponding expenditure share of good $i=1,2,3$ is given by
\begin{equation}
\label{eq:s3}
s_i(\vec{p})=\frac{\partial\ln \E(\vec{p})}{\partial\ln p_i}=\frac{\alpha}{3} + (1-\alpha)\frac{p_i}{p_1+p_2+p_3}.
\end{equation}

The restriction $\alpha<3/2$ guarantees that $s_i(\vec{p})>0$ for all price vectors, while the restriction $\alpha>1$ guarantees substitutability. Yet, \eqref{eq:E3} is not in the completion of the linear preference domain, as it violates \eqref{eq_budget_shares_necessary_linear}:
$$
\frac{\partial s_1(\vec{p})}{\partial\ln p_2\, \partial\ln p_3}=-2(\alpha-1)\frac{p_{1}\cdot p_{2}\cdot p_{3}}{(p_{1}+p_{2}+p_{3})^{3}}<0.
$$

\end{example}

\begin{corollary}\label{cor_linear_give_proper_subset_of_S}
For $n\geq 3$ goods, the completion of the domain of all linear preferences is a proper subset of the domain of preferences exhibiting substitutability.
\end{corollary}

The corollary tells nothing about the case of $n=2$ goods, which turns out to be an exception. 
\begin{proposition}\label{prop_linear_gwnerate_S_for_2_goods}
For $n=2$ goods, the completion of the domain of all linear preferences coincides with the set of all preferences exhibiting substitutability.
\end{proposition}
\begin{proof} 

This result follows from an explicit construction. 
Given $\succsim$ such that the expenditure share of the first good $s_{1}(p_1,p_2)= \frac{\partial \ln \E(p_1,p_2)}{\partial \ln p_1}$ is non-decreasing in $p_2$, we need to
find a distribution $\mu$ of value vectors $\vec{v}$ so that the continuous aggregation~\eqref{eq_price_index_linear_aggregate_appendix} of linear preferences leads to~$\succsim$.

To guess an explicit formula for such $\mu$, take the partial derivative $\frac{\partial }{\partial \ln p_1}$ on both sides of~\eqref{eq_price_index_linear_aggregate_appendix}:
\begin{equation}\label{eq_mu_and_budget_shares}
s_{1}(p_1,p_2)=\mu\left(\left\{\frac{v_1}{v_2}\geq \frac{p_1}{p_2} \right\}\right).
\end{equation}
The derivative exists and the identity holds for Lebesgue almost all $(p_1,p_2)$. 
The ratio $\mrs=v_1/v_2$ is the marginal rate of substitution for the corresponding linear preference. We conclude that $1-s_{1}(\,\cdot\, ,1)$ must be the cumulative distribution function of $\mrs$ and the monotonicity of $s_{1}$ makes this possible. Choosing  any such distribution $\mu$ and adding atoms of the weight $1-\lim_{p_1\to +0} s_{1}(p_1,1)$ at $\vec{v}=(0,1)$ and of the weight $\lim_{p_1\to \infty} s_{1}(p_1,1)$ at $\vec{v}=(1,0)$ completes the construction.

\end{proof}

Observe that 
the valuations $\alpha\cdot (v_1, v_2)$ define the same linear preference for any scaling factor $\alpha>0$, i.e., the preference is determined by the ratio $v_1/v_2$. 
Hence, the distribution of linear preferences reproducing a given preference over two substitutes is unique and pinned down by~\eqref{eq_mu_and_budget_shares}. We obtain the following corollary of Proposition \ref{prop_linear_gwnerate_S_for_2_goods}.

\begin{corollary}\label{cor_MRS_reconstruction}
Any preference over two goods exhibiting substitutability can be represented as a continuous aggregation of linear preferences~\eqref{eq_price_index_linear_aggregate_appendix}. The distribution of linear preferences over the population corresponding to $\succsim$ is pinned down uniquely and admits an explicit formula: the cumulative distribution function of the marginal rate of substitution $\mrs={v_1}/{v_2}$ equals $1-s_{1}(\,\cdot\, ,1)$.
\end{corollary}

\begin{example}[translog preferences]
\label{ex_translog}
To illustrate Corollary~\ref{cor_MRS_reconstruction}, let us show how a family of consumers with linear preferences over two goods can generate translog preferences, a popular class of homothetic preferences obtained as a perturbation of Cobb-Douglas preferences in the space of expenditure functions \citep[p.139]{Diewert1974translog}.  A preference $\succsim$ is translog  
if its logarithmic expenditure function has the following form\footnote{\label{footnote_translog}It is important to distinguish a translog preference and a translog demand system defined by $D(\vec{p},B)=B\cdot \nabla f(\vec{p})$, where $f(\vec{p})=\alpha \ln p_1 + (1-\alpha)\ln p_2 - \frac{\beta}{2}\left(\ln\frac{p_1}{p_2}\right)^2$. This parametric family of demands enjoys an invariance property: the aggregate demand of a heterogeneous population has the same functional form. However, the translog demand system does not correspond to a preference of any rational individual: indeed, a simple computation shows that one of the demand components becomes negative  
if $\ln(p_1/p_2)$ is outside of the interval $\left[-\frac{1-\alpha}{\beta},\,\frac{\alpha}{\beta}\right]$.  To get a logarithmic expenditure function corresponding to a preference of a rational consumer, we need to use another formula for low and high price ratios as in~\eqref{eq_translog}. This adjustment comes at the cost of invariance: a weighted average of two distinct logarithmic expenditure functions of the form~\eqref{eq_translog} no longer has this form. In other words, translog preferences are not aggregation-invariant. 
Since translog preferences correspond to a translog demand system in a subset of the price space, one can talk about the ``local invariance'' of translog preferences. We leave a study of this weaker invariance notion relevant for the almost ideal demand system \citep*{deaton1980almost} and its extensions for future research.}
\begin{equation}\label{eq_translog}
\ln \E(p_1,p_2) = 
\begin{cases}
\ln p_1,& \ln \left(\frac{p_{1}}{p_{2}}\right)< -\frac{1-\alpha}{\beta}\\
\alpha \ln p_1 + (1-\alpha)\ln p_2 - \frac{\beta}{2}\left(\ln\frac{p_1}{p_2}\right)^2,&  -\frac{1-\alpha}{\beta}\leq\ln \left(\frac{p_{1}}{p_{2}}\right)\leq\frac{\alpha}{\beta}\\
\ln p_2,&\ln \left(\frac{p_{1}}{p_{2}}\right) > \frac{\alpha}{\beta}
\end{cases},
\end{equation}
where $\alpha\in(0,1)$ and $\beta>0$.


Using~\eqref{eq_price_index_and_shares_for_linear}, we obtain that a distribution of value vectors $\vec{v}=(v_1,v_2)$ aggregates up to the translog preference if and only if 
$$
\ln\left(\frac{v_1}{v_2}\right)\sim \mathbb{U}\left(\left[-\frac{1-\alpha}{\beta},\ \frac{\alpha}{\beta}\right]\right),
$$
where $\mathbb{U}\left(\left[c,d\right]\right)$ denotes the uniform distribution supported on $[c,d]$.\footnote{Curiously enough, this is equivalent to the marginal rate of substitution following the so-called Benford law of digit bias \citep*{benford1938law}.} 

Consider a particular case of a translog preference $\succsim$ with  $\alpha = \beta = 1/2$. We obtain that $\succsim$ can be represented as aggregation over the continuous population of consumers distributed uniformly on $[-1,1]$ so that consumer $\alpha\in[-1,1]$ has utility  $u(x_1,x_2)=e^\frac{\alpha}{2}\cdot  x_1 + e^{-\frac{\alpha}{2}} \cdot x_2$. The corresponding identity~\eqref{eq_price_index_linear_aggregate_appendix} takes the following form:
$$\ln \E(p_1,p_2)=\int_{-1}^1 \ln\left(\min\left\{p_1\cdot e^{-\frac{\alpha}{2}}, \ p_2\cdot e^{\frac{\alpha}{2}} \right\}\right)\,\dd\alpha.$$
\end{example}


\begin{example}[CES with substitutes as an aggregation of linear preferences]\label{ex_logit_as_linear}
 Consider a CES preference $\succsim$ over two substitutes: 
$$
u(\vec{x})=\left(\left(a_1\cdot x_1\right)^{\frac{\sigma-1}{\sigma}}+\left(a_2\cdot x_2\right)^{\frac{\sigma-1}{\sigma}}\right)^{\frac{\sigma}{\sigma-1}}\qquad \mbox{with}\qquad \sigma>1.
$$
Using~\eqref{eq_price_index_and_shares_for_linear} and ~\eqref{eq_price_index_and_shares_for_CES}, we conclude that $\ln(v_1/v_2)$ in a population of consumers with linear preferences generating $\succsim$ follows a logit-type distribution: 
\begin{equation}\label{eq_CES_as_aggregation_of_linear}
\mu\left(\left\{\ln\left(\frac{v_1}{v_2}\right)<t\right\}\right)=\frac{A\cdot e^{(\sigma-1)t}}{A\cdot e^{(\sigma-1)t}+1}\qquad\text{with}\qquad A=\left(\frac{a_1}{a_2}\right)^{1-\sigma}.
\end{equation}
\end{example}

\subsection{Complete monotonicity and the completion of Leontief preferences}

We saw that, for $n=2$ goods, the completion of linear preferences is the whole domain of  preferences exhibiting substitutability. By contrast, the completion of all Leontief preferences turns out to be substantially narrower than the domain of all preferences with complementarity, even for  $n=2$.

By Theorem~\ref{th_continuous_aggregation} and formula~\eqref{eq_price_index_and_shares_for_Leontief} for  expenditure functions of Leontief preferences, the completion of Leontief preferences over $n\geq 2$ goods is the set of all preferences $\succsim$ with expenditure functions $\E$ of the following form:
\begin{equation}
\label{eq_price_index_Leontief_aggregate}
\ln \E(\vec{p})=\int_{\R_+^n} \ln\langle\vec{v},\vec{p}\rangle\,\dd\mu(\vec{v})
\end{equation}
for some probability measure $\mu$ on $\R_+^n$. 

Note that $\ln\langle\vec{v},\vec{p}\rangle$ is an infinitely differentiable function of $\vec{p}\in \R_{++}^n$, \fed{In the appendix, we check that one can} hence, so is $E(\mathbf{p})$ in~\eqref{eq_price_index_Leontief_aggregate}.
Thus, any preference whose expenditure function is not infinitely differentiable is beyond the completion of Leontief preferences. This implies the following corollary.
\begin{corollary}\label{cor_Leontief_does_not_give_all_complements}
For any number $n\geq 2$ of goods, the completion of the domain of Leontief preferences does not encompass the whole domain of preferences exhibiting complementarity.
\end{corollary}
\fed{add a proof}

To illustrate, we provide a specific example of a preference  over two complements that is beyond the 
completion of Leontief preferences.
\begin{example}[A preference over two complements outside of the completion of Leontief]

It suffices to find a preference such that the expenditure share of the first good $s_{1}(p_1,p_2)=\frac{\partial \ln \E(p_1,p_2)}{\partial \ln p_1}$ has a discontinuous derivative. 
The existence of such preference $\succsim$ follows from the characterization of expenditure shares~\eqref{eq_budget_share_characterization_for_two}. To describe $\succsim$ explicitly,
consider the following utility function:
\begin{equation*}
    u(x_1,x_2)=\min\left\{\sqrt{x_1\cdot x_2},\ \  x_1\right\}.
\end{equation*}
A consumer whose preferences are like that behaves as if she switches between Cobb-Douglas and Leontief preferences at $p_1=p_2$. A simple computation gives the expenditure share:
$$s_{1}(p_1,p_2)=\begin{cases}
\frac{1}{2}, & \frac{p_1}{p_2}<1\\
\frac{p_1}{p_1+p_2},& 1\leq \frac{p_1}{p_2}
\end{cases}.$$
As we see, $s_{1}$ is decreasing in $p_2$ and, hence, $\succsim$ exhibits complementarity. However, this preference cannot be obtained by aggregating Leontief preferences~\eqref{eq_price_index_Leontief_aggregate} since the expenditure share has a discontinuous derivative.\footnote{For complements, expenditure shares can have discontinuous derivatives but are necessarily continuous themselves. This is a simple corollary of~\eqref{eq_budget_share_characterization_for_two}. \fed{add details} By contrast, expenditure shares for substitutes can be discontinuous, e.g., for linear preferences, expenditure shares are step functions.}
\end{example}

The condition that a preference $\succsim$ belongs to the completion of Leontief preferences is substantially more restrictive than the requirement of smoothness of the expenditure function. An infinitely smooth function $f=f(\lambda)$, $\lambda\in\R_{++}$, is called completely monotone if$^{\,}$\footnote{In economics, completely monotone functions arise as the dependence of decision maker's payoff on the discount factor $\lambda$. \fed{Add reference}  Indeed, by Bernstein's theorem, a function $f$ is completely monotone if and only if $f(\lambda)=\int_{\R_+} e^{-\lambda t} \dd\nu (t)$ for some positive measure $\nu$ \citep[Theorem 1.4]{schilling2012bernstein}, i.e., $f$ is the expected utility of a risk-neutral decision maker with geometric discounting for a stream of payoff $\nu$. 

Utility functions with mixed risk aversion ($u(x)=\int_{\R_+} (1-e^{-\lambda x})\dd\nu(x)$) provide another economic context for  completely monotone functions  \citep*{caballe1996mixed}.
}
$$(-1)^k\cdot \frac{\dd^k }{\dd \lambda^k} f\geq 0\quad \mbox{for all $k=0,1,2,\ldots$}.$$
Complete monotonicity provides a necessary condition for $\succsim$ to be in the completion of Leontief preferences.
\begin{proposition}\label{prop_complete_monotonicity}
If a preference $\succsim$ belongs to the completion of Leontief preferences, then the demand  $D_{i}(\vec{p}, b)$ for good~$i$ is a completely monotone function of $p_i$ for each $i=1,\ldots, n$ and $b>0$.
\end{proposition}

\begin{proof}The result follows from the integral representation~\eqref{eq_price_index_Leontief_aggregate} of the expenditure function $\E$ and the possibility to 
interchange   differentiation with respect to $p_i$ and integration. The derivatives of the integrand in~\eqref{eq_price_index_Leontief_aggregate} can be computed explicitly
$$\frac{\partial^{k+1}}{\partial^k p_i}\ln\langle\vec{v},\vec{p}\rangle=(-1)^k\frac{v_i^k\cdot  k!}{\left(\langle\vec{v},\vec{p}\rangle\right)^{k+1}}.$$
Hence,
$$(-1)^k\frac{\partial^{k+1}}{\partial^k p_i}\ln \E(\vec{p})= k!\int_{\R_+^n}\frac{v_i^k}{\left(\langle\vec{v},\vec{p}\rangle\right)^{k+1}}\,\dd\mu(\vec{v}) \geq 0.$$
Since $D_i(\vec{p},b)=b\cdot\frac{\partial}{\partial p_i}\ln \E(\vec{p})$ by formula~\eqref{eq_demand_price_index}, we conclude that $D_i(\vec{p},b)$ is a completely monotone functions of $p_i$.
\end{proof}

For $n=2$ goods, we are able to provide a simple criterion for a preference $\succsim$ to be in the completion of Leontief preferences. This criterion suggests that the necessary condition of complete monotonicity established in Proposition~\ref{prop_complete_monotonicity} is almost sufficient.

A function $f=f(\lambda)$, $\lambda\in  \R_{++}$, is called a Stieltjes function if it can be represented as
\begin{equation}\label{eq_Steltjes}
f(\lambda)=\int_{\R_+}\frac{1}{\lambda+z} \dd\nu(z)
\end{equation}
for some positive measure $\nu$ on $\R_+$.\footnote{These functions are omnipresent in various branches of mathematics such as probability theory, spectral theory, continued fractions, and potential theory \citep*{schilling2012bernstein}.} The Stieltjes functions are completely monotone functions that themselves can be obtained as the Laplace transform of a completely monotone \fed{add reference} density.\footnote{The integral operator on the right-hand side of~\eqref{eq_Steltjes} is known as the Stieltjes transform. It equals the Laplace transform applied to $\nu$ twice: $f(\lambda)=\int_{\R_+} e^{-\lambda t}\left(\int_{\R_+} e^{-tz}\dd \nu(z)\right)\dd t$. By Bernstein's theorem, the set of completely monotone functions is the set of all functions equal to the Laplace transform of positive measures \citep[Theorem 1.4]{schilling2012bernstein}.  Hence, Stieltjes functions are those completely monotone functions that are Laplace transforms of completely monotone ones.}
\begin{proposition}\label{prop_Stieltjes}
For $n=2$ goods, a preference $\succsim$ belongs to the completion of Leontief preferences if and only if the demand for the first good is given by ${D_{1}(\vec{p},b)}$ with $\vec{p}=(p_1,1)$ and $b=1$ is a Stieltjes function of the price $p_1$.
\end{proposition}
\begin{proof}
    By~\eqref{eq_price_index_Leontief_aggregate}, for any $\succsim$ in the completion,  
 we have
$$
\frac{\partial }{\partial p_1}\ln \E(p_1,1)=\int_{\R_+^2}\frac{v_1}{v_1p_1+v_2}\,\dd\mu(v_1,v_2)
=\int_{\R_+}\frac{1}{p_1+z}\dd\nu(z),$$
where $\nu$ is the distribution of ${v_2}/{v_1}$. Since ${D_{1}((p_1,1),1)}=\frac{\partial }{\partial p_1}\ln \E(p_1,1)$, we get 
\begin{equation}\label{eq_nu_and_mu_Stieltjes}
{D_{1}((\lambda,1),1)}=\int_{\R_+}\frac{1}{\lambda+z}\dd\nu(z)
\end{equation}
 and conclude that ${D_{1}((\lambda,1),1)}$ is a Stieltjes function for any $\succsim$ from the completion. In the opposite direction, if $f$ is a Stieltjes function, then $f(\lambda)={D_{1}((\lambda,1),1)}$ for a preference $\succsim$ from the completion of Leontief preferences corresponding to $\mu$ in~\eqref{eq_price_index_Leontief_aggregate}  such that the ratio $v_2/v_1$ is $\nu$-distributed.
\end{proof} 
\medskip

The right-hand side of~\eqref{eq_nu_and_mu_Stieltjes} is called the Stieltjes transform of $\nu$. The Stieltjes transform is invertible \citep[Chapter 2]{schilling2012bernstein}.
 Hence,  if $\succsim$ belongs to the completion of Leontief preferences, the demand determines the distribution $\nu$ satisfying~\eqref{eq_nu_and_mu_Stieltjes} uniquely.
As a result, the distribution of Leontief preferences over the population that leads to $\succsim$ is pinned down uniquely. Namely, the continuous aggregation of Leontief preferences~\eqref{eq_price_index_Leontief_aggregate} with distribution $\mu$ of $(v_1,v_2)$ such that the ratio ${v_2}/{v_1}$ is distributed according to $\nu$ gives $\succsim$.  Leontief preferences with the same ratio 
coincide and so the distribution of preferences corresponding to $\succsim$ is indeed unique.

The Stieltjes transform can be inverted explicitly using tools from complex analysis.
Before describing the tools, we give an example obtained with their help.
\begin{example}[CES with complements as an aggregation of Leontief preferences]\label{ex_CES_as_Leontief}
We show that any CES preference  $\succsim$ over $n=2$ complements~\eqref{eq_CES} belongs to the completion of Leontief preferences.

First, consider a particular case with the elasticity of substitution $\sigma=\frac{1}{2}$ and weights $\vec{a}=(1,1)$; the corresponding utility function is the harmonic mean. The utility and the demand for the first good are as follows: $$u(x_1,x_2)=\left(\frac{1}{x_1}+\frac{1}{x_2}\right)^{-1}\qquad \mbox{and}\qquad D_{1}((p_1,p_2),b)=\frac{b}{p_1+\sqrt{p_1\cdot p_2}}.$$
By Proposition~\ref{prop_Stieltjes}, finding a probability distribution $\nu$ on $\R_+$ such that the identity~\eqref{eq_nu_and_mu_Stieltjes} holds is enough to show that $\succsim$ is in the completion of Leontief preferences.
We end up with the following equation:
$$\frac{1}{\lambda+\sqrt{\lambda}}=\int_{\R_+}\frac{1}{\lambda+z}\dd\nu(z).$$
One can check that $\nu$ with a density $\varphi$ given by 
\begin{equation}\label{eq_density_CES_symmetric}
\varphi(z)=\frac{1}{\pi}\frac{1}{\sqrt{z}(1+z)}
\end{equation}
is a solution, hence, $\succsim$ is indeed in the completion. By taking any distribution $\mu$ of $\vec{v}=(v_1,v_2)$ such that $v_2/v_1$ is $\nu$-distributed (e.g., $v_1$ equals $1$ identically and $v_2$ has distribution $\nu$), we represent $\succsim$ via a continuous aggregation of Leontief preferences~\eqref{eq_price_index_Leontief_aggregate}.

The above analysis  extends to any CES preference over two complements. The utility function and the expenditure share have the form 
$$ u(x_1,x_2)=\left(\left(a_1\cdot x_1\right)^{\frac{\sigma-1}{\sigma}}+\left(a_2\cdot x_2\right)^{\frac{\sigma-1}{\sigma}}\right)^{\frac{\sigma}{\sigma-1}}\qquad\mbox{and}\qquad  D_{1}((p_1,p_2),b) =\frac{b}{p_1+  A\cdot p_1^{\sigma}\cdot p_2^{1-\sigma}},
$$
where $\sigma\in (0,1)$ and $A=(a_1/a_2)^{1-\sigma}$. The  distribution $\nu$ of ${v_2}/{v_1}$ has to solve the equation
$$\frac{1}{\lambda+A\cdot \lambda^{\sigma}}=\int_{\R_+}\frac{1}{\lambda+z}\dd\nu(z).$$
We find a solution using a guess-and-verify trick. One can show that $\nu$ with density 
\begin{equation}\label{eq_density_CES_general}
\varphi(z)=\frac{1}{\pi}\cdot \frac{\sin (\pi\sigma)}{A^{-1}\cdot z^{2-\sigma}+z\cdot \cos (\pi\sigma)+A\cdot z^\sigma}
\end{equation}
is a solution.  Formula~\eqref{eq_density_CES_symmetric} is a particular case of~\eqref{eq_density_CES_general} for $\sigma={1}/{2}$ and $a_1=a_2$.
 \end{example} 
One may wonder about the ``guess'' part of the guess-and-verify trick, which we used to derive formulas~\eqref{eq_density_CES_symmetric} and~\eqref{eq_density_CES_general}. It relies on the following observation from complex analysis. For any distribution $\nu$ over $\R_+$, its Stieltjes transform is defined not only for $\lambda\in\R_{++}$ but also for all complex values of $\lambda\in \C\setminus \R_-$, where $\C$ denotes the complex plane. Moreover, this extended Stieltjes transform of $\nu$ is an analytic function on $\C\setminus \R_-$. The values of this analytic continuation above and below the ``cut'' over the negative reals can be used to reconstruct $\nu$. The answer is given by the Stieltjes-Perron formula: if $f$ is the Stieltjes transform of a measure  $\nu$ with density $\varphi$, then
\begin{equation}\label{eq_StieltjesPerron}
\varphi(z)=\frac{1}{2\pi i }\cdot \lim_{\varepsilon\to 0} \left(f(-z+i\varepsilon)-f(-z-i\varepsilon)\right),
\end{equation}
where $i$ is the imaginary unit and $\varepsilon$ tends to zero from above.
\fed{Moreover, $\nu$ admits a density if and only if the limit in~\eqref{eq_StieltjesPerron} exists for all $z\in \R_+$. A general version of the Stieltjes-Perron formula applicable even if $\nu$ does not have a density can be found in. \fed{reference}} 

Combining the Stieltjes-Perron formula and Proposition~\ref{prop_Stieltjes}, we get the following corollary.
\begin{corollary}\label{cor_complex_analysis}
If a preference $\succsim$ over $n=2$ goods belongs to the completion of Leontief preferences, then the demand for the first good $D_{1}((p_1, 1),1)$ as a function of its price $p_1$ admits an analytic continuation to complex prices $p_1\in \C\setminus \R_+$. The function $\varphi$ given by~\eqref{eq_StieltjesPerron} for $$f(\lambda)=D_{1}((\lambda, 1),1)$$ is the density of the unique distribution of ${v_2}/{v_1}$ such that the continuous aggregation of Leontief preferences~\eqref{eq_price_index_Leontief_aggregate} gives $\succsim$.
\end{corollary}
Note that the analytic continuation is unique if exists. Hence, Corollary~\ref{cor_complex_analysis} can be used to check whether a given preference belongs to the completion of Leontief preferences. First, we check whether the expenditure share admits an analytic continuation. If it does, we compute a candidate for the distribution $\nu$ via the Stieltjes-Perron formula. Finally, we check that what we got is a probability distribution, and the expenditure share can be obtained as its Stieltjes transform. A preference passes the test if and only if it is in the completion. Example~\ref{ex_CES_as_Leontief} illustrated this approach.

\subsection{Application: Fisher markets, fair division, complexity, and bidding languages}\label{sec_Fisher}


Consider a population of consumers with budgets $b_1,\ldots, b_m$ and homothetic preferences $\succsim_1,\ldots,\succsim_m$ over $n$ goods, and let 
$\vec{X}\in \R_{++}^n$ be a fixed supply vector. This economy is known in the algorithmic economics literature as the Fisher market\footnote{Named after Irving Fisher, who introduced a hydraulic method for equilibrium price computation; see \citep*{brainard2005compute}.}  and is by far the most studied economy from a computational perspective \citep[Chapters 5 and 6]{nisanalgorithmic}. Since the classical works of \cite*{varian1974equity} and \cite*{hylland1979efficient},
Fisher markets and their modifications are used for fair allocation of private goods without monetary transfers, giving rise to a famous mechanism known as the competitive equilibrium with equal incomes (CEEI) or the pseudo-market mechanism \citep*{moulin2019fair, pycia2022pseudo}. \fed{rephrase}


 A collection of bundles $\vec{x}_1,\ldots, \vec{x}_m$ and a price vector $\vec{p}$ form a competitive equilibrium  (CE) of the Fisher market with preferences $\succsim_1,\ldots,\succsim_m$, budgets $b_1,\ldots, b_m$, and total supply $\vec{X}$ if
\begin{equation}\label{eq_CEEI}
\vec{x}_k\in D_{k}(\vec{p},b_k)\quad  \mbox{for each consumer $k$}\quad  \mbox{ and } \quad \sum_{k=1}^m \vec{x}_k=\vec{X},
\end{equation}
i.e., each consumer buys her preferred feasible bundle, and the market clears. We pinpoint that money has no intrinsic value and the Fisher market is equivalent to an exchange economy where each agent $k$ is endowed with the fraction $\beta_k=b_k/B$ of $\vec{X}$ where $B$ is the total budget.

One can think of a CE as an allocation mechanism: agents report their preferences, and the mechanism computes an equilibrium and allocates each agent her bundle $\vec{x}_k$. In this interpretation, budgets $b_k$ represent agents' entitlement to the goods in the bundle $\vec{X}$. The case of equal entitlements $b_1=\ldots=b_m$  (CEEI)  
is especially important. In this case, each agent selects her best bundle from the same budget set and, hence, the resulting allocation is envy-free in the sense that $\vec{x}_k\succsim_k \vec{x}_l$ for any pair of agents $k$ and $l$. 
Since any CE is Pareto optimal by the first welfare theorem, CEEI gives a simple recipe to provide strong fairness and guarantee efficiency. CEEI and its various modifications have been applied to rent division \citep*{goldman2015spliddit},  chores allocation \citep*{bogomolnaia2017competitive}, course allocation \citep*{budish2017course, kornbluth2021undergraduate, soumalias2022machine}, cloud computing \citep*{devanur2018new}, school choice \citep*{ashlagi2016optimal,he2018pseudo}, and other problems \citep*{echenique2021constrained}.

Despite its attractive properties, the popularity of CEEI remains limited as computing its outcome is a challenging problem. It is known that an equilibrium allocation $\vec{x}_1,\ldots, \vec{x}_m$ can be obtained via maximizing the Nash social welfare
\begin{equation}\label{eq_EisenbergGale}
\prod_{k=1}^m \left(\frac{u_{k}(\vec{x}_k)}{\beta_k}\right)^{\beta_k}
\end{equation}
over all bundles $\vec{x}_1,\ldots, \vec{x}_m$ such that $\sum_{k=1}^m \vec{x}_k=\vec{X}$. This result tightly related to the existence of an aggregate consumer was established by \cite*{eisenberg1959consensus} for linear preferences but holds for all homothetic preferences; see \citep*{shafer1982market}. 
Although the Eisenberg-Gale problem is convex, computing its solutions is not an easy task unless $n $ or $ m$ are small. Even in the benchmark case of linear preferences, algorithms with good theoretical performance have required more than a decade of research and dozens of papers using cutting-edge 
techniques; see, e.g., \citep*{devanur2002market, orlin2010improved, vegh2012strongly}. Developing algorithms with good performance in practice is critical for large-scale applications of Fisher markets---e.g., to fair recommender systems \citep*{gao2022infinite} and Internet ad markets \citep*{conitzer2022pacing}---but despite the recent progress this problem is yet to be solved.
\smallskip

We examine the question of finding a CE from the preference aggregation perspective. This perspective sheds light on why  computing a CE can be challenging in seemingly innocent domains such as linear preferences and helps to identify domains where computing a CE is easy.

To find a CE, it is enough to compute the vector of equilibrium prices $\vec{p}$. Once we know~$\vec{p}$,
each agent is allocated her demanded bundle $\vec{x}_k$ at these prices.\footnote{If utilities are not strictly concave, $D_{k}(\vec{p},b)$ may not be a singleton, e.g., in
the domain of linear preferences. Even in this case, once equilibrium prices are known, choosing bundles  $\vec{x}_k$ from each agent's demand so that $\sum_k \vec{x}_k=\vec{x}$ is a simple problem, which boils down to a maximum flow computation \citep*{devanur2002market, branzei2019algorithms}.} Thus the essence of computing a CE is finding a vector of prices $\vec{p}$ such that the market demand matches the supply. In other words, we need to find $\vec{p} $ such that \emph{the aggregate consumer's demand contains $\vec{x}$}. This simple observation, combined with our insights about the structure of aggregate preferences, has many implications.

Finding a CE for a population of consumers boils down to finding a CE for one aggregate consumer, and we know that aggregation is easier to handle in the space of logarithmic expenditure functions. Recall that the demand is proportional to the gradient of the logarithmic expenditure function~\eqref{eq_demand_price_index} and, hence, $\vec{p}$ is an equilibrium price vector if and only if$^{\,}$\footnote{For expenditure functions that are not smooth the gradient is to be replaced with the superdifferential.}
$$B\cdot\nabla \ln \E_{\agr}(\vec{p})=\vec{X},$$
where $B$ is the total budget.
Interpreting this identity as the first order condition and taking into account the concavity of $\ln \E_{\agr}$, we conclude that $\vec{p}$ is a vector of equilibrium prices whenever 
\begin{equation}\label{eq_Shmurev_aggregate}
\vec{p} \quad  \mbox{ is the global maximum of } \qquad   \langle\vec{X},\vec{p}\rangle-B\cdot\ln \E_{\agr}(\vec{p}).
\end{equation}
Combining this result with   Theorem~\ref{th_representative_index}, we get the following proposition.
\begin{proposition}
A vector $\vec{p}$ is a vector of equilibrium prices for a population of consumers with homothetic preferences $\succsim_1,\ldots,\succsim_m$, budgets $b_1,\ldots, b_m$, and total supply $\vec{X}$ if and only if $\vec{p}$ is the global maximum of
\begin{equation}\label{eq_Shmurev}
\langle\vec{X},\vec{p}\rangle-\sum_{k=1}^m b_k \cdot\ln \E_{k}(\vec{p}).
\end{equation}
\end{proposition}
This optimization problem is convex. 
Its particular case for linear preferences has been known and obtained as the Lagrange dual to the Eisenberg-Gale optimization problem \citep*{cole2017convex, devanur2016new,shmyrev2009algorithm}. Our approach explains the preference-aggregation origin of this dual and provides a generalization to all homothetic preferences almost without any computations.
\smallskip 

Optimization problem~\eqref{eq_Shmurev_aggregate} indicates that 
to find a CE for a population of consumers with preferences from a certain domain $\D$ we must be able to find the market equilibrium for any preference that can be obtained as an aggregation of preferences from $\D$. In other words, the complexity of finding a CE is determined not by the domain $\D$ of individual preferences itself but rather by its completion $\D^\inv$. We illustrate this point for the domain of linear preferences over two goods.

By Proposition~\ref{prop_linear_gwnerate_S_for_2_goods}, aggregation of linear preferences over $n=2$ goods (domain $\D$) gives all preferences with substitutability ($\D^\inv$).
We will show that any algorithm computing an approximate CE for preferences from $\D$ can be used to compute an approximate CE for $\D^\inv$. Hence, finding CE for $\D$ cannot be easy if it is hard for $\D^\inv$.
Let us call $\vec{p}$ an $\varepsilon$-equilibrium price vector if there are $\vec{x}_k\in D_{k}(\vec{p},b_k)$, $k=1,\ldots, m$ such that $$\langle \vec{p},\vec{e} \rangle\leq  \varepsilon\cdot B,\qquad\mbox{where}\quad  e_i=\left|x_i-\sum_{k=1}^m x_{k,i}\right|,$$
i.e., the excess demand is relatively small compared to the total budget.
\begin{proposition}\label{prop_approx_algorithm}
Let $\D$ be the domain of linear preferences over two goods and assume we have an algorithm computing an $\varepsilon$-equilibrium price vector for any population of agents with preferences from $\D$.
 Then a $3\varepsilon$-equilibrium price vector for a population of $m$ agents with preferences from $\D^\inv$ can be computed by applying the algorithm as a black box to an auxiliary population with preferences from $\D$ and the number of agents of the order of $m/\varepsilon$.
\end{proposition}
We prove Proposition \ref{prop_approx_algorithm} in Appendix~\ref{app_approximate}. The idea is to approximate preferences from $\D^\inv$ by the aggregate preference of linear consumers so that the expenditure shares differ by at most $\varepsilon$. Such an approximation can be constructed via Corollary~\ref{cor_MRS_reconstruction} and requires a number of auxiliary linear consumers of the order of $1/\varepsilon$. As we show in Appendix~\ref{app_approximate}, if expenditure shares in two populations differ by at most $\varepsilon$, then $\varepsilon$-equilibrium price vector for one population is an $(1+n)\varepsilon$-equilibrium price vector for the other. Since $n=2$, an  $\varepsilon$-equilibrium price vector for the approximating population of linear consumers gives a $3\varepsilon$-equilibrium price vector for the population of consumers with preferences from $\D^\inv$.

\medskip 
The example of linear preferences demonstrates that even a simple parametric domain---if the choice of parameters is not aligned with aggregation---can have a large non-parametric completion. As a result, the simplicity of a parametric domain does not carry over to the aggregate behavior thus complicating the computation of a CE. 
To preserve the simplicity of a parametric domain, the choice of parameters is to be aligned with aggregation. Motivated by this concern, we consider computing a CE in parametric domains invariant with respect to aggregation.

Fix a finite family of ``elementary'' preferences $\succsim_1,\ldots, \succsim_q$ and consider the domain $\D=\D^\inv$ of all preferences that can be obtained by aggregating the elementary preferences.  We will call such invariant domains finitely-generated.\fed{Define finitely-generated domains earlier?} Cobb-Douglas preferences are an example of a finitely generated domain; see Example~\ref{ex_CobbDouglas_as_aggregation_of_linear}.
By Theorem~\ref{th_representative_index}, a finitely generated $\D$ consists of all preferences $\succsim$ whose expenditure function $\E$ can be represented as $\ln \E=\sum_{l=1}^q t_l \ln \E_{l}$ and, hence, the vector of coefficients $\vec{t}\in \Delta_{q-1}$ provides a parameterization of $\D$. 
\begin{proposition}
Consider a finitely-generated invariant domain $\D=\{\succsim_1,\ldots, \succsim_q\}^\inv$ and fix~$\varepsilon\geq 0$. Assume we have access to an algorithm finding an $\varepsilon$-equilibrium vector of prices for $m=1$ agent and using at most $T$ operations. Then, an $\varepsilon$-equilibrium price vector for a population of $m\geq 1$ agents can be computed in time of the order of $m\cdot q+T$.
\end{proposition}
\begin{proof}
    If preferences of individual agents are represented by $\vec{t}_1,\ldots, \vec{t}_m$ and $\beta_1,\ldots, \beta_m$ are relative incomes, then, by Theorem~\ref{th_representative_index}, the aggregate consumer corresponds to $\vec{t}=\sum_{k=1}^m  \beta_k\cdot \vec{t}_k$. Computing $\vec{t}$ requires a number of operations of the order of $m\cdot q$. Applying the one-agent algorithm to the aggregate agent, we get an $\varepsilon$-equilibrium vector of prices for the original population in~$T$ operations.
\end{proof}  

The linear growth of running time with the number of agents $m$ and the absence of large hidden constants suggests that finitely-generated domains can be a natural candidate for scalable fair division mechanisms. Note that the best running time for linear preferences achieved by \cite*{orlin2010improved} and \cite*{vegh2012strongly} grows as $m^4$. 

In economic design, the choice of a preference domain corresponds to the choice of a bidding language, i.e., the information about the true---possibly substantially more complicated---preferences that the participants can report to a mechanism. \fed{rephrase?}
Our results indicate the advantage of 
bidding languages corresponding to finitely-generated invariant domains. Using such domains offers considerable flexibility to the designer. For example, apart from Cobb-Douglas preferences, one can consider domains generated by a finite collection of linear preferences. By adding preferences exhibiting substitutability or complementarity among certain subsets---for example, pairs---of goods to the collection of elementary preferences, we can allow agents to express both substitutability and complementarity patterns while keeping the domain narrow. The use of such domains and bidding languages in practice requires additional experimental evaluation as in \citep*{budish2022can}.

\section{Indecomposable preferences}\label{sec_indecomposable}
In this section, we study indecomposable preferences, i.e., those that cannot be obtained by aggregating distinct preferences within a given domain. They play the role of elementary building blocks since any preference can be obtained by aggregating indecomposable preferences. 

We already saw an example of such a representation in Section~\ref{sec_invariant}, where we represented any preference over two substitutes as a continuous aggregation of linear preferences. In contrast to  Section~\ref{sec_invariant}, where we started by specifying ``elementary'' preferences and asked what can be obtained by aggregating them, now we start from a given domain and aim to identify these elementary preferences. 

\begin{definition}
For a given domain $\D$, a preference $\succsim$ from $\D$ is indecomposable if it cannot be represented as an aggregation of two distinct preferences $\succsim'$ and $\succsim''$ from $\D$. The set of all indecomposable preferences from $\D$ is denoted by $\D^\ind$.
\end{definition}

Recall that a point $x$ from a subset $X$ of a linear space is called an extreme point of $X$ if it cannot be represented as $\alpha x'+(1-\alpha)x''$ with $\alpha\in (0,1)$ and distinct\footnote{Usually, one assumes that $X$ is convex but we do not make this assumption.} $x',x''\in X$. The set of all extreme points of $X$ is denoted by  $X^\ext$. Theorem~\ref{th_representative_index} implies the following corollary.
\begin{corollary}\label{cor_indecomposable_as_extreme}
A preference $\succsim$ is indecomposable in a preference domain $\D$ if and only if its logarithmic expenditure function is an extreme point of the set of logarithmic expenditure functions $$\mathcal{L}_\D=\left\{f=\ln\left(\E_{\succsim'}\right)\colon \succsim'\in\D\right\}.$$
\end{corollary}
The Choquet theorem states that, if $X$ is a compact convex subset of a locally convex topological vector space, then any point  $x\in X$ can be obtained as the average of its extreme points $x'\in X^\ext$ with respect to some Borel probability measure $\mu$ supported on $X^\ext$:
\begin{equation}\label{eq_Choquet}
x=\int_{X^\ext} x'\, \dd
\mu(x');
\end{equation}
see \citep*{phelps2001lectures}.
Using the Choquet theorem, we obtain the following result demonstrating that indecomposable preferences can indeed be seen as elementary building blocks.

\begin{theorem}\label{th_choquet}
If $\D$ is a closed domain invariant with respect to aggregation, then any preference $\succsim$ from $\D$ can be obtained as a continuous aggregation of indecomposable preferences, i.e., there exists a Borel measure $\mu$ supported on $\D^\ind$ such that the expenditure function $\E=\E_\succsim$ can be represented as follows
\begin{equation}\label{eq_choquet_representation_of_indecomposable}
\ln \E(\vec{p})=\int_{\D^\ind} \ln \E_{\succsim'}(\vec{p})\, \dd \mu(\succsim')
\end{equation}
for any vector of prices $\vec{p}\in \R_{++}^n$.
\end{theorem}
 As in Theorem~\ref{th_continuous_aggregation},  the integral~\eqref{eq_choquet_representation_of_indecomposable} is formally defined in  Appendix~\ref{app_topology}.  We prove Theorem \ref{th_choquet} in Appendix~\ref{sec_th_continuous_aggregation_proof}. The essence of the proof is 
 verifying that the topological assumptions of the Choquet theorem hold.

Representation~\eqref{eq_choquet_representation_of_indecomposable} is especially useful if the set of indecomposable preferences is small relative to the whole domain $\D$.  We will see that this is the case for substitutes, but it is not true for complements or the full domain of homothetic preferences. 
\subsection{Indecomposability in the full domain}\label{sec_indecomposable_full_domain}

Let $\D$ be the domain of all homothetic preferences. 
It is easy to guess some indecomposable preferences from $\D$: for example, linear and Leontief preferences are indecomposable. It turns out that there are many more.
Let us call $\succsim$ a Leontief preference over linear composite goods if it corresponds to a utility function of the form
\begin{equation}\label{eq_extreme_full_domain}
    u(\mathbf{x})=\min_{\vec{a}\in A} \left\{\chi_\vec{a}(\vec{x})\right\},
\end{equation}
where $A$ is a finite or countably infinite subset of $\R_+^n$ and each $\vec{a}\in A$ defines a linear composite good $\chi_\vec{a}(\vec{x})$ by
$$
\chi_\vec{a}(\vec{x})=\sum_{i=1}^n a_i\cdot  x_i.
$$
The interpretation is that an agent treats the collection of bundles $\vec{a}\in A$ as perfect complements. Leontief and linear preferences are particular cases. For Leontief preferences, the bundles are, in fact, single goods and so each $\vec{a}\in A$ has only one non-zero coordinate. Linear preferences correspond to a single bundle $\vec{a}$, i.e.,  $A=\{\vec{a}\}$. Geometrically, Leontief preferences over linear composite goods are exactly those preferences that have upper contour sets with piecewise linear boundary; see Figure~\ref{fig_extreme_full}.
\begin{figure}
	\begin{center}
		\includegraphics[width=5cm]{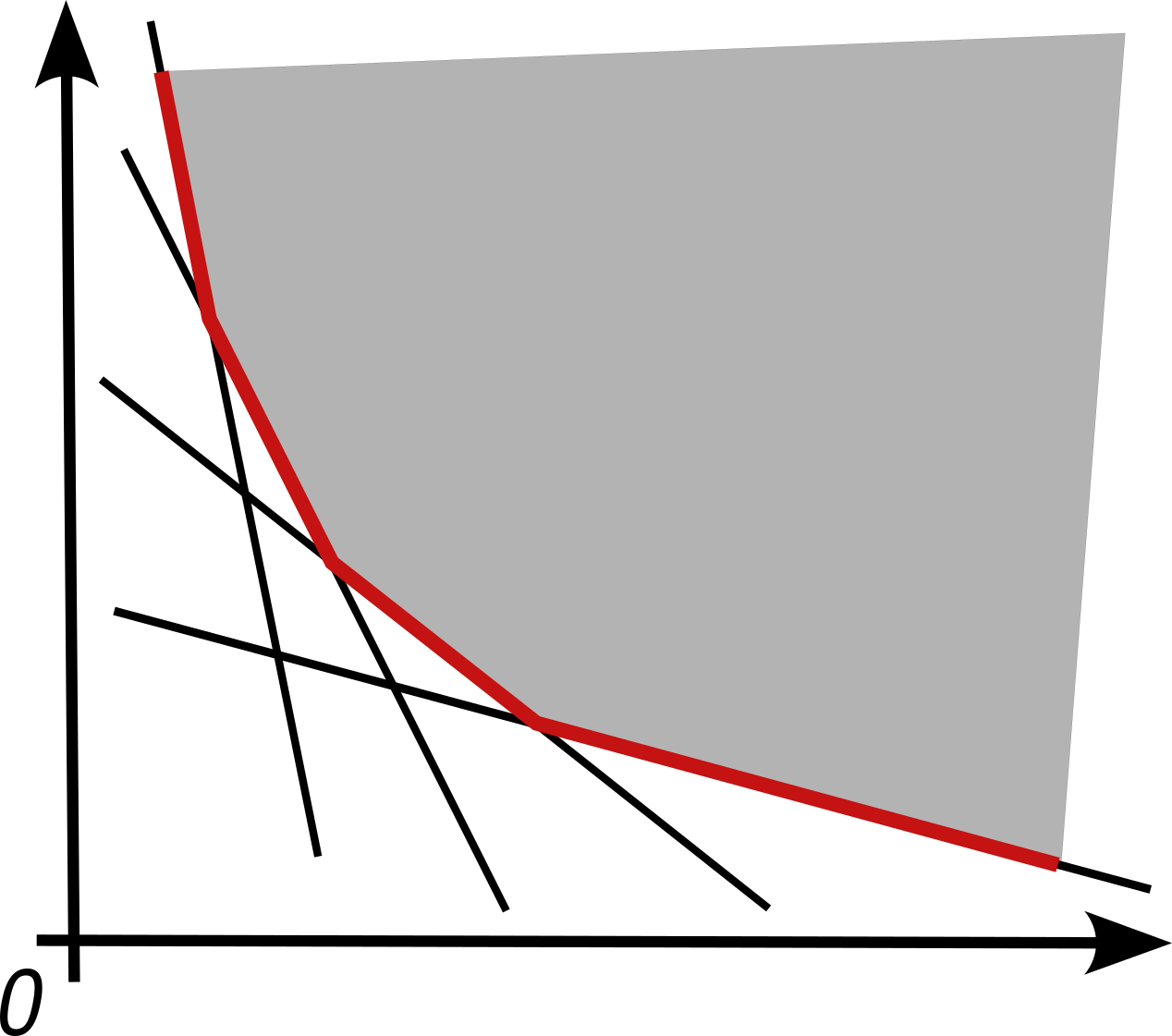}
	\end{center}
	\caption{A generic Leontief preference over linear composite goods. Such preferences have upper contour sets with piecewise linear boundary and
 are indecomposable in the domain of all homothetic preferences. \label{fig_extreme_full}
	}
\end{figure}

\begin{proposition}\label{prop_extreme_points_full_domain}
For any number of goods $n$,  Leontief preferences over linear composite goods~\eqref{eq_extreme_full_domain} are indecomposable in the domain of all homothetic preferences. 
\end{proposition}
 We provide a detailed proof of Proposition \ref{prop_extreme_points_full_domain} in~Appendix~\ref{app_proof_indecomposable_full_domain}, and sketch the proof later in this subsection. Proposition~\ref{prop_extreme_points_full_domain} implies that linear and Leontief preferences are indeed indecomposable.
Another immediate corollary is that aggregation of linear and Leontief preferences together is far from giving the full domain. Any preference of the form~\eqref{eq_extreme_full_domain} is indecomposable and, hence, cannot be represented as an aggregation of linear or Leontief preferences unless it is linear or Leontief itself.  For example, one can take a  preference $\succsim$ corresponding to  $$u(\vec{x})=\min\{x_1+2\cdot x_2,\ 2\cdot x_1+x_2\}.$$
The corollary can be strengthened. Expenditure shares for $\succsim$ are not monotone, i.e., $\succsim$ exhibits neither substitutability nor complementarity. Since $\succsim$ is indecomposable, we conclude that not every preference can be represented as 
an aggregation of preferences exhibiting substitutability and preferences exhibiting complementarity.

We see that the full domain has a lot of indecomposable preferences. To formalize this observation, note that piecewise linear concave functions are dense in the set of all concave functions. Accordingly, indecomposable preferences are dense in the full domain $\D$, and extreme points of the set $\mathcal{L}_\D$ of logarithmic expenditure functions  are dense in this set.\footnote{The existence of non-trivial convex sets with dense extreme points highlights that finite-dimensional intuition can be misleading in infinite-dimensional convex geometry \citep*{poulsen1959convex}. In economic literature, such sets have appeared in the context of the $n$-good monopolist problem with $n\geq 2$, where mechanisms can be identified with convex functions on $[0,1]^n$ such that their gradients also belong to $[0,1]^n$ \citep*{manelli2007multidimensional}.} 
\medskip

 The intuition of Proposition~\ref{prop_extreme_points_full_domain} is as follows: since indecomposable preferences are extreme points of the set of logarithmic expenditure functions $\mathcal{L}_\D$, finite-dimensional linear programming intuition suggests that natural candidate indecomposable preferences are those for which the maximal number of constraints defining the domain are active.\fed{reference} Leontief preferences over linear composite goods are those preferences $\succsim$ for which the concavity constraint on the expenditure function $\E$ is active almost everywhere!

\medskip
 
 Let us now sketch the proof of Proposition~\ref{prop_extreme_points_full_domain} here (see Appendix~\ref{app_proof_indecomposable_full_domain} for  a complete proof). A utility function $u$ is of the form~\eqref{eq_extreme_full_domain} if and only if the corresponding expenditure function is also piecewise linear: $\E=\min_{c\in C} \sum_{i=1}^n c_i \cdot p_i $ for a finite or countable $C\subset \R_+^n$. 
To demonstrate indecomposability, we need to show that if  $\E=\E_{1}^\alpha\E_{2}^{1-\alpha}$, then $\E_{1}$ and $\E_{1}$ are proportional to each other (and thus to $\E$). By strict concavity of the function $h(\vec{t})=t_1^{\alpha}\cdot t_2^{1-\alpha}$ on rays not passing through the origin, $\E$ cannot be linear in regions where $\E_{1}$ and $\E_{1}$ are not proportional. Hence, $\E_{1}$ and $\E_{2}$ must be proportional in each of the linearity regions of $\E$. As these regions cover the whole space, $\E_{1}$ and $\E_{2}$ are proportional  everywhere implying that $\succsim$ is indecomposable.
\smallskip

In Appendix~\ref{app_proof_indecomposable_full_domain}, we also explore how close Proposition~\ref{prop_extreme_points_full_domain} is to characterizing all indecomposable preferences. We show that 
if there is a neighborhood of a point where the concavity constraint on $\E$ is inactive, then the preference can be decomposed (Proposition~\ref{prop_strictly_concave_decomposable}). The idea is that we can find small perturbation $\psi=\psi(\vec{p})$
vanishing outside of this neighborhood and such that   $\E_1=\E\cdot (1+\psi)$ and 
$\E_2=\E/ (1+\psi)$ are valid logarithmic expenditure functions. 
Since $\ln \E=1/2 \cdot \ln \E_1+1/2\cdot \ln \E_2 $, the preference corresponding to $\E$ can indeed be decomposed.

Intuitively, a concave function is either piecewise linear or there is a neighborhood where it is strictly concave and, hence, Propositions~\ref{prop_extreme_points_full_domain} and~\ref{prop_strictly_concave_decomposable} seem to cover all possible cases. However, there is a family of pathological examples not captured by this intuition, e.g., concave functions whose second derivative is a non-atomic measure supported on a Cantor set. Proposition~\ref{prop_pathology} formulated and proved in the appendix shows that such pathological preferences are also indecomposable.


\medskip

Proposition~\ref{prop_extreme_points_full_domain} has implications for the geometric mean of convex sets. 
 Consider the collection $\mathcal{X}$ of all closed convex subsets $X$ of $\R_{+}^n$ that do not contain zero and are upward-closed, i.e., all those that can be obtained as upper contour sets of homothetic preferences.
We call a set $X\in \mathcal{X}$ indecomposable if it cannot be represented as the geometric mean $X_1^{\lambda}\otimes X_2^{1-\lambda}$ with distinct  $X_1$ and $X_2$ from $\mathcal{X}$ and $\lambda\in (0,1)$. Proposition~\ref{prop_extreme_points_full_domain} implies that convex polytopes (with a possibly infinite number of faces) are indecomposable.
There is mathematical literature inspired by \cite*{gale1954irreducible} and studying a similar concept of indecomposability where instead of taking weighted geometric means, one takes convex combinations with respect to the Minkowski addition.\footnote{\cite*{gale1954irreducible} calls such sets irreducible.} 
In contrast to our setting, planar sets that are indecomposable in the sense of Gale form a simple parametric family \citep*{gale1954irreducible,silverman1973decomposition}. However, in the dimension $3$ and higher, Gale's indecomposability behaves similarly to ours: indecomposable sets are dense in all convex sets, and one can derive some necessary and some sufficient conditions of indecomposability that almost match each other but yet no criterion is known; see, e.g.,  \citep*{sallee1972minkowski}. \fed{check}

\subsection{The domain of  substitutes and the simplex property}\label{sec_S_indecomposable}

Consider the domain $\mathcal{D}_{S}$ 
 of all homothetic preferences over $n$  substitutes. Linear preferences belong to $\mathcal{D}_{S}$  and are indecomposable since they are indecomposable even in the larger domain of all homothetic preferences by Proposition~\ref{prop_extreme_points_full_domain}. For two goods, there are no other indecomposable preferences in $\mathcal{D}_{S}$.
\begin{proposition}\label{prop_linear_are_all_indecomposable_for_S}
For $n=2$ goods, a preference $\succsim$ is indecomposable in the domain $\mathcal{D}_{S}$ of homothetic preferences with substitutability if and only if $\succsim$ is linear. 
\end{proposition}
\begin{figure}
	\begin{center}
		\includegraphics[width=5cm]{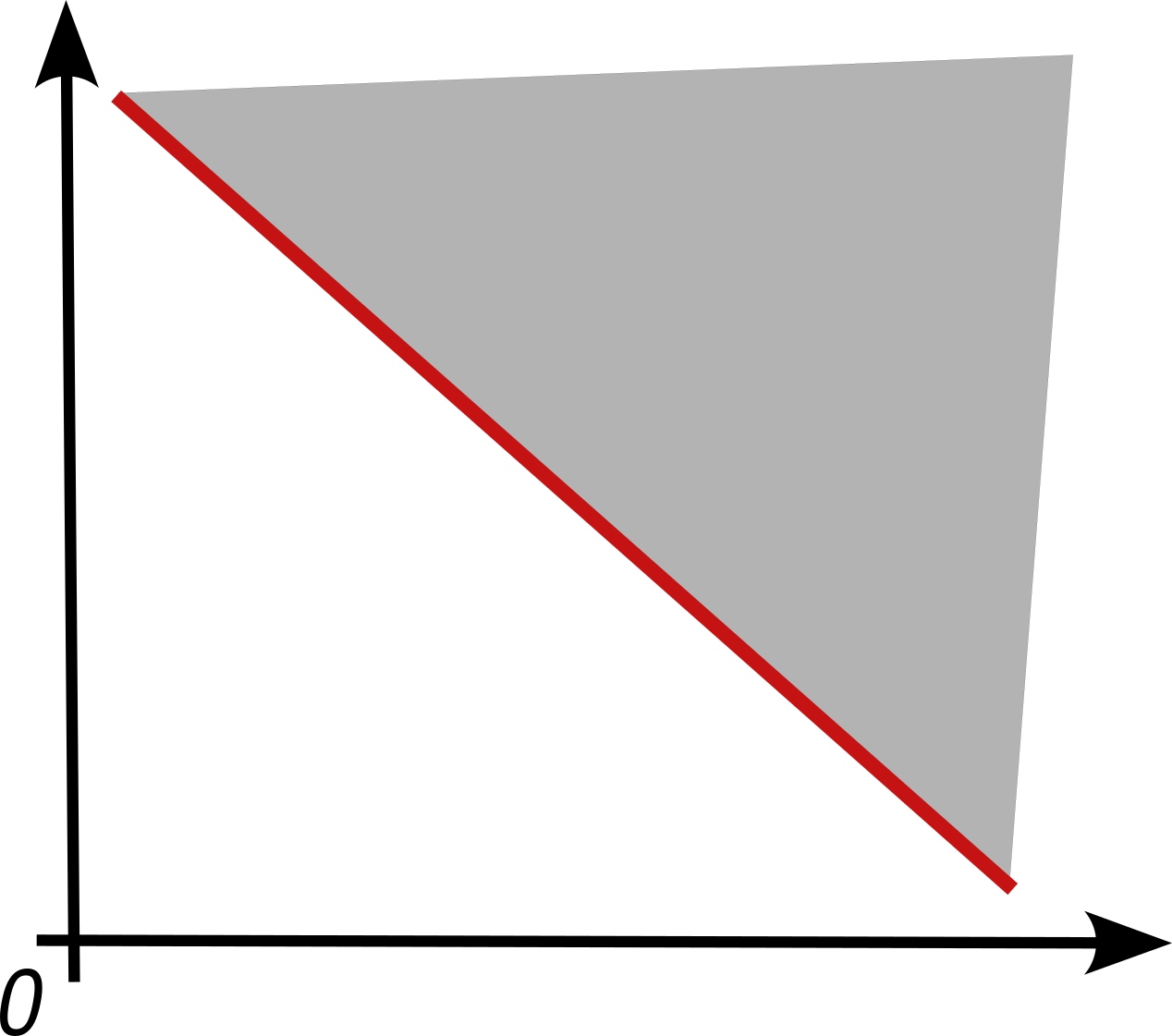}
	\end{center}
	\caption{For $n=2$ substitutes, linear  preferences exhaust all indecomposable ones. \label{fig_extreme_substitutes}
	}
\end{figure}

\begin{proof} From Corollary~\ref{cor_MRS_reconstruction}, we know that any preference over $n=2$ goods exhibiting substitutability can be obtained by aggregating linear preferences. Hence, any non-linear preference can be decomposed, and we get Proposition~\ref{prop_linear_are_all_indecomposable_for_S}. \end{proof}

Corollary~\ref{cor_MRS_reconstruction} 
 provides an explicit Choquet decomposition~\eqref{eq_choquet_representation_of_indecomposable} for $\mathcal{D}_{S}$.
Moreover, the corollary states that the decomposition is unique in the sense that the distribution of linear preferences over the population is pinned down uniquely. This phenomenon has the following geometric interpretation. \fed{rephrase this paragraph}

 Consider a collection of $d$ points in a finite-dimensional linear space
 such that no subset of $k\leq d$ points belongs to a $(k-2)$-dimensional linear subspace. The convex hull of such a collection is called a simplex.  A 
 simplex has the property that any point has a unique representation as a convex combination of extreme points.
The uniqueness of the decomposition characterizes simplices among all other closed convex subsets of a finite-dimensional space.
In the infinite-dimensional space, this property can be used to define a simplex, namely, a compact convex set is called a simplex if each point can be uniquely represented as the average of the extreme points, i.e., the measure in the Choquet integral~\eqref{eq_Choquet} is uniquely defined \citep*{phelps2001lectures}.
Accordingly, we say that a closed domain of preferences is a simplex domain if there is a unique way to represent any preference as an aggregation of indecomposable preferences, i.e., the measure  $\mu$ in~\eqref{eq_choquet_representation_of_indecomposable} is unique. Using this vocabulary, one can rephrase Proposition \ref{prop_linear_gwnerate_S_for_2_goods}  as follows. 
\begin{corollary}\label{cor_simplex_S}
For two goods, the domain of homothetic preferences exhibiting substitutability is a simplex domain. 
\end{corollary}

For $n\geq 3$ goods, Corollary~\ref{cor_linear_give_proper_subset_of_S} implies that the set of indecomposable preferences in $\mathcal{D}_{S}$ is richer than the linear preference domain. Describing that set explicitly and checking whether $\mathcal{D}_{S}$ is a simplex domain for $n\geq 3$ remains an open question.

\subsection{The domain of  complements}
Let us discuss the domain $\mathcal{D}_{C}$ of homothetic preferences exhibiting complementarity. Leontief preferences are indecomposable in $\mathcal{D}_{C}$ since they are indecomposable in the full domain. By Corollary~\ref{cor_Leontief_does_not_give_all_complements}, aggregation of Leontief preferences does not give the whole $\mathcal{D}_{C}$ even for $n=2$ goods and, hence, there must be other indecomposable preferences. It turns out that indecomposable preferences are dense in $\mathcal{D}_{C}$, and their structure resembles the one for the full domain. We call  a preference $\succsim$  a Leontief preference over Cobb-Douglas composite goods if it corresponds to a utility function
\begin{equation}\label{eq_extreme_C}
    u(\vec{x})=\min_{\vec{a}\in A} \left\{\chi_\vec{a}(\vec{x})\right\},
\end{equation}
where $A$ is finite or countably infinite subset of $\R_{++}\times \Delta_{n-1}$ and each $\vec{a}=(a_0,a_1,\ldots, a_n)\in A$ defines a Cobb-Douglas composite good $\chi_\vec{a}(\vec{x})$ by
$$
\chi_\vec{a}(\vec{x})= a_0\cdot\prod_{i=1}^n x_i^{a_i}.
$$
Cobb-Douglas and Leontief preferences are  particular cases of~\eqref{eq_extreme_C} corresponding, respectively, to a singleton $A=\{\vec{a}\}$ and to $A=\{(a_0^1,\vec{e}_1),\ldots, (a_0^n,\vec{e}_n)\}$ where $\vec{e}_i$ is the vector whose $i$th entry is $1$ and others are $0$. A generic preference~\eqref{eq_extreme_C} is depicted in Figure~\ref{fig_extreme_complements}. We call a Leontief preference $\succsim$ over Cobb-Douglas composite goods non-trivial if the set $A$ contains at least two vectors $\vec{a}$ and $\vec{a}'$ with  $(a_1,\ldots, a_n)\ne (a_1',\ldots, a_n')$. Equivalently, $\succsim$ is non-trivial if it is not a standard Cobb-Douglas preference.
\begin{figure}
	\begin{center}
		\includegraphics[width=5cm]{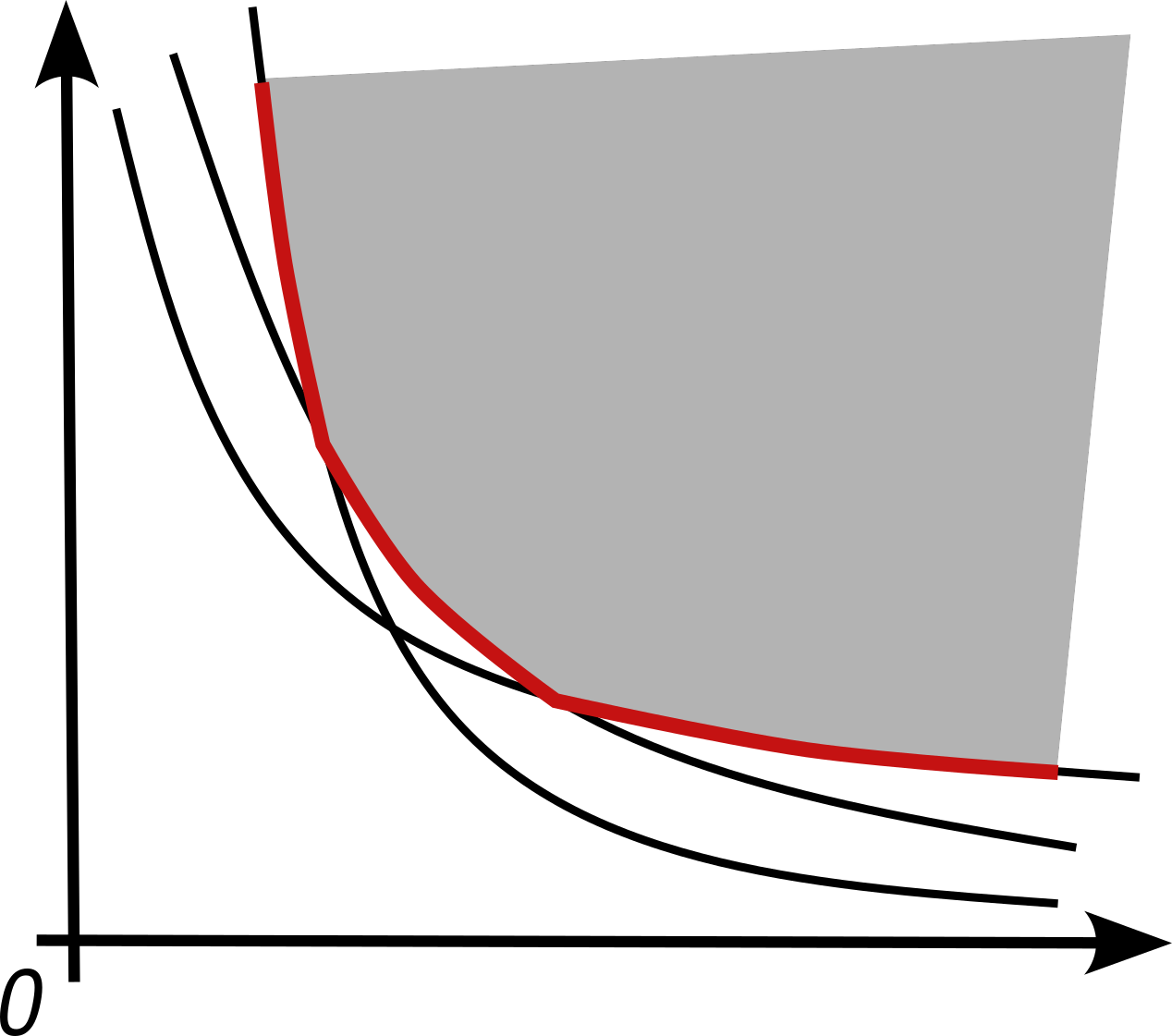}
	\end{center}
	\caption{A generic Leontief preference over Cobb-Douglas composite goods. Such preferences have upper contour sets with piecewise Cobb-Douglas boundary and
 are indecomposable in the domain of all preferences exhibiting complementarity. \label{fig_extreme_complements}
	}
\end{figure}
\begin{proposition}
\label{prop_indecomposable_for_C}
For $n=2$ goods, non-trivial Leontief preferences over Cobb-Douglas composite goods are indecomposable in the domain of homothetic preferences with complementarity.
\end{proposition}
The requirement of non-triviality is needed since standard Cobb-Douglas preferences can be decomposed as an aggregation of $\succsim_i$ corresponding to $u_i(\vec{x})=x_i$; see Example~\ref{ex_CobbDouglas_as_aggregation_of_linear}. Note that $\succsim_i$---which can be seen as either degenerate linear or degenerate Cobb-Douglas preference---is indecomposable in $\D_C$ since it is indecomposable even in the full domain by Proposition~\ref{prop_extreme_points_full_domain}. 

Proposition~\ref{prop_indecomposable_for_C} is proved in Appendix~\ref{app_proof_indecomposable_complement}. The idea is similar to Proposition~\ref{prop_extreme_points_full_domain} dealing with indecomposability in the full domain: indecomposable preferences correspond to expenditure functions~$\E$ with the maximal number of active constraints. In contrast to Proposition~\ref{prop_extreme_points_full_domain} where the concavity of $\E$ was the only constraint that matters, now we have the new monotonicity constraint on expenditure shares. Leontief preferences over Cobb-Douglas composite goods are obtained if the space is partitioned into regions where the monotonicity constraint is active (expenditure shares are constant) or the concavity constraint is active (the expenditure function is linear). The former regions correspond to hyperbolic parts of the upper contour sets, and the latter regions, to cusps.
\fed{Discuss why complements are hard and substitutes are easy}

%

\subsection{Application: preference identification}\label{sec_identification}

Market demand reflects individual preferences, but information loss is unavoidable. For example, aggregate behavior does not distinguish populations where a pair of agents swapped their preferences and incomes or where a pair of agents with identical preferences is replaced with one agent with the joint income. We can still ask whether market demand determines the distribution of preferences over the population, i.e., whether, by looking at the aggregate behavior, it is possible to determine what fraction of the population's income corresponds to agents with preferences of a particular kind. 

Consider a population of consumers with homothetic preferences from some domain $\D$. An analyst observes market demand generated by this population for any vector of prices but has no information on the income distribution nor on the size of the population. For any subset of preferences $\D'\subset 
\D$, the analyst aims to identify what fraction of the total income corresponds to agents in $\D'$. \fed{Drop population size}

In general, identification is impossible. To illustrate, consider Example~\ref{ex_CobbDouglas_as_aggregation_of_linear} again. The aggregate demand corresponding to $u_{\agr}(x_1,x_2)=x_1^{1/3}\cdot x_2^{2/3}$ can be generated by a population where each agent has Cobb-Douglas preferences with budget shares $\left(1/3,2/3\right)$ or, alternatively, by the population where $1/3$ of the total income is earned by agents with preference $u_{1}(\vec{x})=x_1$ and the remaining $2/3$ is by those with  $u_{2}(\vec{x})=x_2$.

The domain $\D$ of linear preferences over $n=2$ goods is an exception. A linear preference over two goods is determined by its $\mrs$ captured by the ratio ${v_1}/{v_2}$. By Corollary~\ref{cor_MRS_reconstruction}, the fraction of income corresponding to consumers with  $\mrs$ above a certain threshold $\alpha$ is equal to the fraction of income spent by the population on the first good for prices $p_1=\alpha\cdot  p_2$, i.e.,
$$\mu\left(\left\{\frac{v_1}{v_2}\geq\frac{p_1}{p_2}\right\}\right)=s_{\agr,1}(\vec{p})=\frac{p_1\cdot D_{\agr,1}(\vec{p},B)}{p_1\cdot D_{\agr,1}(\vec{p},B)+ p_2\cdot D_{\agr,2}(\vec{p},B)}.$$
Hence, even a few observations of market demand $D_\agr$ at non-collinear price vectors can give a good understanding of the preference distribution over the population.

More generally, the distribution of preferences from a domain $\D$ can be identified if any preference $\succsim$ obtained by aggregation of preferences from $\D$
cannot be decomposed over $\D$ in a different way. The following corollary provides a geometric interpretation of this property relying on the notion of simplex domains from Section~\ref{sec_S_indecomposable}.
\begin{corollary}\label{cor_identification}
If the completion $\D^\inv$
of $\D$ is a simplex domain, and $\D$ consists of indecomposable preferences, then the distribution of preferences over the population can be identified from observations of the market demand.
\end{corollary}

Apart from linear preferences over two goods, there are many other domains satisfying the requirements of Corollary~\ref{cor_identification}. One can take $\D$ given by any finite collection of preferences $\{\succsim_1,\dots, \succsim_q\}$, none of which can be obtained as an aggregation of the others. For example, if $q$ equals the number of goods $n$ and  each $\succsim_k$ corresponds to $u(\vec{x})=x_k$, then the income fraction of consumers with preference $\succsim_k$ is equal to the expenditure share $s_{\agr,k}(\vec{p})$ at any price $\vec{p}$; see also Example~\ref{ex_CobbDouglas_as_aggregation_of_linear}. We stress that just one observation of aggregate behavior at any particular vector of prices $\vec{p}$ turns out to be enough to determine the distribution of preferences. The origin of this phenomenon is not the orthogonality of preferences (which we do not require) but the fact that the dimension of the domain of preferences does not exceed the dimension of the consumption space. To illustrate this point, note that, if  $\succsim_1,\dots, \succsim_q$ are Cobb-Douglas preference with vectors of parameters  $\vec{a}_1,\ldots, \vec{a}_q$ that are linearly independent  (possible only if  $q\leq n$), then one observation of market demand also gives a linear system enough for identification.

Another domain $\D$ satisfying conditions of Corollary~\ref{cor_identification} is the domain of Leontief preferences over two goods. By Corollary~\ref{cor_complex_analysis}, any preference from its completion $\D^\inv$ can be uniquely decomposed over Leontief preferences. Hence, $\D^\inv$ is a simplex domain, and Leontief preferences are indecomposable in it. We conclude that, in theory, the distribution of Leontief preferences can be identified.\footnote{A peculiarity is that the Stieltjes-Perron inversion formula underlying Corollary~\ref{cor_complex_analysis} requires continuation of demand to complex prices. Therefore, it guarantees identification but gives no practical recipe for reconstructing the distribution of preferences for an analyst who observes demand for real prices only. Instead, the analyst can use real-inversion techniques for the Stieltjes transform, e.g., \citep*{widder1938stieltjes,love1980real}.}

\bibliography{main}

\appendix

\section{Convex analysis basics}\label{sec_convex_basics}
\paragraph{Superdifferentials.} Let $X$ be a convex subset of $\R^n$ and $f\colon X\to \R$ a concave function.
Any such function $f$ is continuous in the relative interior of $X$ but can have discontinuities on the boundary \cite[7.24 Theorem]{charalambos2013infinite}. The {superdifferential} of  $f$ is defined by
\begin{equation}\label{eq_superdifferential}
	\partial f(\vec{x}) = \{\vec{p}\in \R^n\ \colon \ f(\vec{y}) \leq f(\vec{x}) + \langle \vec{p},\, \vec{y} - \vec{x} \rangle, \ \ \ \forall \vec{y}\in X\}.
\end{equation}
The superdifferential is non-empty in the relative interior of $X$. Recall that the gradient $\nabla f(x)$ is the vector of partial derivatives
$$\nabla f(\vec{x})=\left(\frac{\partial f}{\partial x_1},\ldots, \frac{\partial f}{\partial x_n}\right).$$
The gradient is defined for any $\vec{x}$ where partial derivatives exist. 
At any such point,  the superdifferential $\partial f(\vec{x})$ consists of just one element: $\partial f(\vec{x})=\{\nabla f(\vec{x})\}$ \citep[Theorem 25.1]{rockafellar1970convex}. By the Alexandrov theorem, a concave function is twice differentiable except for a set of zero Lebesgue; see \cite[Theorem 25.5]{rockafellar1970convex}. In particular, the gradient $\nabla f(\vec{x})$ is defined almost everywhere and coincides with the superdifferential.


\paragraph{Application to demand and expenditure functions.}
Consider a consumer with budget $b>0$ and homothetic preferences $\succsim$ on $\R_+^n$ represented by a homogeneous utility function $u$. The demand $D(\vec{p},b)$, considered as a function of prices and the budget, is referred to as the Marshallian demand. The demand as a function of prices and the utility level $w>0$ is referred to as the Hicksian demand:
$$H(\vec{p},w)=\argmin_{\vec{x}:\, u(\vec{x})\geq w} \langle\vec{p},\vec{x}\rangle.$$

Sheppard's lemma for homothetic preferences gives the following identity:
$$H(\vec{p},w)= w\cdot \partial E(\vec{p}),
$$
where $\partial E(\vec{p})$ is the superdifferential of the expenditure function. This identity holds for all $\vec{p}\in \R_{++}^n$ including those points where $E$ is not differentiable \cite[][p.141]{mas1995microeconomic}. 

By the utility-maximization/cost-minimization duality \citep*{diewert1982duality}, $D(\vec{p},b)=H(\vec{p},w)$ if $w$ is such that $\langle\vec{p},\vec{x}\rangle=b$ for $\vec{x}\in H(\vec{p},w)$. Hence, $D(\vec{p},b)= w\cdot \partial E(\vec{p})$
for some $w=w(b)$. Choosing $w(b)$ so that each bundle from the right-hand side has the price of $b$, we get
$$D(\vec{p},b)= b\cdot \frac{\partial E(\vec{p})}{E(\vec{p})}=b\cdot \partial \ln E(\vec{p}),$$
where we used the Euler identity for $1$-homogeneous functions: $\langle \partial E(\vec{p}),\vec{p}\rangle=E(\vec{p})$. 

By the Alexandrov theorem, there exists a set $A\subset \R_{++}^n$ such that $ \R_{++}^n\setminus A$ has zero Lebesgue measure and $\E$ is differentiable for any $\vec{p}\in A$. By homogeneity of $\E$, the set $A$ can be selected so that $\vec{p}\in A$ $\Rightarrow $ $\alpha\cdot \vec{p}\in A$ for any $\alpha>0$.
For $\vec{p}\in A$, the superdifferential $\partial \E$ consists of just one element, the gradient  $\partial \E(\vec{p})=\{\nabla \E(\vec{p}) \}$. 
Consequently, the Marshallian demand $D(\vec{p},b)$ contains just one bundle on the set $\vec{p}\in A$ of full measure and can be thought of as a single-valued function
$$D(\vec{p},b)= b\cdot \nabla \ln E(\vec{p})
$$
defined almost everywhere.

For a bundle $\vec{x}\in D(\vec{p},b)$, the expenditure share of a good~$i$ is defined by $x_i\cdot p_i/ b$. For $\vec{p}\in A$, the demand is single-valued, and so the expenditure share
is a single-valued function of prices defined almost everywhere and satisfying the identity
$$\vec{s}(\vec{p},b)=\left(p_1\cdot \frac{\partial \ln \E}{\partial p_1},\ldots,p_n\cdot \frac{\partial \ln \E}{\partial p_n}\right)=\left(\frac{\partial \ln \E}{\partial \ln p_1},\ldots, \frac{\partial \ln \E}{\partial \ln p_n}\right).$$

\section{Topology on preferences and integration}\label{app_topology}

There are two high-level reasons why we need a topology on preferences. The topology is necessary to define the closure of preference domains in our discussion of completion, but, most importantly, the topology is needed to formalize integration over preferences and to apply Choquet theory \citep*{phelps2001lectures}. Recall that  Choquet theory deals with compact convex subsets of locally convex topological vector spaces. 
Our goal is to identify the domain of all homothetic preferences with 
a compact convex subset of a Banach space (a complete normed and, hence, locally convex vector space).

We represent a homothetic preference $\succsim$ by its logarithmic expenditure function $\ln \E$. We call two functions $f$ and $g$ equivalent if $f-g=\const$.
Since the expenditure function is defined up to a multiplicative factor, each preference corresponds to the class of equivalent logarithmic expenditure functions.

Let $\L$ be the set of classes of equivalent continuous functions $f$ on $\R_{++}^n$ that can be obtained as logarithmic expenditure functions of homothetic preferences. 
%
The set $\L$ is in one-to-one correspondence with the domain of homothetic preferences. Hence, to define a topology and integration for preferences, it is enough to define them for $\L$. We first introduce a metric structure. To motivate the definition of a distance, we need some estimates on the magnitude of expenditure functions.

\begin{lemma}
For any expenditure function $\E$, the following inequality holds
\begin{equation}\label{eq_lipshitz_price_index}
\left|\ln \E(\vec{p})-\ln \E(\vec{p}')\right|\leq \max_i \left|\ln p_i -\ln p_i'\right|
\end{equation}
for any pair of price vectors $\vec{p}$ and $\vec{p}'$ from $\R_{++}^n$.
\end{lemma}
In other words, logarithmic expenditure functions are $1$-Lipshitz functions of logarithms of prices.
\begin{proof}
We need to show that 
$$\min_{i} \frac{p_i}{p_i'}\leq \frac{\E(p)}{\E(p')}\leq \max_{i} \frac{p_i}{p_i'}.$$
It is enough to demonstrate the upper bound, and the lower bound will follow by flipping the roles of $\vec{p}$ and $\vec{p}'$.

Recall that $\E(\vec{p})$ is the minimal budget that the agent needs to achieve the unit level of utility for prices $\vec{p}$. Given prices $\vec{p}$ and $\vec{p}'$, define $\vec{p}''=\max_{i} \frac{p_i}{p_i'}\cdot \vec{p}'$. The price of each good under $\vec{p}''$ is higher than for $\vec{p}$ and, hence, the agent needs at least as much money to achieve the same utility level. Thus
$$\E(\vec{p})\leq \E(\vec{p}'')= \max_{i} \frac{p_i}{p_i'}\cdot \E(\vec{p}'),$$
where we used the homogeneity of the expenditure function. Dividing both sides by $\E(\vec{p}')$, we obtain the desired inequality and complete the proof.
\end{proof}
Denote by $\vec{e}$ the vector of all ones $\vec{e}=(1,\ldots, 1)$.\fed{clash with excess demand} By the lemma, we see that any expenditure function satisfies the following estimate
\begin{equation}\label{eq_price_index_upper_bound}
    \left|\frac{\ln \E(\vec{p})-\ln \E(\vec{e})}{1+\max_i |\ln p_i|}\right|\leq 1
\end{equation}
for any vector of prices. The normalization in~\eqref{eq_price_index_upper_bound}  suggests how to define a distance so that the set of logarithmic expenditure functions has a bounded diameter. 

We define the distance  between preferences $\succsim$ and $\succsim'$ or, equivalently, between the corresponding logarithmic expenditure functions $f=\ln \E$ and $f'=\ln \E'$ as follows:
\begin{equation}\label{eq_distasnce}
d(\succsim,\succsim')=d(f,f')=\sup_{\vec{p}\in \Delta_{n-1}\cap \R_{++}^n} \left|\frac{\left(\ln \E(\vec{p})-\ln \E(\vec{e})\right)-\left(\ln \E'(\vec{p})-\ln \E'(\vec{e})\right)}{\left(1+\max_i |\ln p_i|\right)^2}\right|.
\end{equation}
 The denominator in~\eqref{eq_distasnce} is squared so that 
 \begin{equation}\label{eq_boundary_convergence}
 \frac{\ln \E(\vec p)-\ln \E(\vec{e})}{\left(1+\max_i |\ln p_i|\right)^2}\to 0 \quad \mbox {as $\vec{p}$ approaches the boundary of $\Delta_{n-1}$.}
\end{equation}
 Hence, the supremum is always attained at an interior point of the simplex and so can be replaced with the maximum. 
  
  Note that the ratio in~\eqref{eq_distasnce} does not depend on the choice of a logarithmic expenditure function from the class of equivalent ones.
  On the other hand, the distance between any two distinct preferences or, equivalently, between two non-equivalent logarithmic expenditure functions is non-zero as the values of logarithmic expenditure functions on $\R_{++}^n$ are determined by their values on the interior of the simplex $\Delta_{n-1}\cap \R_{++}^n$ since $\E(\alpha\cdot \vec{p})=\alpha\cdot \E(\vec{p})$.
  \medskip

The metric structure on preferences allows one to define convergence and closed sets. A closure of a domain $\D$ of preference consists of all limit points of $\D$, i.e., of all the  preferences  $\succsim$ such that there exists a sequence of preferences $\succsim^{(l)}\in \D$ with $d(\succsim,\succsim^{(l)})\to 0$ as $l\to \infty$. A closed domain is a domain that  coincides with its closure.

Open sets are complements of closed ones, and so the metric defines a topology. Once the topology is defined, one constructs the Borel measurable structure in the standard way \citep[Section 4.4]{charalambos2013infinite}. Hence, we can write  integrals of the form 
$$\int_\D G(\succsim)\, \dd \mu(\succsim) = \int_\L G(f)\, \dd \mu(f)$$
formally where $G$ is a Borel-measurable function and $\mu$ is a Borel measure (as usual, we identify functions and measures on preferences and on logarithmic expenditure functions). In all our examples, the integrated function $G$ is continuous and, hence, measurable. 

By~\eqref{eq_price_index_upper_bound}, the diameter of $\L$ does not exceed~$2$. Hence, $\L$ is a bounded convex set. To fit the assumptions of the Choquet theory, we need to show that $\L$ is compact and can be thought of as a subset of a Banach space. We achieve both goals by constructing an isometric compact embedding of $\L$ into a Banach space.

Consider the Banach space $\cC\left(\Delta_{n-1}\right)$ of all continuous functions on the simplex endowed with the standard sup-norm $\|h\|=\sup_{\vec{p}\in\Delta_{n-1}} |h(\vec{p})|$. 
\begin{lemma}\label{lemma_embedding}
Let  $\T$ be a map that maps a logarithmic expenditure function $f=\ln \E$ to a function $\T[f]$ on $\Delta_{n-1}$ given by \fed{rephrase}
\begin{equation}\label{eq_embedding}
\T[f](\vec{p})=\begin{cases}\displaystyle{\frac{\ln \E(\vec p)-\ln \E(\vec{e})}{\left(1+\max_i |\ln p_i|\right)^2}}, &\vec{p}\in \Delta_{n-1}\cap \R_{++}^n\\
0,& \mbox{otherwise}
\end{cases}.
\end{equation}
Then $E$ is an isometric embedding of the set $\L$ of logarithmic expenditure functions into the Banach space of continuous functions $\cC\left(\Delta_{n-1}\right)$ and the image $\T[\L]$ is a compact convex set.
\end{lemma}
\begin{proof}
The function $\T[f]$ is continuous in the interior of the simplex by the continuity of expenditure functions, and it is continuous on the boundary by~\eqref{eq_boundary_convergence}. Hence, $\T[f]$ belongs to $\cC(\Delta_{n-1})$.
By the definition of the distance~\eqref{eq_distasnce} and the norm in $\cC(\Delta_{n-1})$, we get $d(f,f')=\left\|\T[f]-\T[f']\right\|$. Hence, $\T$ preserves the distance and, in particular, $f\ne f'$ implies $\T[f]\ne \T[f']$. Thus $\T$ is an isometric embedding of $\L$ in $\cC(\Delta_{n-1})$.

The diameter of the image $\T[\L]$ of $\L$ does not exceed~$2$ by~\eqref{eq_price_index_upper_bound}. Hence, $\T[\L]$ is a bounded subset of $\cC(\Delta_{n-1})$. By Lemma~\eqref{eq_lipshitz_price_index},  functions from $\T[\L]$ are uniformly equicontinuous. Applying the  Arcell\`{a}-Ascoli theorem, we conclude that the closure of $\T[\L]$ is compact.\footnote{Consider a subset $\mathcal{T}$ of the set $\cC(X)$ of continuous functions on a compact set $X$ with the $\sup$-norm. The
Arcell\`{a}-Ascoli theorem states that the closure $\overline{\mathcal{T}}$ of $\mathcal{T}$ is compact in $\cC(X)$ if $\mathcal{T}$ is bounded and functions from $\mathcal{T}$ are uniformly equicontinuous.} \fed{Add a reference}

It remains to show that $\T[\L]$ is closed and convex. The set $\L$ is convex by Theorem~\ref{th_representative_index} and~$\T$ maps convex combinations to convex combinations, hence $\T[\L]$ is convex. To show that it is closed, consider a sequence of functions $h^{(l)}\in \T[\L]$ converging to some $h$ and show that the limit belongs to $\T[\L]$. Convergence in $\|\cdot \|$ implies pointwise convergence, and hence $h$ is equal to zero at the boundary of the simplex. At any $\vec{p}$ from the interior, we obtain that the sequence of expenditure functions
${\E^{(l)}(\vec{p})}/{\E^{(l)}(\vec{e})}$  corresponding to $h^{(l)}$ converges to $g(\vec{p})=\exp\left((1+\max_i|\ln p_i|)^2\cdot h(\vec{p})\right)$. As concavity is preserved under pointwise limits, $g$ is a non-negative concave function on $\Delta_{n-1}\cap \R_{++}^n$ and, hence, there is a preference with an expenditure function $\E=g$. Therefore, $h=\T[\ln \E]$ and so $\T[\L]$ is closed. 
\end{proof}

By Lemma~\ref{lemma_embedding}, one can think of $\L$ and the set of all homothetic preferences as a closed convex subset of $\cC(\Delta_{n-1})$ and thus can use Choquet theory.

\section{Proofs}

\subsection{Proof of Theorem~\ref{th_representative_index}}\label{sec_th1_proof}
\fed{Delete redundant parentheses}
\begin{proof}
By the result of \cite*{eisenberg1961aggregation}, we know that the aggregate consumer exists and her preference corresponds to the following utility function
\begin{equation}\label{eq_EisenbergGale_utility_APPENDIX}
u_{\agr}(\vec{x})=\max\left\{\prod_{k=1}^m \left(\frac{u_{k}(\vec{x}_k)}{\beta_k}\right)^{\beta_k}\quad:\quad \vec{x}_k\in \R_+^n,\ \ k=1,\ldots,m,\quad \sum_{k=1}^m \vec{x}_k=\vec{x}\right\}.
\end{equation}
Our goal is to compute the corresponding expenditure function $\E_\agr(\vec{p})$ and check that it satisfies the identity 
\begin{equation}\label{eq:pr_APPENDIX}
\ln \left({\E}_{{\agr}}(\vec{p})\right)=\sum_{k=1}^{m} \beta_{k}\cdot  \ln \left(\E_{k}(\vec{p})\right).
\end{equation}
As an intermediate step, we compute the indirect utility of the aggregate consumer. Recall that the indirect utility of a consumer with a direct utility function $u$ is given by
\begin{equation}\label{eq_indirect_proof_of_TH1}
v(\vec{p},b)=\max\left\{ u(\vec{x})\,\colon\, \vec{x}\in \R_+^n, \ \langle\vec{x},\vec{p}\rangle\leq b  \right\}=\frac{b}{\E(\vec{p})}.
\end{equation}
For the aggregate consumer, we get 
\begin{equation}\label{eq_representative_indirect_util}
v_\agr(\vec{p},b)=\max\left\{\prod_{k=1}^m \left(\frac{u_k(\vec{x}_k)}{\beta_k}\right)^{\beta_k}\quad:\quad  \vec{x}_k\in \R_+^n,\ \ k=1,\ldots,m,\quad \left\langle\vec{p},\  \sum_{k=1}^m \vec{x}_k \right\rangle\leq b\right\}.
\end{equation}
Plug  in $b=1$ and consider an optimal collection of bundles $\vec{x}_k$, $k=1,\ldots,m$, in \eqref{eq_representative_indirect_util}.
 Denote their prices $\langle \vec{p},\vec{x}_k\rangle $ by $\alpha_k$. Our goal is to show that $\alpha_k=\beta_k$. The argument is along the lines of  \cite*{eisenberg1961aggregation}. 
 Rescale each bundle $\vec{x}_k$ to make its price equal to $\beta_k$. We obtain a new collection of bundles $\vec{x}_k'=\frac{\beta_k}{\alpha_k}\cdot \vec{x}_k$, which also satisfies the aggregate budget constraint $\left\langle\vec{p},\  \sum_{k=1}^m \vec{x'}_k \right\rangle\leq 1$. By the optimality of $\vec{x}_k$, the product of utilities $\prod_k \big(u_k(\vec{x}_k)\big)^{\beta_k}$ is at least as big as $\prod_k \big(u_k(\vec{x}_k')\big)^{\beta_k}$. By the homogeneity of utilities, this inequality of the products can be rewritten as follows:
$$1\leq \prod_{k=1}^m\left(\frac{\alpha_k}{\beta_k}\right)^{\beta_k}.$$
Taking logarithm, we get an equivalent inequality
\begin{equation}\label{eq_proof_dual_Eisenberg}
0\leq \sum_{k=1}^m \beta_k\cdot \ln \frac{\alpha_k}{\beta_k}.
\end{equation}
The concavity of the logarithm implies an upper bound on the right-hand side
\begin{equation}\label{eq_proof_dual_Eisenberg2}
\sum_{k=1}^m \beta_k\cdot \ln \frac{\alpha_k}{\beta_k}\leq \ln\left(\sum_k \beta_k\cdot \frac{\alpha_k}{\beta_k}\right)=\ln\left(\sum_k \alpha_k\right)\leq \ln(1)=0.
\end{equation}
Inequalities~\eqref{eq_proof_dual_Eisenberg} and~\eqref{eq_proof_dual_Eisenberg2} can only be compatible if they are, in fact, equalities. As the logarithm is strictly concave, the equality between the first two expressions in~\eqref{eq_proof_dual_Eisenberg2} implies that the ratio $\frac{\alpha_k}{\beta_k}$ is a constant independent of $k$. Since the average value of the logarithms is zero, this constant equals one.\fed{rephrase} We conclude that $\alpha_k=\beta_k$.

We proved that $\langle \vec{p},\vec{x}_k\rangle =\beta_k$ for any optimal collection of bundles $\vec{x}_k$, $k=1,\ldots, m$, from the optimization problem~\eqref{eq_representative_indirect_util}. In particular, the inequality $\langle \vec{p},\vec{x}_k\rangle\leq \beta_k$ always holds at the optimum. Therefore, we can replace the budget constraint of the aggregate consumer $\left\langle p,\,  \sum_{k=1}^m \vec{x}_k\right\rangle\leq 1$ with a stronger requirement of individual budget constraints $\langle \vec{p},\vec{x}_k\rangle\leq \beta_k$, $k=1,\ldots,m$, and this modification will not alter the value:
\begin{equation}\notag
v_\agr(\vec{p},1)=\max\left\{\prod_{k=1}^m \left(\frac{u_k(\vec{x}_k)}{\beta_k}\right)^{\beta_k}\quad:\quad  \vec{x}_k\in \R_+^n, \quad \langle \vec{p},\vec{x}_k\rangle\leq \beta_k, \ \   k=1,\ldots,m\right\}.
\end{equation}
The maximization of the product reduces to maximizing each term $u_k(\vec{x}_k)$ separately over the corresponding budget set $\langle \vec{p},\vec{x}_k\rangle\leq \beta_k$, which gives the indirect utility of consumer $k$:
\begin{equation}\notag
v_\agr(\vec{p},1)=\prod_{k=1}^m\left(\frac{\max\Big\{ u_k(\vec{x}_k)\quad:\quad  \vec{x}_k\in \R_+^n,\quad  \vec{p}\cdot \vec{x}_k \leq \beta_k\Big\}}{\beta_k}\right)^{\beta_k}= \prod_{k=1}^m\left(\frac{v_k(\vec{p},\beta_k)}{\beta_k}\right)^{\beta_k}.
\end{equation}
Expressing each indirect utility through expenditure functions via formula~\eqref{eq_indirect_proof_of_TH1}, we end up with the following equality:
$\E_\agr(\vec{p})=\prod_{k=1}^m \left(\E_k(\vec{p})\right)^{\beta_k}$.
Taking the logarithm of both sides, we obtain identity~\eqref{eq:pr_APPENDIX}, completing the proof.
\end{proof}

\subsection{Proof of Proposition~\ref{prop_welfare_as_persuasion}}\label{app_persuasion}

\begin{proof}
	Recall that $\underline{W}$ and $\overline{W}$ denote the infimum and the supremum of $W=\sum_{k} b_k\cdot w(\succsim_k)$ over all finite populations of agents with an aggregate preference $\succsim_\agr$, total income $B$, and individual domain of preferences $\D$.
	Our goal is to show that 
	$$\underline{W}=B\cdot \vex_{\L_\D}\big[w\big]\big(\succsim_\agr\big), \qquad \overline{W}=B\cdot \cav_{\L_\D}\big[w\big]\big(\succsim_\agr\big),$$ 
	and for any $W'$ such that $\underline{W}<W'<\overline{W}$, there is a finite population with~$W=W'$.
	
	Let us prove the formula for $\overline{W}$. By  Theorem~\ref{th_representative_index}, populations of agents with preferences from $\D$ compatible with the aggregate preference $\succsim_\agr$
	correspond to all the different ways 
	to represent $\ln \E_{\agr}$ as a finite convex combination $\ln \E_{\agr}=\sum_k \beta_k \ln \E_k$ with $\ln \E_k$ from $\L_\D$. Hence,
	$$\overline{W}= \sup\left\{ B\cdot \sum_k \beta_k\cdot w(\succsim_k)\quad:\quad \ln \E_{\agr}=\sum_k \beta_k \ln \E_k,\quad \ln \E_k\in \L_\D,\quad \beta_k\geq 0, \quad \sum_k \beta_k=1\right\}.$$
	Comparing this formula to the definition of concavification~\eqref{eq_concavification}, we conclude that
	$$\overline{W}=B\cdot \cav_{\L_\D}\big[w\big]\big(\succsim_\agr\big).$$
	A mirror argument for  $\underline{W}$ is omitted.
	
	It remains to construct a finite population with $W=W'$ for $\underline{W}<W'<\overline{W}$. We already know that there are populations with $W$ arbitrarily close to $\underline{W}$ and $\overline{W}$. Hence, we can find $W^l$, $l=1,2$, such that  there are populations $(\succsim_k^l,b_k^l)_{k=1,\ldots,K^l}$ with  $W=W^l$ and $W^1<W'<W^2$. Express $W'$ as a convex combination $\alpha W^1+(1-\alpha)W^2$
	and consider the weighted union of the populations corresponding to $W^1$ and $W^2$:
	$$\big(\succsim_k^1,\ \alpha\cdot b_k^1\big)_{k=1,\ldots,K^1} \bigcup \big(\succsim_k^2,\ (1-\alpha)\cdot b_k^2\big)_{k=1,\ldots,K^2}.$$
	The constructed population has  $W=W'$.
\end{proof}

\subsection{Generalizations and proofs of Theorems~\ref{th_continuous_aggregation} and \ref{th_choquet}}\label{sec_th_continuous_aggregation_proof}

We first prove Theorem~\ref{th_choquet} and then formulate and prove a general result containing both Theorem~\ref{th_continuous_aggregation} and Theorem~\ref{th_choquet} as particular cases.

Recall that $\D^\inv$ is the completion of a domain $\D$, i.e., the minimal closed domain invariant with respect to aggregation and containing $\D$. The set of indecomposable preferences in $\D$ is denoted by $\D^\ind$ and is composed of all preferences $\succsim\in \D$ that cannot be represented as an aggregation of two distinct preferences $\succsim',\succsim''\in \D$.

For the reader's convenience, we repeat the statement of Theorem~\ref{th_choquet}.

\begin{theorem*}
If $\D$ is a closed domain such that $\D=\D^\inv$, then a preference $\succsim$ belongs to $\D$ 
if and only if there exists a Borel probability measure $\mu$ supported on $\D^\ind$ such that the expenditure function $\E=\E_\succsim$ can be represented as follows
\begin{equation}\label{eq_choquet_representation_of_indecomposable_appendix}
\ln \E(\vec{p})=\int_{\D^\ind} \ln \E_{\succsim'}(\vec{p})\, \dd \mu(\succsim')
\end{equation}
for any vector of prices $\vec{p}\in \R_{++}^n$.
\end{theorem*}
The requirement that $\mu$ is supported on $\D^\ind$ means that the complement of this set of preferences has $\mu$-measure zero. 
The topology on preferences and logarithmic expenditure functions is described in Appendix~\ref{app_topology}. The Borel structure is defined by this topology. 

We will need the Choquet theorem formulated below. Consider a (not necessarily convex) subset $X$  of a linear space. A point $x\in X$ is an extreme point of $X$ if it cannot be represented as $\alpha x'+(1-\alpha)x''$ with $\alpha\in (0,1)$ and distinct $x',x''\in X$. All extreme points of $X$ are denoted by~$X^\ext$.
\begin{theorem*}[Choquet's theorem; see \cite*{phelps2001lectures}, Section 3]\label{th_choquet_abstract}
If $X$ is a metrizable compact convex
subset of a locally convex space, then a point $x$ belongs to $X$ if and only if there is a Borel probability measure $\mu$ on $X$  supported on $X^\ext$ such that
\begin{equation}\label{eq_Choquet_hull}
x=\int_{X^\ext} x'\, \dd\mu(x').
\end{equation}
\end{theorem*}
In our application, $X$ will be a subset of the Banach space of continuous functions with the $\sup$-norm. A Banach space is a complete separable normed space. Each such space is locally convex and metrizable via the metric induced by the norm. 

The identity~\eqref{eq_Choquet_hull} is to be understood in the weak sense, i.e., for any continuous linear functional~$F$,
$$F[x]=\int_{X^\ext} F[x']\, \dd\mu(x').$$

\begin{proof}[Proof of Theorem~\ref{th_choquet}]
A homothetic preference $\succsim$ is represented by a family of equivalent logarithmic expenditure functions which differ by a constant.  Let $\L$ be the set of all classes of equivalent logarithmic expenditure functions corresponding to homothetic preferences. Denote by $\L_\D$ the subset of $\L$ corresponding to the domain $\D$. 
The set $\L_\D$ is closed and convex. Indeed, convexity follows from invariance of~$\D$ by Corollary~\ref{cor_convex_hull}, and closedness of $\D$ is inherited by $\L_\D$ as the topologies on preferences and logarithmic expenditure functions are aligned.

By Lemma~\ref{lemma_embedding}, the set $\L$ admits an affine isometric compact embedding $\T$ into the Banach space $\cC\left(\Delta_{n-1}\right)$ of continuous functions on the simplex $\Delta_{n-1}$ with the $\sup$-norm. 
Since $\L_\D$ is a closed convex subset of $\L$, 
the embedding $\T[\L_{\D}]$ is a compact convex  subset of $\cC\left(\Delta_{n-1}\right)$.

Applying the Choquet theorem to $X=\T[\L_{\D}]$, we conclude that a logarithmic expenditure function $\ln \E$ belongs to $\L_{{\D}}$ if and only if there is a measure $\mu$ supported on $X^\ext$ such that $x=\T[\ln \E]$ is given by the integral of the form~\eqref{eq_Choquet_hull}. 

As there is a natural bijection between ${\D}$ and $X$, we can assume that $\mu$ is a measure on~$\D$. By Theorem~\ref{th_representative_index}, a preference $\succsim'$ is indecomposable in $\D^\ind$ if and only if its logarithmic expenditure function $\ln \E_{\succsim'}$ cannot be represented as a convex combination of two non-equivalent expenditure functions from $\D$ (Corollary~\ref{cor_indecomposable_as_extreme}). Hence, $\succsim'$ belongs to $\D^\ind$ if and only if $\T[\ln \E_{\succsim'}]$ is in $X^\ext$. We obtain that $\succsim$ with an expenditure function $\E$ is contained in the completion $\D^\inv$ if and only if
\begin{equation}\label{eq_almost_there}
\T[\ln \E]=\int_{\D^\ind} \T[\ln \E_{\succsim'} ]\, \dd\mu(\succsim')
\end{equation}
for some $\mu$ supported on ${\D}^\ind$.
To get the desired pointwise identity~\eqref{eq_choquet_representation_of_indecomposable_appendix}, it remains to apply an appropriate linear functional on both sides.

Let $F_\vec{p}$ be the functional on $\cC\left(\Delta_{n-1}\right)$ evaluating a function at some $\vec{p}\in \Delta_{n-1}$. This functional is continuous, and the family of such functionals with $\vec{p}\in \Delta_{n-1}\cap \R_{++}^n$ separates points, i.e., if two functions are not equal, there is a functional taking different values on them. Hence,~\eqref{eq_almost_there} is equivalent to the following identity
$$F_\vec{p}\Big[\T[\ln \E]\Big]=\int_{\D^\ind} F_\vec{p}\Big[\T[\ln \E_{\succsim'} ]\Big]\, \dd\mu(\succsim')$$
for all $\vec{p}\in \Delta_{n-1}\cap \R_{++}^n$.
Plugging in the explicit form~\eqref{eq_embedding} of the embedding $\T$, we conclude that~$\succsim\in \D$ if and only if its logarithmic expenditure function $\succsim$ can be represented as follows 
$$\frac{\ln \E(\vec p)-\ln \E(\vec{e})}{\left(1+\max_i |\ln p_i|\right)^2}=\int_{\D^\ind} \frac{\ln \E_{\succsim'}(\vec p)-\ln \E_{\succsim'}(\vec{e})}{\left(1+\max_i |\ln p_i|\right)^2} \, \dd\mu(\succsim')
$$
for all $\vec{p}\in \Delta_{n-1}\cap \R_{++}^n$. Multiplying both sides by the denominator, we get
\begin{equation}\label{eq_almost_almost_there}
\ln \E(\vec p)=\const + \int_{\D^\ind}\ln \E_{\succsim'}(\vec p) \, \dd\mu(\succsim')
\end{equation}
for $\vec{p}\in \Delta_{n-1}\cap \R_{++}^n$. Since $\E(\alpha\cdot \vec{p})=\alpha\cdot \E(\vec{p})$, the identity extends to $\R_{++}^n$. As expenditure functions that differ by a constant correspond to the same preference, the constant in~\eqref{eq_almost_almost_there} can be absorbed by~$\ln \E$. This completes the proof of Theorem~\ref{th_choquet}.
\end{proof}

With the help of Theorem~\ref{th_choquet}, we can prove that, without any assumptions on the domain $\D$,
the completion $\D^\inv$  is obtained by continuous aggregation of preferences from $\overline{\D}^\ind$, i.e., of indecomposable preferences from the closure of $\D$. This result extends both  Theorem~\ref{th_choquet} and Theorem~\ref{th_continuous_aggregation}.
\begin{theorem}\label{th_general_appendix}
 For any domain $\D$, a preference $\succsim$ belongs to its completion $\D^\inv$ if and only if the expenditure function $\E$ of $\succsim$  admits the following representation
\begin{equation}\label{eq_continuous_aggregation_proposition_appendix}
    \ln \E(\vec{p})=\int_{\overline{\D}^\ind}
\ln \E_{\succsim'}(\vec{p})\,\dd\mu(\succsim'), \qquad \vec{p}\in \R_{++}^n,
\end{equation}
for some Borel probability measure $\mu$ supported on $\overline{\D}^\ind$.
\end{theorem}
 Theorem~\ref{th_choquet} corresponds to closed invariant domains $\D$ and Theorem~\ref{th_continuous_aggregation} is a corollary since $\overline{\D}^\ind$ is a subset of $\overline{\D}$.

Our proof of Theorem~\ref{th_general_appendix} relies on already proved Theorem~\ref{th_choquet} and on Milman's converse to
the Krein-Milman theorem. Recall that $\conv[Z]$ denotes the closed convex hull of a set $Z$.
\begin{proposition}[Milman; see \cite*{phelps2001lectures}, Proposition 1.5]\label{prop_milman}
If $X$ is a compact convex
subset of a locally convex space and $X=\conv[Z]$, then extreme points $X^\ext$ are contained in the closure~$\overline{Z}$.
\end{proposition}
From the definition of extreme points, it is immediate that if $Z\subset X$, then any extreme point of $X$ contained in $Z$ is an extreme point of $Z$. Hence,  the conclusion of Proposition~\ref{prop_milman} can be strengthened as $X^\ext \subset  \overline{Z}^\ext.$ 

Since the completion corresponds to taking closed convex hull (Corollary~\ref{cor_convex_hull}) and indecomposable preferences correspond to extreme points, Proposition~\ref{prop_milman} implies the following corollary.
\begin{corollary}\label{cor_indecomposable_in_completion}
For any preference domain $\D$, all indecomposable preferences of its completion are contained in indecomposable preferences of its closure, i.e., $(\D^\inv)^\ind\subset \overline{D}^\ind$.
\end{corollary}
With this corollary, Theorem~\ref{th_general_appendix} follows from  Theorem~\ref{th_choquet} almost immediately.
\begin{proof}[Proof of Theorem~\ref{th_general_appendix}]
Apply Theorem~\ref{th_choquet} to
the closed invariant domain $\D'=\D^\inv$. We get that $\succsim$ is in $\D^\inv$ if and only if there is a measure supported $(\D^\inv)^\ind$ such that
$$\ln \E(\vec{p})=\int_{(\D^\inv)^\ind}
\ln \E_{\succsim'}(\vec{p})\,\dd\mu(\succsim').$$
By Corollary~\ref{cor_indecomposable_in_completion}, $(\D^\inv)^\ind\subset \overline{D}^\ind$, which completes the proof.
\end{proof}

\subsection{Proof of Proposition~\ref{prop_linear_invariant_hull}}\label{app_proof_linear_hull}

\fed{TODO: Complete paying attention to closure and the support of $\mu$ and $\nu$, and also prove a version with expenditure shares}

Observe first that, as atomic measures are dense in all measures, taking the closure of the set of preferences described by \eqref{eq_price_index_linear_aggregate_appendix} boils down to allowing arbitrary probability measures. Hence, the completion of linear preferences is the set of all $\succsim$ admitting representation~\eqref{eq_price_index_linear_aggregate_appendix} with some probability \ed{measure $\mu$ on $\R_+^n\setminus \{0\}$.} This representation can be interpreted as the result of continuous aggregation of linear preferences where non-atomic populations are allowed.

Define $\varepsilon_i=\ln v_i$ and denote the distribution of $\boldsymbol{\varepsilon}$ by $\nu$. Setting $w_i = -\ln p_i$, we get: $$-\ln\left(\min_{i=1,\ldots, n}\frac{p_i}{v_{i}} \right)=\max_{i=1,\ldots, n} (w_i+\varepsilon_i).$$ Hence,~\eqref{eq_price_index_linear_aggregate_appendix} is equivalent to
$$-\ln \big(\E(e^{-w_1},\ldots,e^{-w_n}\big)=\int_{\R^n} \left(\max_{i=1,\ldots n} (w_i+\varepsilon_i)\right)\dd \nu(\boldsymbol{\varepsilon}).
$$
As the right-hand side has the form of the expected utility in  \eqref{eq: ARUM}, this proves Proposition~\ref{prop_linear_invariant_hull}. 

\subsection{Proof of \eqref{eq:E3} being an expenditure function}\label{app_exp}

It suffices to check that the function $E(\vec{p})$ defined by \eqref{eq:E3} is homogeneous, concave, and non-negative. Homogeneity and non-negativity are straightforward. The concavity of $E(\vec{p})$ follows from quasi-concavity since $E(\vec{p})$ is homogeneous \citep[Chapter 3, p.104]{boyd2004convex}. To prove the quasi-concavity of $ E(\vec{p})$, we show that $\ln E(\vec{p})$ is concave.  To verify that, we compute the quadratic form of the Hessian of $\ln  E(\vec{p})$ on a vector $\vec{y}\in \R^3\setminus\{0\}$. The Hessian is negative definite:
$$\frac{\alpha}{3}\left(\frac{y_1^2}{p_1^2}+\frac{y_2^2}{p_2^2}+\frac{y_3^2}{p_3^2}\right) + (\alpha-1)\left(\frac{y_1+y_2+y_3}{p_1+p_2+p_3}\right)^2<-\left(1-\frac{2}{3}\alpha\right)\max_{i=1,2,3}\frac{y_i^2}{p_i^2}<0,$$
which implies concavity of $\ln E(\vec{p})$. Thus, \eqref{eq:E3} is an expenditure function.

\subsection{Proof of Proposition~\ref{prop_approx_algorithm}}\label{app_approximate}

The proof relies on the following lemma showing that if preferences in two populations have expenditure shares that are close, then  $\varepsilon$-equilibrium price vectors are close as well.
\begin{lemma}\label{lm_approximate_CE}
Consider two populations of $m$ consumers with the same budgets $b_1,\ldots, b_m$ but different preferences over $n$ goods:  $\succsim_1,\ldots, \succsim_m$ in the first population and $\succsim_1',\ldots, \succsim_m'$ in the second one. Assume that the expenditure shares $s_{k,i}(\vec{p})$ and  $s'_{k,i}(\vec{p})$ differ by at most some $\delta>0$ for any consumer $k$, good $i$,  and price~$\vec{p}$.
Then, any $\varepsilon$-equilibrium price vector for one population is an $(\varepsilon+n\delta)$-equilibrium price vector for the other. 
\end{lemma}
\begin{proof}[Proof of Lemma~\ref{lm_approximate_CE}]
The demand of an agent $k$ for a good $i$ can be expressed through expenditure shares as follows:
\begin{equation}\label{eq_demand_through_shares}
D_{k,i}(\vec{p},b_k)=s_{k,i}(\vec{p})\cdot \frac{b_k}{p_i}.
\end{equation}
Let $\vec{x}_1+\ldots+\vec{x}_m$ and $\vec{x'}_1+\ldots+\vec{x'}_m$ be market demands of the two populations from the statement of the lemma at some vector of prices $\vec{p}$. By~\eqref{eq_demand_through_shares} and the assumption that expenditure shares differ by at most $\delta$,
$$\sum_{i=1}^n p_i\cdot \left|\sum_{k=1}^m x_{k,i}-\sum_{k=1}^m x'_{k,i}\right|\leq \sum_{i=1}^n\sum_{k=1}^m b_k \cdot \max_{i,k}\left|s_{k,i}(\vec{p})- s'_{k,i}(\vec{p})\right|\leq n\cdot B\cdot\delta.$$
Hence, if $\vec{p}$ is an $\varepsilon$-equilibrium price vector for $\succsim_1',\ldots, \succsim_m'$, it is an $(\varepsilon+n\delta)$-equilibrium price vector for $\succsim_1,\ldots, \succsim_m$.
\end{proof}

To prove the proposition, it remains to show that any preference $\succsim$ over two substitutes can be approximated by the aggregate preference $\succsim'=\succsim_\agr$ of an auxiliary population with linear preferences so that the expenditure shares differ by at most~$\varepsilon$ at any vector of prices and the number of auxiliary agents is of the order of $1/\varepsilon$.

\begin{proof}[Proof of Proposition~\ref{prop_approx_algorithm}] Since $s_{\succsim,1}+s_{\succsim,2}=1$ and expenditure shares depend on the ratio of prices only, it is enough to ensure that $\big|s_{\succsim,1}(p_1,1)-s_{\agr,1}(p_1,1)\big|\leq \varepsilon$ for any $p_1\in \R_{++}$. As $\succsim$ exhibits substitutability,
$s_{\succsim,1}(\,\cdot\,,1)$ is a non-increasing function  with values in $[0,1]$. For any such function $f$, 
there is a piecewise-constant function $f_\varepsilon$ with at most $1/\varepsilon+1$ jumps such that the two functions differ by at most $\varepsilon$; indeed, one can take $f_\varepsilon=\varepsilon\cdot\left[f/\varepsilon\right]$, where $[t]$ denotes the integer part of a real number $t$. By  Corollary~\ref{cor_MRS_reconstruction}, any piecewise-constant non-decreasing function with values in $[0,1]$ is equal to $s_{\agr,1}(\,\cdot\,,1)$ for a population of linear consumers with the number of consumers equal to the number of jumps, marginal rates of substitution given by positions of the jumps, and budgets determined by jumps' magnitude. We conclude that $1/\varepsilon +1$ linear consumers are enough to approximate expenditure shares of any preference exhibiting substitutability with precision $\varepsilon$.  Combined with Lemma~\ref{lm_approximate_CE}, this observation completes the proof.
\end{proof}

Note that constructing the approximation in a computationally efficient way requires solving the equation $s_{\succsim,1}(p_1,1)=\varepsilon\cdot l$ multiple times for various $\succsim$ and $l$. Provided that there is an oracle computing expenditure shares, one can use binary search for this task.

\subsection{Indecomposability in the full domain and proof of Proposition~\ref{prop_extreme_points_full_domain}}\label{app_proof_indecomposable_full_domain}

We will need the following lemma.
\begin{lemma}\label{lm_Jensen}
	Consider a function $h(\vec{t})=t_1^\alpha\cdot  t_2^{1-\alpha}$ where $\alpha\in (0,1)$ and $\vec{t}\in \R_{++}^2$. If $\vec{t}\ne \const \cdot \vec{t}'$, then 
	$$h\big(\lambda \vec{t}+ (1-\lambda )\vec{t}'\big)>\lambda h(\vec{t})+(1-\lambda)h(\vec{t}')$$
	for any $\lambda\in (0,1)$.
\end{lemma}	
\begin{proof}
	The result follows from strict concavity of $g(\lambda)=h\big(\lambda \vec{t}+ (1-\lambda )\vec{t}'\big)$. To demonstrate strict concavity, it is enough to show that the second derivative $g''(\lambda)<0$. After a linear change of the variable, this requirement boils down to the negativity of the second derivative of $\gamma^\alpha (1+\gamma)^{1-\alpha}$ with respect to $\gamma$.  We omit the elementary computation.
\end{proof}	
Using Lemma \ref{lm_Jensen}, we prove Proposition~\ref{prop_extreme_points_full_domain}.
\begin{proof}[Proof of Proposition~\ref{prop_extreme_points_full_domain}] 
A utility function $u$ is of the form~\eqref{eq_extreme_full_domain} if and only if the corresponding expenditure function is also piecewise linear:
\begin{equation}\label{eq_piecewise_price_index_appendix}
\E=\min_{c\in C} \left(\sum_{j=1}^n c_j\cdot p_j\right), 
\end{equation}
where $C\subset \R_+^n$ is finite or countable. 
We need to show that preferences with such expenditure functions are indecomposable.
Towards a contradiction, assume that $\E$ of the form~\eqref{eq_piecewise_price_index_appendix} can  be represented as $$
\ln \E = \alpha \ln \E_1 + (1-\alpha) \ln \E_2, 
$$
where $\E_1$ and $\E_2$ are expenditure functions representing distinct homothetic preferences and $\alpha\in (0,1)$. Hence, $\E_1$ and $\E_2$ are not proportional to each other, i.e., the ratio $\E_1/\E_2\ne \const$. By continuity of expenditure functions, this means that there is a linearity region of $\E$ where $\E_1/\E_2\ne \const$. Therefore, we can find $\vec{p},\vec{p}'\in \R_{++}^n$  from the same linearity region of $\E$ such that 
\begin{equation}\label{eq_decomp_appendix}
\frac{\E_1(\vec{p})}{\E_2(\vec{p})} \neq \frac{\E_1(\vec{p}')}{\E_2(\vec{p}')}.
\end{equation}
By homogeneity of expenditure functions, we can assume that $\vec{p}$ and $\vec{p}'$ are normalized so that $\E(\vec{p})=\E(\vec{p}')=1$. Since $\vec{p}$ and $\vec{p}'$ belong to the same linearity region, the value of $\E$ at the mid-point $\vec{p}''=(\vec{p}+\vec{p}')/2$ is also equal to $1$. Therefore,
$$1= \E\left(\vec{p}''\right)= \E_1(\vec{p}'')^\alpha \E_2(\vec{p}'')^{1-\alpha}\geq \left(\frac{1}{2}\E_1(\vec{p})+\frac{1}{2}\E_1(\vec{p}')\right)^\alpha \left(\frac{1}{2}\E_2(\vec{p})+\frac{1}{2}\E_2(\vec{p}')\right)^{1-\alpha},$$
where we used concavity of $\E_1$ and $\E_2$. By Lemma~\ref{lm_Jensen}, the right-hand side admits the following lower bound:
$$\left(\frac{1}{2}\E_1(\vec{p})+\frac{1}{2}\E_1(\vec{p}')\right)^\alpha \left(\frac{1}{2}\E_2(\vec{p})+\frac{1}{2}\E_2(\vec{p}')\right)^{1-\alpha}> \frac{1}{2}\E_1(\vec{p})^\alpha \E_2(\vec{p})^{1-\alpha}+\frac{1}{2}\E_1(\vec{p}')^\alpha \E_2(\vec{p}')^{1-\alpha}.$$
The right-hand side can be rewritten as
$$\frac{1}{2}\E_1(\vec{p})^\alpha \E_2(\vec{p})^{1-\alpha}+\frac{1}{2}\E_1(\vec{p}')^\alpha \E_2(\vec{p}')^{1-\alpha}=\frac{1}{2}\E(\vec{p})+\frac{1}{2}\E(\vec{p}')=1.$$
We end up with a contradictory inequality $1>1$. Therefore, $\E$ cannot be represented as a convex combination~\eqref{eq_decomp_appendix}, and we conclude that the corresponding preference is indecomposable.
\end{proof}

Let us explore whether there are other indecomposable preferences in the full domain. For simplicity, we focus on the case of $n=2$ goods. As opposed to
preferences with piecewise linear $u$ and $\E$ considered in Proposition~\ref{prop_extreme_points_full_domain}, we examine preferences with expenditure functions $\E$ that are strictly concave in a neighborhood of a certain point.

We say that a function of one variable $h=h(t)$ is strictly concave in the neighborhood of  $t=t_0$ if there is $\varepsilon>0$ and $\delta>0$ such that the second derivative of $h''(t)<-\delta$ for almost all $t$ in the $\varepsilon$-neighborhood $[t_0-\varepsilon, t_0+\varepsilon]$ of $t_0$. We recall that the second derivative exists almost everywhere for any concave function by Alexandrov's theorem. \fed{Reference}
\begin{proposition}\label{prop_strictly_concave_decomposable}
	Consider a preference $\succsim$ over two goods with expenditure function $\E$.
	If there is a point $\vec{p}_0\in \R_{++}^2$ and a direction $\vec{r}\in \R^2\setminus\{0\}$ such that $g(t)=\E(\vec{p}_0+t\cdot \vec{r})$ is strictly concave in the neighborhood of  $t=0$, then $\succsim$ is not indecomposable.
\end{proposition}	
\begin{proof}
	Since $\E(\alpha\cdot \vec{p})=\alpha\cdot  \E(\vec{p})$, the values of the expenditure function on the line $p_2=1$ determine its values everywhere by $\E(p_1,p_2)=p_2\cdot \E(p_1/p_2,\, 1)$. Accordingly, the condition from the statement is equivalent to the existence of $t_0$ such that $g(t)=\E(t,1)$ is strictly concave in the neighborhood of $t_0$. \fed{Add details}
	
	Let us show that if $g(t)=\E(t,1)$ is strictly concave in the neighborhood of $t_0$, then the preference $\succsim$ is an aggregation of some distinct $\succsim_1$ and $\succsim_2$. By strict concavity, $g''<-\delta$ on $[t_0-\varepsilon, t_0+\varepsilon]$ for some $\varepsilon,\delta>0$.	Let $\varphi(z)$ be a smooth function on $\R$ not equal to zero identically and vanishing outside of the interval $[-1,1]$ together with all its derivatives. For example, one can take 
	$$\varphi(z)=\exp\left(-\frac{1}{1-z^2}\right)$$
	for $z\in(-1,1)$ and zero outside. Define 
	$$
	g_1(t)= \left(1+\gamma\cdot  \varphi\big(\varepsilon(t-t_0)\big)\right)g(t)\qquad\text{and}\qquad g_2(t)= \frac{1}{1+\gamma\cdot  \varphi\big(\varepsilon(t-t_0)\big)} g(t)
	$$
	for some constant $\gamma>0$. Note that $g_1=g_2=g$ outside the $\varepsilon$-neighborhood of $t_0$. The second derivatives of $g_1$ and $g_2$ continuously depends on  $\gamma$ and, for $\gamma=0$, the derivatives are bounded from above by $-\delta$ in the $\varepsilon$-neighborhood  of $t_0$. Hence, for small enough $\gamma>0$, the second derivative is non-positive, i.e., both $g_1$ and $g_2$ are concave.
	
	Define $\E_1(\vec{p})=p_2\cdot g_1(p_1/p_2)$ and $\E_2(\vec{p})=p_2\cdot g_2(p_1/p_2)$. These are non-negative homogeneous concave functions that are not proportional to each other. Hence, $\E_1$ and $\E_2$ are expenditure functions of some distinct preferences $\succsim_1$ and $\succsim_2$. By the construction,
	$$\ln \E= \frac{1}{2}\ln \E_1+\frac{1}{2}\ln \E_2$$
	and thus $\succsim$ is the aggregate preference for a pair of consumers with preferences  $\succsim_1$ and $\succsim_2$ and equal incomes.
\end{proof}	

It may seem that Propositions~\ref{prop_extreme_points_full_domain} and~\ref{prop_strictly_concave_decomposable} cover all possible preferences in the case of two goods: intuitively, an expenditure function $\E$ is either piecewise linear or there is a point in the neighborhood of which $\E$ is strictly concave. However, there are pathological examples not captured by the two propositions. 

Any concave function $f$ on $\R_+$  can be represented as
$$f(t)=f(0)-\int_0^t \left(\int_0^s \dd\nu(q)\right)\dd s$$
for some positive measure $\nu$ on $\R_+$. This $\nu$  is the uniquely defined distributional second derivative of~$f$. Abusing the notation, we will write $\nu=f''$. Note that the classical pointwise second derivative (where exists) equals the density of the absolutely continuous component of $\nu$. 

Propositions~\ref{prop_extreme_points_full_domain} and~\ref{prop_strictly_concave_decomposable} address the cases where the second derivative of $\E(t,1)$ is either an atomic measure with a nowhere dense set of atoms or has an absolutely continuous component with a strictly negative density on a certain small interval. 

Recall that $\nu$ is called singular if there is a set of zero Lebesgue measure such that its complement has $\nu$-measure zero. For example, atomic measures with discrete sets of atoms are singular, but there are other singular measures such as non-atomic measures supported on a Cantor set or atomic measures with an everywhere dense set of atoms.  
\begin{proposition}\label{prop_pathology}
	If $\succsim$ is a preference over two goods with an expenditure function $\E$ such that the second distributional derivative of $g(t)=\E(t,1)$ is singular, then $\succsim$ is indecomposable in the full domain.
\end{proposition}
We see that the set of indecomposable preferences is broader than suggested by Proposition~\ref{prop_extreme_points_full_domain}. Note that, in the particular case of two goods, Proposition~\ref{prop_extreme_points_full_domain}  is a direct corollary of Proposition~\ref{prop_pathology}.
\begin{proof}
It is enough to show that if $\succsim$ is an aggregation of two distinct preferences $\succsim_1$ and $\succsim_2$, then the second distributional derivative of $g$ has a non-zero absolutely continuous component. In other words, we need to show that the classical derivative $g''\ne 0$ on a set of positive Lebesgue measure. 

Let  $\E_1$ and $\E_2$ be expenditure functions of $\succsim_1$ and $\succsim_2$. Since the preferences are distinct, $\E_1\ne \const \cdot \E_2$.  By the assumption,  $\E=\E_1^\alpha\cdot \E_2^{1-\alpha}$ with some $\alpha\in(0,1)$. Without loss of generality, we can assume that $\alpha=1/2$. Indeed, if $\alpha\ne\frac{1}{2}$, one can define new expenditure functions $\E_1'= \E_1^{\alpha-\varepsilon}\cdot \E_2^{1-\alpha+\varepsilon}$ and $\E_2'= \E_1^{\alpha+\varepsilon}\cdot \E_2^{1-\alpha-\varepsilon}$ for some $\varepsilon<\min\{\alpha,\,  1-\alpha\}$ so that $\E=\sqrt{\E_1'\cdot \E_2'}$.

Hence, $g=\sqrt{g_1\cdot g_2}$ where $g_1(t)=\E_1(t,1)$ and $g_2(t)=\E_2(t,1)$ are non-negative concave functions not proportional to each other. Computing the classical second derivative of $g$, we obtain
$$g''=\frac{g_1''\cdot g_2+g_2'' \cdot g_1}{2\sqrt{g_1\cdot g_2}}-\frac{\big(g_1'\cdot g_2-g_2'\cdot g_1\big )^2 }{4(g_1\cdot g_2)^{3/2}}.$$
Both terms are non-positive.  The numerator in the second term can be rewritten as follows
$$\big(g_1'\cdot g_2-g_2'\cdot g_1\big )^2=\left((g_2)^2\cdot\left(\frac{g_1}{g_2}\right)' \right)^2$$
Since the ratio $g_1/g_2$ is non-constant, its derivative $(g_1/g_2)'$ is non-zero on a set of positive measure. Thus the distributional
derivative $g''$ contains a non-zero absolutely continuous component. We conclude that preferences such that $g''$ has no absolutely continuous component are indecomposable.
\end{proof}	

The approach from Proposition~\ref{prop_pathology} extends to $n>2$ goods. It can be used to show that if the distributional second derivative of 
$$g(t)=\E(p_1,\ldots,p_{i-1},\, t,\, p_{i+1},\ldots, p_n)$$
 is singular for any $i=1,\ldots,n $ and any fixed $$p_{-i}=(p_1,\ldots,p_{i-1}, p_{i+1},\ldots, p_n)\in \R_{++}^{n-1},$$ then the corresponding preference is indecomposable. This result extends Proposition~\ref{prop_extreme_points_full_domain}.

\subsection{Proof of Proposition~\ref{prop_indecomposable_for_C}}
\label{app_proof_indecomposable_complement}

\begin{proof}
Consider a  Leontief preference $\succsim$ over Cobb-Douglas composite goods. It corresponds to a utility function
$$    
u(\vec{x})=\min_{\vec{a}\in A} \left\{a_0\cdot\prod_{i=1}^n x_i^{a_i}\right\},
$$
where $A$ is finite or countably infinite subset of $\R_{++}\times \Delta_{n-1}$. Assume that $\succsim$ is non-trivial, i.e., there are $\vec{a},\vec{a}'\in A$ such that $(a_1,\ldots, a_n)\ne (a_1',\ldots, a_n')$. 
The intersection of convex sets corresponds to convexification of the maximum of their support functions \cite[Theorem~7.56]{charalambos2013infinite}. Since expenditure functions are support functions of upper contour sets up to a sign, the expenditure function corresponding to $\succsim$ takes the following form
$$\E(\vec{p})=\cav\left[\max_{\vec{a}\in A} \left\{\frac{1}{a_0}\cdot\prod_{i=1}^n p_i^{a_i}\right\}\right],$$
where $\cav$ denotes concavification. 

Let us focus on the case of $n=2$ goods. In this case, $\E$ has a particularly simple structure. The positive orthant $\R_{++}^2$ is partitioned into a finite or countably infinite number of cones of two types: (I) cones where $\E$ is linear (II) cones where $\E$ coincides with the Cobb-Douglas expenditure function ${1}/{a_0}\cdot\prod_{i=1}^n p_i^{a_i}$ for some~$\vec{a}\in A$.
The cones of type (I) and (II)  interlace, and derivatives of $\E$ change continuously. Note that there must be at least one cone of type (I) as, otherwise, $\E$ would be an expenditure function of standard Cobb-Douglas preferences, which is ruled out by the non-triviality assumption.

Let us show that such a preference $\succsim$ over two goods is indecomposable. Towards a contradiction, assume that
\begin{equation}\label{eq_proof_of_extreme_complements_decomposition}
\ln \E=\alpha\cdot  \ln \E_1+(1-\alpha)\ln \E_2,\qquad \alpha\in (0,1),
\end{equation}
where $\E_1$ and $\E_2$ correspond to two distinct preferences $\succsim_1$ and $\succsim_2$ exhibiting complementarity. As in the proof of Proposition~\ref{prop_extreme_points_full_domain}, one shows 
that in each linearity region of $\E$, the expenditure functions  $\E_1$ and $\E_2 $ are proportional to each other. In other words, $\E_1=\const \cdot \E_2$ in each cone of type~(I),  where the constant can depend on the cone.

Recall that, by~\eqref{eq_budget_shares_as_elasticities}, the partial derivative of a logarithmic expenditure function with respect to $\ln p_i$ is the expenditure share of  good~$i$. Denote the expenditure shares for $\succsim$, $\succsim_1$, and $\succsim_2$ by $s_i$, $s_{1,i}$, and $s_{2,i}$, respectively. Since $\E_1$ and $\E_2$ are proportional in cones of type (I), we obtain that $s_i=s_{1,i}=s_{2,i}$ there. 

Consider cones of type (II). In these cones, $s_i$ is constant since expenditure shares are constant for Cobb-Douglas preferences. Taking the partial derivative on  both sides of~\eqref{eq_proof_of_extreme_complements_decomposition}, we
get
$$s_i(p_1,p_2)=\alpha\cdot s_{1,i}(p_1,p_2)+(1-\alpha)s_{2,i}(p_1,p_2).$$
The expenditure shares depend only on the ratio of prices. By complementarity, $s_{1,i}$ and $s_{2,i}$ must be non-increasing functions of the ratio $p_{3-i}/p_i$. 
Note that if $s_i$ is constant and one of $s_{1,i}$ or $s_{2,i}$ increases, the other must decrease, violating the monotonicity requirement. Hence, in all the cones of type (II), $s_i$, $s_{1,i}$ and $s_{2,i}$ are constant. These constants must be all equal. Indeed, suppose that $s_i=c$, $s_{1,i}=c_1$ and $s_{2,i}=c_2$  in some cone of type (II) with $c_1\ne c_2$. Since $s_i$ is continuous and coincides with $s_{1,i}$ and $s_{2,i}$ in the neighboring cone of type (I), $s_{1,i}$ and $s_{2,i}$ are discontinuous on the boundary between the two cones and at least one of these discontinuities necessarily violates the monotonicity requirement. Thus in cones of type (II), $s_i=s_{1,i}=s_{2,i}$.

We conclude that $s_i=s_{1,i}=s_{2,i}$ everywhere, i.e., partial derivatives of $\ln \E$, $\ln \E_1$, and  $\ln \E_2$ coincide. Hence, $\ln \E_1=\ln \E_2+\const$, i.e., $\E_1$ and   $\E_2$ are proportional. Thus $\succsim_1=\succsim_2$, which contradicts the assumption that the two preferences are distinct. We conclude that, for $n=2$, any non-trivial Leontief preference over Cobb-Douglas composite goods is indecomposable.
\end{proof}

\subsection{Characterization of expenditure shares  for two goods}\label{app_feasible shares}

For a preference $\succsim$ over $n=2$ goods, consider the expenditure shares $s_1(\vec{p})$ and $s_2(\vec{p})$ of these goods. We aim to characterize functions that one can get as expenditure shares.
Since $s_1+s_2=1$, we can focus on the expenditure share $s_1$ of the first good. As $s_1(\alpha\cdot \vec{p})=s_1(\vec{p})$, the expenditure share can be seen as a function of one variable $z=p_1/p_2$. Our goal is to characterize functions $h=h(z)$ such that $h=s_1$ for some homothetic preference $\succsim$.

\begin{lemma}\label{lm_budget_share_two_goods}
A function $h:\ \R_{++}\to \R$ is
 the expenditure share of the first good associated with some homothetic preference $\succsim$ over two goods (i.e., $h(p_1/p_2)=s_1(\vec{p})$, $\vec{p}\in \R_{++}^2$) if and only if
\begin{equation}\label{eq_s_and_Q}
h(z)=\frac{z}{z+Q(z)}, 
\end{equation}
where $Q:\ \R_{++}\to \R_+\cup \{+\infty\}$ is a non-negative non-decreasing function. 
\end{lemma}
Leontief and linear preferences correspond, respectively, to the extreme cases of $Q(z)=\const$ and an infinite step function $$Q(z)=\begin{cases}
0, & z\leq \alpha,\\
+\infty & z> \alpha.
\end{cases}$$

Note that for any non-increasing function $h$ with values in $[0,1]$, the function $$Q(z)=\frac{1}{h(z)}-1$$ satisfies the requirement of Lemma~\ref{lm_budget_share_two_goods}.
\begin{corollary}\label{cor_any_decreasing_goes}
Any non-increasing function $h$ with values in $[0,1]$ is the expenditure share of the first good for some preference $\succsim$ exhibiting substitutability among the two goods.
\end{corollary}

\begin{proof}[Proof of Lemma~\ref{lm_budget_share_two_goods}]
Consider a homothetic preference $\succsim$ for $n=2$ goods. The expenditure share $s_1(\vec{p})$  of the first good satisfies $s_1(\alpha \vec{p})= s(\vec{p})$ for any $\alpha>0$, hence, $s_1(\vec{p})=s_1(z,1)$, where $z=p_1/p_2$. 

Let us show that, for any $\succsim$, the function $h(z)=s_1(z,1)$ admits the representation~\eqref{eq_s_and_Q}. In other words, we need to show that
$$Q(z)=\frac{z}{s_1(z,1)}-z$$
is non-negative and non-decreasing. Expressing the expenditure share through the logarithmic expenditure function by~\eqref{eq_budget_shares_as_elasticities} and denoting $\pi(z)=\E(z,1)$, we get
\begin{equation}\label{eq_Q_and_P}
Q(z)=\frac{\pi(z)}{\pi'(z)}-z.
\end{equation}
Note that $\pi$ is a non-negative non-decreasing concave function of $z$. Hence, $Q$ is non-negative as well. To show that $Q$ is non-decreasing, let
 us differentiate both sides of~\eqref{eq_Q_and_P}. We get
 \begin{equation}\label{eq_Q_prime}
Q'(z)=-\frac{\pi(z)}{(\pi'(z))^2} \cdot \pi''(z).
\end{equation}
We see that $Q'\geq 0$ and so $Q$ is non-decreasing.\footnote{If $\pi'$ is not differentiable, the identity~\eqref{eq_Q_prime} is to be understood in the sense of distributional derivatives: $\pi''$ is a non-negative measure, and the right-hand side is a measure having density $-{\pi(z)}/{(\pi'(z))^2}$ with respect to $\pi''$.} We conclude that $h=s_1(z,1)$ admits the representation~\eqref{eq_s_and_Q}.

To prove the converse, consider $h$ of the form~\eqref{eq_s_and_Q} with non-negative non-decreasing $Q$ and construct the corresponding preference. The identity~\eqref{eq_Q_and_P} suggests how to define the expenditure function. First, we define $\pi$ by
$$\frac{\pi'(z)}{\pi(z)}=\frac{1}{z+Q(z)}.$$
 Integrating this identity, we obtain
$$\pi(z)=\exp\left(\int_1^{z} \frac{1}{w+Q(w)}\,\dd w\right).$$
By the construction, $\pi$ is non-negative. Since the identity~\eqref{eq_Q_prime} is hardwired in the definition of $\pi$ and the function $Q$ is non-decreasing, we conclude that $\pi''$ is non-positive. Hence, $\pi$ is concave. Define $\E$ by
$$\E(p_1,p_2)=p_2\cdot \pi(p_1/p_2).$$
This function is homogeneous, non-negative, and concave. Thus $\E$ is an expenditure function corresponding to some homothetic preference. By the construction,  $s_1(z,1)=h(z)$
completing the proof.
\end{proof}

\section{Robust welfare analysis and Bayesian
 persuasion: a formal connection}\label{app_connection_to_persuasion}

Let us discuss the formal connection between robust welfare analysis as described in Section~\ref{sec_robust} and Bayesian persuasion,  a benchmark model for a situation where an informed party decides what information to reveal to an uninformed one and has an objective depending on induced beliefs \citep*{kamenica2011bayesian}. 
Mathematically, persuasion boils down to solving the following optimization problem. We are given a set of states $\Omega$, a prior belief $\mu \in \Delta(\Omega)$ where $\Delta(\Omega)$ denotes the simplex of probability distributions over $\Omega$, and an objective function $g$ defined on $\Delta(\Omega)$. The goal is to maximize 
$$\sum_k \beta_k \cdot g({\mu_k})$$
over all possible ways to represent the prior ${\mu}$ as a finite convex combination $\mu =\sum_k \beta_k \cdot \mu _k $ with $\mu _k\in \Delta(\Omega)$. 
The persuasion problem has an elegant geometric solution: the optimal value of the persuasion problem is $\cav_{\Delta(\Omega)}[g](\mu)$   \citep*{aumann1995repeated,kamenica2011bayesian}.

\smallskip
The similarity between persuasion and finding the maximal value of a functional $W$  of the form~\eqref{eq_welfare_additive}  compatible with the observed aggregate behavior of a population must be apparent: the set $\L_\D$ of logarithmic expenditure functions plays the role of $\Delta(\Omega)$, the logarithmic expenditure function of the aggregate preference corresponds to the prior~$\mu$, and $w=w(\succsim)$ considered as a function on  $\L_\D$ is an analog of informed party's  objective $g$. 

 The difference is that, in Bayesian persuasion, the concavification takes place over a simplex $\Delta(\Omega)$ while the set $\L_\D$ of logarithmic expenditure functions is not necessarily a simplex. Recall that a convex set is a simplex if each point can be represented as a weighted average of the extreme points in a unique way; we call $\D$ a simplex domain if the corresponding set of logarithmic expenditure functions $\L_\D$ is a simplex; see Section~\ref{sec_S_indecomposable}.
\fed{Discuss if there is a reduction to Bayesian persuasion for arbitrary $\D$}

For simplex domains, finding the maximal value of $W$ compatible with the observed aggregate behavior is equivalent to a persuasion problem. An elementary example where this equivalence holds is the domain $\D$ of Cobb-Douglas preferences; see Example~\ref{ex_direction_welfare_revisited}.

Let us also illustrate the equivalence for the domain $\D_{S}$ of all preferences exhibiting substitutability over $n=2$  goods as in Example~\ref{ex_CES_welfare}. Any preference $\succsim\in \D_{S}$ can be represented as an aggregation of linear preferences, and this representation is unique (Corollary~\ref{cor_MRS_reconstruction}). Linear preferences over two goods form a one-parametric family with the marginal rate of substitution $\mrs=v_1/v_2\in \R_+\cup\{+\infty\}$ as a parameter.
A preference~$\succsim\in \D_{S}$ defines a unique distribution $\mu$ of $\mrs$ by formula~\eqref{eq_mu_and_budget_shares}: the cumulative distribution function is equal to $1-s_1(\,\cdot\,,1)$. The functional $w(\succsim)$ can equivalently be thought of as a function of $\mu$. Thus the welfare maximization problem takes the following form. We are given $\mu_\agr \in\Delta(\R_+\cup\{+\infty\})$ and a functional $w=w(\mu)$. The goal is to maximize  
$$B\cdot \sum_k \beta_k \cdot w({\mu _k})$$
over all possible ways to represent the prior ${\mu}$ as a finite convex combination $\mu =\sum_k \beta_k \cdot \mu _k $ with $\mu _k\in \Delta(\R_+\cup\{+\infty\})$.

We conclude that, for two substitutes, finding the maximal welfare compatible with the observed aggregate behavior is equivalent to persuasion with the set of states $\Omega=\R_+\cup\{+\infty\}$, the cumulative distribution function of the prior $\mu$ equal to $1-s_1(\,\cdot\,,1)$, and the objective $w=w(\mu)$. 
Persuasion problems with a continual state space are not easy to solve analytically unless some further assumptions are made. For example, if $\mu$ is finitely supported, i.e., there is a finite number of preference ``types'' in the population, then the support can be taken as the new set of states reducing the problem to the well-understood case of persuasion with a finite number of states. If $\mu$ 
has infinite support, tractability can be gained by imposing assumptions on the objective $w$. Tractable cases include convex $w$ as in Examples~\ref{ex_CES_welfare} and~\ref{ex_CES_welfare_approx}, or $w$
depending on $\mu$ through the mean value of a given function $\varphi$, i.e., $w=w\left(\int \varphi(z)\dd\mu(z)\right)$  as in \citep*{dworczak2019simple, arieli2019optimal, kleiner2021extreme}, or $w$ depending on $\mu$ through a quantile of $\varphi$ as in \citep*{yang2022distributions}.

\end{document}